\documentclass{article}

\usepackage{amsmath,amssymb,amsfonts,amsthm}
\usepackage{bbold,stmaryrd,graphicx}

\newcommand{\rMap}[1]{\xrightarrow{#1}}
\newcommand{\lMap}[1]{\xleftarrow{#1}}
\newcommand{\rMon}{\rightarrowtail}

\theoremstyle{plain}

\newtheorem{theorem}{Theorem}[section]
\newtheorem{lemma}[theorem]{Lemma}
\newtheorem{proposition}[theorem]{Proposition}
\newtheorem{corollary}[theorem]{Corollary}

\theoremstyle{definition}
\newtheorem{definition}[theorem]{Definition}
\newtheorem{example}[theorem]{Example}
\newtheorem{fact}[theorem]{Fact}

\newcommand{\Sem}[1]{\llbracket #1 \rrbracket}

\newcommand{\Iff}{\Leftrightarrow}

\newcommand{\Cond}{\mathop{|}}
\newcommand{\sIndep}{\mathop{\perp}}

\newcommand{\Indep}{\mathop{\bot\!\!\!\!\!\;\bot}}

\newcommand{\Supp}{\text{supp}}

\newcommand{\Forces}{\Vdash}
\newcommand{\Syn}[1]{\mathsf{#1}}

\newcommand{\Equiv}{\sim}

\newcommand{\EqAE}{=_{\text{a.s.}}}

\newcommand{\Epi}{\twoheadrightarrow}

\newcommand{\sCat}[1]{\mathbb{#1}}
\newcommand{\Op}[1]{#1^{\mathsf{op}}}

\newcommand{\Id}{\mathsf{id}}

\newcommand{\Set}{\textbf{Set}}
\newcommand{\Sur}{\mathbb{Sur}}

\newcommand{\sBP}{\mathbb{SBP}}

\newcommand{\R}{\mathbb{R}}

\newcommand{\I}{\mathbb{I}}
\newcommand{\Iop}{\Op{\I}}

\newcommand{\Sh}[1]{\underline{#1}}

\newcommand{\Psh}{\mathrm{Psh}}
\newcommand{\ShAt}{\mathrm{Sh}_{\mathrm{at}}}

\newcommand{\Yon}{\mathbf{y}}
\newcommand{\AS}{\mathbf{a}}

\newcommand{\NEls}{\mathsf{\underline{NV}}}

\newcommand{\SRV}{\Sh{\mathsf{RV}}}

\newcommand{\Sort}{\mathsf{Sort}}
\newcommand{\Arity}{\mathrm{arity}}
\newcommand{\FV}{\mathrm{FV}}

\newcommand{\eqdist}{\stackrel{d}{=}}
\newcommand{\One}{\mathbb{1}}

\def\BibTeX{{\rm B\kern-.05em{\sc i\kern-.025em b}\kern-.08em
    T\kern-.1667em\lower.7ex\hbox{E}\kern-.125emX}}
    
\begin{document}

\title{Equivalence and Conditional Independence
in Atomic Sheaf Logic%
\thanks{\textbf{{\copyright}Alex Simpson 2026. This is the author's version of the work. It is posted here for your personal use. Not for redistribution. The definitive Version of Record was published in the Journal of the ACM, 
http://dx.doi.org/10.1145/3809163\,.}}
}

\author{Alex Simpson%
\thanks{
Research supported by John Templeton Foundation grant number 39465 (2013--14).
\\
\begin{tabular}{ll}
\includegraphics[scale=0.12]{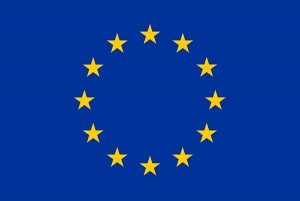}
\hspace*{-3ex}
& 
 \begin{tabular}[b]{l}
This project has received funding from the European Union's 
\\ Horizon  2020  research and innovation programme under the Marie 
\\  Sk{\l}odowska-Curie grant agreement No 731143.
\end{tabular}
\end{tabular}
\\
Research supported by ARIS Programme P1-0294.
}
\\[2ex]
Faculty of Mathematics and Physics, University of Ljubljana\\
 Institute of Mathematics, Physics and Mechanics, Ljubljana
}

\date{}

\maketitle

\begin{abstract}
We propose a semantic foundation for logics for reasoning in settings that possess 
a distinction between equality of variables, a coarser equivalence of variables, and a notion of 
conditional independence between variables.  We show that such relations
can be modelled naturally in  atomic sheaf toposes. 
Equivalence of variables is modelled by an intrinsic relation of 
\emph{atomic equivalence} that is possessed by every atomic sheaf. 
We identify additional structure  on the category generating the atomic topos  (primarily, the existence of a system of \emph{independent pullbacks}) that allows the relation of conditional independence to be interpreted in the topos. We then study the logic of equivalence and conditional  independence that is induced by the internal logic of the topos. This \emph{atomic sheaf logic} is a classical logic that validates a number of fundamental reasoning principles relating equivalence and conditional independence.
 As a concrete example of this abstract framework, we use the atomic topos over the category of surjections between finite nonempty sets as our main running example. 
 In this category, the interpretations of equivalence and conditional independence coincide with those given by the 
 multiteam semantics of independence logic, in which  
 the role of equivalence is taken by the relation of mutual inclusion. A major difference from independence logic is that, in  atomic sheaf logic, the multiteam semantics of the equivalence and conditional independence relations is embedded within a classical surrounding logic. 
 At the end of the paper, we briefly outline two other instances of our framework, to demonstrate its versatility.
 The first of these is a  category of \emph{probability sheaves}, in which atomic equivalence is 
 equality-in-distribution, and the conditional independence relation is the usual probabilistic one.
 Our other example is the \emph{Schanuel topos} (equivalent to nominal sets) where equivalence is 
 orbit equality and conditional independence amounts to a relative form of separatedness. 
\end{abstract}

\section{Introduction}
\label{section:introduction}

This paper provides a  study of fundamental logical principles for reasoning about relations of \emph{independence} and \emph{conditional independence} together with  an associated relation of \emph{equivalence}. The principles,
which are obtained via the abstract mathematical framework of \emph{sheaf theory}, 
are  general, in the sense that that they apply uniformly to  different instantiations of the notions of (conditional) independence and equivalence in a number of very different application areas. The paper  focuses on the mathematical development of a general theory that is intended to be cross-disciplinary in its applicability, but with computer science as a particularly prominent source of  target application areas.  

Notions of \emph{independence} and \emph{conditional independence} arise in many scientific areas. One particularly significant area is in probability and statistics, where it has long been recognised that conditional independence relations are subject to subtle rules of inference~\cite{dawidA,spohn}. Such rules, in a graphical formulation, are crucial in the technology of Bayesian networks~\cite{PP,GPP,GP}. In a more logical form, they have received recent interest in the area of program verification, where, for example, versions of separation logic based on probabilistic independence have been developed~\cite{BHL,BDHS,LAS}. 
% The resulting logics are typically intuitionistic, being based on  Kripke-like semantics in partially ordered models.

In a different direction, the \emph{dependence} and  \emph{independence logics} of V\" a\" an\" anen and Gr\" adel~\cite{Vaananen,GV}
are concerned with purely logical notions of dependence and independence between variables. 
Such logics are based on  \emph{team semantics}, which develops Hodges' compositional approach~\cite{hodges}  to the semantics of  
 \emph{independence-friendly logic}~\cite{HS}  into a fully fledged semantic framework. One of the attractions of team semantics is the close relationship it enjoys with database theory and notions of dependence and independence that arise therein~\cite{HK}. There is also an intriguing aspect to team semantics: it gives rise to logics that are
exotic in  character.
This point is discussed in more detail in Section~\ref{subsection:relationship-multiteam}.

The starting point of the present paper, in Sections~\ref{section:multiteams} and~\ref{section:atomic-sheaves},
 is the observation that the interpretation given by team semantics, more precisely by its \emph{multiteam} variant~\cite{DHKM}, to conditional independence 
statements is equivalent to interpreting these relations  in a certain sheaf topos, namely the topos of atomic sheaves on  the category
$\Sur$  of finite nonempty sets and surjections. This means that the team semantics of conditional independence automatically has 
a logic canonically associated with it: the internal logic of the topos. Since the topos is atomic, this internal logic is ordinary classical logic, albeit with a nonstandard semantics. We thus obtain a classical logic suitable for reasoning with
conditional independence relations endowed with their (multi)team semantics (Section~\ref{section:atomic-sheaf-logic}). 

One advantage of the atomic sheaf perspective on conditional independence is that it is very general.
We axiomatise structure, on the generating category of the  topos,  that gives rise to a canonical interpretation of conditional independence relations. For this, we define, in Section~\ref{section:independent-pullbacks},
the notion of \emph{independent pullback structure} on a category, 
closely related to the \emph{conditional independence structure} of~\cite{simpsonA}, but with a much more compact axiomatisation. 
We also expose a surprisingly rich interplay between independent pullback structure and the induced 
atomic sheaves. Building on this, in Section~\ref{section:atomic-conditional-independence}, we define
\emph{atomic conditional independence},  generalising the
multiteam conditional-independence relation to any atomic sheaf topos over 
a generating category with sufficient structure

Along the route to defining conditional independence, we observe, in Section~\ref{section:atomic-equivalence},
that every object of an atomic topos 
carries, in addition to the standard {equality} relation on the object, an additional intrinsic equivalence relation, which we call \emph{atomic equivalence}. Logically, this provides us with a canonical equivalence relation between variables that is, in general, coarser than equality. In the example of the atomic topos on the category $\Sur$, atomic equivalence turns out to coincide with a relation of interest in team semantics, namely the \emph{equiextension} relation.

One important contribution of the paper is the identification of fundamental axioms for 
relations of equivalence and conditional independence that are validated by the 
general interpretation of these relations in atomic toposes (over generating categories with enough structure).
These axioms include the usual quantifier-free axioms from the literature 
(for example, axioms formalising the reasoning principles for  conditional independence from~\cite{dawidA,spohn}), but also new first-order axioms that 
fully exploit the use of atomic sheaf logic. In Sections~\ref{section:atomic-equivalence} and~\ref{section:atomic-conditional-independence},
we identify five such principles: the \emph{transfer principle},
the \emph{invariance principle}, the \emph{principle of independent equivalence},
the
\emph{independent existence principle} and the property of \emph{existence preservation}.

Throughout Sections~\ref{section:atomic-sheaves}--\ref{section:atomic-conditional-independence}, the abstract definitions are illustrated in the case of   atomic sheaves over the category $\Sur$, which is our main running example, chosen because of  its connection to (multi)team semantics. In Sections~\ref{section:probability-sheaves} and~\ref{section:Schanuel} we present two other examples of our
general structure, in order to give some indication of its versatility. 
Section~\ref{section:probability-sheaves} presents an atomic sheaf topos over a category of probability spaces. The resulting category of \emph{probability sheaves} (first introduced in \cite{simpsonB}) includes sheaves of random variables, over which 
equality coincides with the probabilistic relation of \emph{almost sure equality},
atomic equivalence coincides with  the relation of \emph{equality in distribution}, and the atomic conditional independence relation coincides with the usual probabilistic relation.  
Section~\ref{section:Schanuel} very briefly indicates how  the Schanuel topos (which is equivalent to the category of nominal sets~\cite{GabP,pitts}) fits into our framework. In this case, atomic equivalence is the relation of orbit equality, and conditional independence amounts to a relative form of separatedness. 

Finally, in Section~\ref{section:related}, we discuss related and potential future work,
including a detailed comparison with team and multiteam semantics in Section~\ref{subsection:relationship-multiteam},
and a  discussion of potential computer science applications in Section~\ref{subsection:cs-applications}.

This paper is an expanded version of a conference paper~\cite{conference}, presented at the thirty-ninth annual ACM/IEEE Symposium on
Logic in Computer Science (LICS), held in 
Tallinn, Estonia in July 2024. In comparison with the conference paper, this journal version includes proofs of all main results, as well as an expanded discussion on sheaves in Section~\ref{section:atomic-sheaves} 
%an expanded account of independent pullback structure in~Section~\ref{section:independent-pullbacks},
% a new \emph{principle of independent equivalence}~\eqref{indep:Y} in Section~\ref{section:atomic-conditional-independence},
and also a substantially expanded presentation of our second main example, the category of probability sheaves, which occupies
Section~\ref{section:probability-sheaves}. We further  include three new appendices containing lengthy proofs that we prefer not to incorporate into  the main body of the paper, where they would interrupt the flow. 

\section{Multiteam semantics}

\label{section:multiteams}

\noindent
\emph{Dependence logic}~\cite{Vaananen} and \emph{independence logic}~\cite{GV}  extend first-order logic with new logical primitives expressing notions of dependence and independence between variables. These logics are based on the realisation that such new primitives can be interpreted semantically, by replacing the usual \emph{assignments} used to interpret variables in logical formulas with
 \emph{teams} (sets of assignments) or with \emph{multiteams} (multisets of assignments).
The relevant definitions are as follows, where $A$ is an arbitrary set.

\begin{itemize}

\item
An $A$-valued \emph{assignment} is a function $\mathcal{V} \to A$ where $\mathcal{V}$ is a (without loss of generality finite) set of variables.

\item
An $A$-valued \emph{team}~\cite{hodges,Vaananen}  is a set of assignments with common variable set $\mathcal{V}$.
%Team semantics defines $\mathcal{A} \models_t \phi$, where $t$ is a team.

\item
An $A$-valued \emph{multiteam}~\cite{DHKM} is a multiset of assignments with common variable set $\mathcal{V}$.
%Multiteam semantics defines  $\mathcal{A} \models_{m} \phi$, where $m$ is a multiteam.

\end{itemize}

%\noindent
%Sometimes teams and multiteams can be arbitrary (multi)sets [[REFS]], sometimes they are restricted to finite 
%(multi)sets [[REFS]].

\noindent
Teams and  multiteams give a canonical semantics to 
a variety of interesting new logical relations between variables,
such as those expressing \emph{dependence} $=\!(\Syn{x},\Syn{y})$, \emph{independence} $\Syn{x} \sIndep {\Syn{y}}$, 
\emph{conditional independence} $\Syn{x} \sIndep_{\Syn{z}} \Syn{y}$, \emph{inclusion} $\Syn{x} \subseteq \Syn{y}$,  \emph{equiextension} $\Syn{x} \bowtie \Syn{y}$  
and \emph{exclusion} $\Syn{x} |  \Syn{y}$, to give a non-exhaustive list. 
We review this in detail, in the case of multiteams, focusing on two of the above relations:
conditional independence and 
 equiextension. %  which are common to many areas of interest, such as those summarised in Section~\ref{section:introduction}.

A \emph{multiset} of elements from a set $A$ is a function $m \colon A \to \mathbb{N}$, which assigns to every element $a \in A$ a \emph{multiplicity} $f(a)$. A multiset $m$ is \emph{finite} if its \emph{support} (the set $\Supp(m) := \{a \mid m(a) > 0\}$) is 
finite. A multiset $m \colon A \to \mathbb{N}$ can alternatively be presented by a set $\Omega$ together with a function $M \colon \Omega \to A$ satisfying, for all $a \in A$, the fibre $M^{-1}(a)$ has cardinality $m(a)$. 
The elements of $\Omega$ can be thought of as names for distinct element occurrences in the multiset (so each element in $A$ has as many names as its multiplicity). Note also that the function $M$ has the support set $\Supp(m)$ as its image.
Of course a multiset $m \colon A \to \mathbb{N}$ has many different presentations by finite-fibre functions. However, given two such representations $M \colon \Omega \to A$ and $M' \colon \Omega' \to A$, there exists  %(using choice!)
 a  %(not necessarily unique) 
bijection $i \colon \Omega \to \Omega'$ such that $M = M' \circ i$. (The proof of this statement, although simple, requires the axiom of choice.) 
So multisets are in one-to-one correspondence with isomorphism classes of presentations.

In the case of a finite multiset $m \colon A \to \mathbb{N}$, the domain set $\Omega$ of a presentation $M \colon \Omega \to A$ is necessarily finite, and all functions with finite domain present finite multisets. Thus there is a one-to-one correspondence between finite multisets and isomorphism classes of finite-domain presentations. 
(Moreover, because the multisets are now finite, the axiom of choice is no longer needed.)

Since a \emph{multiteam} is  a multiset of assignments with a common  $\mathcal{V}$.
it can be presented by a finite-fibred function of the form
\[
{M} \colon \Omega \to (\mathcal{V} \to A) \enspace .
\]
%This formulation suggests it would be natural to consider a generalisation of multiteam, in which one drops the 
%finite fibre requirement and obtains functions $M$ that correspond to multiteams with potentially infinite multiplicities given as
%functions $m$ that assign an arbitrary cardinality  $m(a)$ to every $a \in A$. For the present paper, we put this suggestion to one side, and instead go in the other direction, restricting to finite multiteams, as, for example, in
As in~\cite{DHKM}, we  restrict attention to finite multiteams. % (A similar restriction is made, e.g., in~\cite{DHKM}.)
Henceforth, by \emph{multiteam} we mean a finite multiset of assignments with common $\mathcal{V}$.   Such finite multiteams  correspond to functions $M$, as above, for which the set $\Omega$ is finite. 
Equivalently, by  transposition, a multiteam can be represented by a function of the form
\[
\underline{M} \colon \mathcal{V} \to ( \Omega  \to A)
\]
While this is just a simple set-theoretic reorganisation of the notion of multiteam, it provides 
an illuminating alternative perspective on multiteam semantics, which we now elaborate.

One can think of a function $X \colon \Omega \to A$ as a \emph{nondeterministic variable} valued in $A$. Here the terminology is motivated by analogy with the notion of  \emph{random variable} from probability theory.
%specified as a (measurable) function from a \emph{probability space} $\Omega$ (the \emph{sample space}) to a measurable space $A$ of possible values. 
In our setting, we view the \emph{set} $\Omega$ as a finite \emph{sample set}, a nondeterministic version of a
\emph{sample space} in probability theory. The sample set represents a realm of possible nondeterministic choices.
% with no probability associated---one is  concerned only with the collection of sample possibilities available. 
With this terminology, a multiteam presented as
$\underline{\rho} \colon \mathcal{V}\to ( \Omega  \to A)$
is simply an assignment of $A$-valued nondeterministic variables (with shared sample set) to logical variables. 
(In this paper, we restrict to finite sample sets.
Nevertheless, the notion of nondeterministic variable obviously generalises to arbitrary sample  sets $\Omega$.)

% we restrict to finite nonempty sets, corresponding to finite nondeterminism.)

We now use the above formulation of multiteams as assignments of nondeterministic variables to recast definitions from multiteam semantics (as in~\cite{DHKM}). Technically, this is simply a straightforward matter of translating the definitions along the equivalence between the two formulations of multiteam. However, even if mathematically equivalent, our formulation of multiteam encourages a different `local' style of presentation, where the sample sets $\Omega$ play a role similar to that played by \emph{possible worlds} in Kripke semantics and by \emph{forcing conditions} in set theory.

Before addressing semantics, we introduce our syntax. 
For greater generality, we work with a multi-sorted logic. 
This also has the advantage that the sorting constraints on logical primitives provide useful information about their  generality in scope. 
Accordingly, we assume a set $\Sort$ of basic syntactic \emph{sorts} $\Syn{A}, \Syn{B}, \Syn{C}, \dots$.
%We also assume a collection 
%of primitive relation symbols, where each relation symbol $\Syn{R}$ has an \emph{arity} $\Arity(\Syn{R}) \in \Sort^*$ (that is the arity is a finite sequence of sorts). 
Variables $\Syn{x^A}$ have explicit sorts.
We consider three forms of atomic formula.

\begin{itemize}
%\item If $\Arity(\Syn{R}) = \Syn{A}_1 \dots \Syn{A}_n$ and $\Syn{x}_1^{\Syn{A}_1}, \dots, \Syn{x}_n^{\Syn{A}_n}$ is a list of variables of the corresponding sorts, then $\Syn{R}(\Syn{x_1^{{A}_1}}, \dots, \Syn{x}_n^{\Syn{A}_n})$ is an atomic formula.

\item If $\Syn{x^A},\Syn{y^A}$ have the same sort, then $\Syn{x^A} = \Syn{y^A}$ is an atomic formula.

\item If $\Syn{x}_1^{\Syn{A}_1}, \dots, \Syn{x}_n^{\Syn{A}_n}$ and 
$\Syn{y}_1^{\Syn{A}_1}, \dots, \Syn{y}_n^{\Syn{A}_n}$ are two lists of variables of the same length $n \geq 0$ with identical sort lists, %  $\Syn{A}_1,  \dots, \Syn{A}_n$
then 
\begin{equation}
\label{equation:equiv-formula}
\Syn{x}_1^{\Syn{A}_1}, \dots, \Syn{x}_n^{\Syn{A}_n} \Equiv \Syn{y}_1^{\Syn{A}_1}, \dots, \Syn{y}_n^{\Syn{A}_n}
\end{equation}
is an atomic formula.

\item If $\Syn{x}_1^{\Syn{A}_1}, \dots, \Syn{x}_m^{\Syn{A}_m}$ and 
$\Syn{y}_1^{\Syn{B}_1}, \dots, \Syn{y}_n^{\Syn{B}_n}$ and
$\Syn{z}_1^{\Syn{C}_1}, \dots, \Syn{z}_l^{\Syn{C}_l}$
are three lists of variables (with $m,n,l \geq 0$)
then 
\begin{equation}
\label{equation:indep-formula}
\Syn{x}_1^{\Syn{A}_1}, \dots, \Syn{x}_m^{\Syn{A}_m} \sIndep \Syn{y}_1^{\Syn{B}_1}, \dots, \Syn{y}_n^{\Syn{B}_n}
\Cond \Syn{z}_1^{\Syn{C}_1}, \dots, \Syn{z}_l^{\Syn{C}_l}
\end{equation}
 is an atomic formula.
\end{itemize}
%As usual in team semantics, we assume that $\mathcal{A}$ is a structure (with underlying set $A$) for a relational signature, where each $n$-ary relation $\Syn{R}$ in the signature is interpreted by a subset 
%$\Syn{R}^{\mathcal{A}} \subseteq A^n$. Formulas are then constructed using the usual rules for generating first-order formulas, but applied to a wider class of atomic formulas than usual. In this paper, we consider atomic formulas of the following five forms.
%\begin{align*}
%& R(\Syn{x}_1, \dots, \Syn{x}_n) 
%&& \Syn{x} = \Syn{y}
%&& \Syn{x} \Equiv \Syn{y}
%&& \Syn{x} \sIndep \Syn{y}
%&& \Syn{x} \sIndep \Syn{y} \Cond \Syn{z}
%\end{align*}
The first formula expresses  \emph{equality}, as in ordinary (multi-sorted) first-order logic. The remaining two are atomic constructs borrowed from logics associated with team semantics. 

The formula $\vec{\Syn{x}} \Equiv \vec{\Syn{y}}$ represents what we call \emph{equivalence}, which arises in the team-semantics literature as \emph{equiextension}
${\vec{\Syn{x}} \subseteq \vec{\Syn{y}}} \wedge {\vec{\Syn{y}} \subseteq \vec{\Syn{x}}}$, sometimes written with the notation
$\vec{\Syn{x}} \bowtie \vec{\Syn{y}}$. Our more neutral notation and terminology reflects the fact that we will later consider other interpretations of the $\Equiv$ relation. The use of vectors of variables on either side is needed because equivalence is a relation that holds
betwen the vectors $\Vec{\Syn{x}}$ and $\Vec{\Syn{y}}$ \emph{jointly}, and does not reduce to a conjunction of equivalences between components. 

The formula $\vec{\Syn{x}} \sIndep \vec{\Syn{y}} \Cond \vec{\Syn{z}}$ represents \emph{conditional independence} from the \emph{independence logic} of~\cite{GV}, where it is written 
$\vec{\Syn{x}} \sIndep_{\vec{\Syn{z}}} \vec{\Syn{y}}$. In our syntax, we take the conditioning variables out of the subscript position in order to give them more prominence, adopting a notation that is familiar from probability theory. An important special case of conditional independence is when the 
sequence $\vec{\Syn{z}}$ is empty. In such cases, we write simply $\vec{\Syn{x}} \sIndep \vec{\Syn{y}}$ for the resulting relation, which expresses unconditional independence.

It is of course  the atomic formulas $\vec{\Syn{x}} \Equiv \vec{\Syn{y}}$ and $\vec{\Syn{x}} \sIndep \vec{\Syn{y}} \Cond \vec{\Syn{z}}$  that give us \emph{equivalence} and \emph{conditional independence} in the title of this paper.

To define the semantics, we assume we have,  for every sort $\Syn{A}$, an associated set $\Sem{\Syn{A}}$.
%and for every relation symbol $\Syn{R}$
%with $\Arity(\Syn{R}) = \Syn{A}_1 \dots \Syn{A}_n$, the structure interprets $\Syn{R}$ as a 
%subset of the relevant product set:
%\[
%\Sem{\Syn{R}}_{\mathcal{A}} ~ \subseteq ~ \Sem{\Syn{A}_1}_{\mathcal{A}} \times \dots  \times \Sem{\Syn{A}_n}_{\mathcal{A}} \enspace .
%\]
%\noindent
%(Henceforth, we drop the $\mathcal{A}$ subscripts.)
%
%Our main semantic definition defines
%a forcing relation
%\[
%X \Forces_{\Sh{\rho}} \Phi \enspace ,
%\]
%where $\Phi$ is a formula, $X$ is an object of $\sCat{C}$ and 
%\[
%\Sh{\rho} ~ \in ~ \prod_{\Syn{x}^\Syn{A} \in \{\Syn{x}_1^{\Syn{A}_1}, \dots, \Syn{x}_n^{\Syn{A}_n}\}}  \Sh{A}(X)
%\]
% is what we call an \emph{$X$-assignment}---it maps every variable $\Syn{x}^\Syn{A}$
% in a set $\{\Syn{x}_1^{\Syn{A}_1}, \dots, \Syn{x}_n^{\Syn{A}_n}\}$ containing every free variable of $\Phi$ to an element of $\Sh{A}(X)$, where $\Sh{A}$ is the sheaf interpreting the sort $\Syn{A}$ of the variable. 
% 
In a multi-sorted setting, an assignment 
for a finite set $\mathcal{V}$ of variables
is an element
\[
\rho ~ \in ~ \prod_{\Syn{x}^\Syn{A} \in \mathcal{V}}  \Sem{A} \enspace ,
\]
and a multiteam is a finite multiset of assignments.
In the standard multiteam semantics, a formula $\Phi(\Syn{x}_1^{\Syn{A}_1},\dots,\Syn{x}_n^{\Syn{A}_n})$ (i.e., all free variables are in 
$\{\Syn{x}_1^{\Syn{A}_1}, \dots, \Syn{x}_n^{\Syn{A}_n}\}$) is given a satisfaction relation 
\begin{equation}
\label{equation:satisfaction}
\models_m \Phi \enspace ,
\end{equation}
where $m$ is a multiteam of $\{\Syn{x}_1^{\Syn{A}_1}, \dots, \Syn{x}_n^{\Syn{A}_n}\}$-assignments.
%  i.e., multisets of $\{\Syn{x}_1^{\Syn{A}_1}, \dots, \Syn{x}_n^{\Syn{A}_n}\}$-assignments.
%\footnote{In some versions of team semantics, a more general formulation is required, due to the so-called \emph{locality} problem, meaning that it is not enough to restrict assignments to the free variables of $\phi$, see [[REFS]]. Since such issues won't arise in anything we do, we can ignore this subtle point.}
If instead we adopt  the reformulation of multisets described above, a multiteam is given as a \emph{single} assignment 
\begin{equation}
\label{equation:new-multiteam}
\Sh{\rho} ~ \in ~ \prod_{\Syn{x}^\Syn{A} \in \mathcal{V}}  \left(\Omega \to \Sem{A} \right)
\end{equation}
of nondeterministic variables to logic variables, and the satisfaction relation can then be rewritten as
\begin{equation}
\label{equation:satisfaction-B}
\models_{\Sh{\rho}} \Phi \enspace .
\end{equation}
It turns out to be helpful % (as we shall see in Section~\ref{section:atomic-sheaf-logic}) 
to make the sample set $\Omega$, that occurs implicitly within $\Sh{\rho}$, explicit in the notation,
so we write
\begin{equation} 
\label{equation:forces}
\Omega \Forces_{\underline{\rho}} \Phi \enspace \, .
\end{equation}
We here switch to the `forcing' notation $\Forces$, since we  shall view $\Omega$  as a 
`possible world' or `condition' 
(capturing all the nondeterminism that the multiteam uses)
that determines the `local truth' of $\Phi$.
We stress that relations \eqref{equation:satisfaction}, \eqref{equation:satisfaction-B} and
\eqref{equation:forces} all have exactly the same meaning. The only differences are in the formulation of multiset that is used, and whether or not $\Omega$ is explicit in the notation. 

%view the above `global' relation as being stratified into a family of local relations indexed by sample sets $\Omega$.
%We use the forcing notation $\mathcal{A}, \, \Omega \Forces_{\underline{\rho}} \phi$ for such relations, viewing $\Omega$ as a `possible world' or `condition' capturing all the nondeterminism that the multiteam uses. Henceforth, we leave the structure $\mathcal{A}$ implicit, and write:
%
%We stress that the statement $(\mathcal{A},\!) \,\Omega \Forces_{\underline{\rho}} \phi $ has \emph{exactly the same meaning} as the statement $\mathcal{A} \Forces_{\underline{\rho}} \phi $  directly reformulating 
%$\mathcal{A} \models_m \phi $. The only difference is that the sample set $\Omega$ occurring implicitly within $\Sh{\rho}$ has been written  into the notation explicitly. 
%

\begin{figure}[t]
\begin{align*}
%\Omega \Forces_{\!\underline\rho}  \Syn{R}(\Syn{x}_1^{\Syn{A}_1}, \dots, \Syn{x}_n^{\Syn{A}_n}) ~ \Iff ~ & ~
%  \forall \omega \in \Omega. ~ (\underline\rho(\Syn{x}_1^{\Syn{A}_1})(\omega), \dots, \underline\rho(\Syn{x}_n^{\Syn{A}_n})(\omega)) \in \Sem{\Syn{R}}
% \\
\Omega \Forces_{\!\underline\rho}  \Syn{x^A} \! = \! \Syn{y^A} ~ \Iff ~ & ~ \underline\rho(\Syn{x^A})=\underline\rho(\Syn{y^A}) 
\quad \text{(equal functions $\Omega \to \Sem{\Syn{A}}$)}
\\
\Omega \Forces_{\!\underline{\rho}}  \vec{\Syn{x}} \! \Equiv \! \vec{\Syn{y}} ~ \Iff ~ & ~ \underline{\rho}(\vec{\Syn{x}}) \bowtie  \underline{\rho}(\vec{\Syn{y}}) 
\\
%\Omega \Forces_{\!\underline\rho} {\vec{\Syn{x}}} \sIndep {\vec{\Syn{y}}}  ~  \Iff  ~ & ~ 
%   \Sh{\rho}(\Syn{x}) \Indep \Sh{\rho}(\Syn{y})
% \\
%\forall a,b \in A. ~ (\exists \omega \in \Omega.~ \underline{\rho}(x)(\omega) \!=\! a) ~ \wedge ~
%(\exists \omega \in \Omega.~ \underline{\rho}(y)(\omega)\! =\! b)  \\
%& \phantom{\forall a,b \in A. ~ ~~~}
%~ \to ~
%\exists \omega \in \Omega.~ (\underline{\rho}(x)(\omega)\! = \!a  \wedge \underline{\rho}(y)(\omega)\! =\! b)
%\\[1ex]
\Omega \Forces_{\!\underline\rho}  {\vec{\Syn{x}}} \sIndep {\vec{\Syn{y}}}  \Cond {\vec{\Syn{z}}}  ~  \Iff  ~ & ~
\Sh{\rho}(\vec{\Syn{x}}) \Indep \Sh{\rho}(\vec{\Syn{y}}) \Cond \Sh{\rho}(\vec{\Syn{z}}) 
%\forall a,b,c \in A. ~ (\exists \omega \in \Omega.~ \underline{\rho}(x)(\omega) \!= \!a \wedge \underline{\rho}(z)(\omega) \!= \! c) ~ \wedge 
%(\exists \omega \in \Omega.~ \underline{\rho}(y)(\omega) \!= \!b \wedge \underline{\rho}(z)(\omega)\! = \!c) \\
%& \phantom{\forall a,b,c\in A. ~ ~~~}
%~ \to ~
%\exists \omega \in \Omega.~ (\underline{\rho}(x)(\omega) \! = \! a  \wedge \underline{\rho}(y)(\omega) \! = \! b  \wedge \underline{\rho}(z)(\omega)\! = \!c)
\end{align*}
\caption{Multiteam semantics of atomic formulas}
\label{figure:atomic-semantics}
\end{figure}

Figure~\ref{figure:atomic-semantics} 
% and~\ref{figure:general-semantics} 
defines the 
forcing relation $\Omega \Forces_{\underline{\rho}} \Phi$ directly in terms of 
our reformulated multiteams, as in~\eqref{equation:new-multiteam}, for atomic formulas $\Phi$.
% Figure~\ref{figure:atomic-semantics} presents the clauses for atomic formulas.
In the clauses for equivalence and  independence we use the notation $\Sh{\rho}(\vec{\Syn{x}})$, where $\vec{\Syn{x}}$ is a vector  of variables $\Syn{x}_1^{\Syn{A}_1}, \dots, \Syn{x}_n^{\Syn{A}_n}$, to represent
the $(\Sem{\Syn{A}_1} \times \dots \times \Sem{\Syn{A}_n})$-valued nondeterministic variable 
\[
\Sh{\rho}(\vec{\Syn{x}}) ~ := ~ \omega \mapsto (\Sh{\rho}(\Syn{x}_1^{\Syn{A}_1}), \dots, \Sh{\rho}(\Syn{x}_n^{\Syn{A}_n})) ~ \colon~  \Omega \to \Sem{\Syn{A}_1} \times \dots \times \Sem{\Syn{A}_n} \enspace .
\]
We also write  $X \bowtie Y$ and $X \Indep Y \Cond Z$ for the semantic relation of \emph{equiextension} and 
\emph{conditional independence} between nondeterministic variables, as defined below.
\begin{definition}[Equiextension for nondeterministic variables] 
\label{definition:equiextension}
Two  nondeterministic variables $X \colon \Omega \to A$ and $Y \colon \Omega \to A$  are \emph{equiextensive} (notation $X \bowtie Y$) if
they have equal images, i.e., $X(\Omega) = Y(\Omega)$.
%\begin{align*}
% & (\exists \omega \in \Omega.~ X(\omega) \!=\! a) ~ \text{and} ~
%(\exists \omega \in \Omega.~ Y(\omega)\! =\! b) ~ \\
%& ~~~~~~~~~~~~~~~ \text{implies} ~\exists \omega \in \Omega.~ X(\omega) \!=\! a ~ \text{and} ~ Y(\omega)\! =\! b \enspace .
%\end{align*}
\end{definition}

\begin{definition}[Conditional independence for nondeterministic variables]
\label{definition:conditional-independence-NV}
 Let  $X \colon \Omega \to A$, $~ Y \colon \Omega \to B$ and $Z \colon \Omega \to C$ be nondeterministic variables. We say that $X$ and $Y$ are  \emph{conditionally independent} given $Z$ 
(notation $X \Indep Y \Cond Z$) if, for all $a \in A, b \in B, c \in C$,
\begin{align*}
  (\exists \omega \in \Omega.~ X(\omega) \!=\! a ~ \text{and}  ~ Z(\omega) \!= \!c) ~ 
 \text{and}   & ~ (\exists \omega \in \Omega.~ Y(\omega)\! =\! b~ \text{and} ~ Z(\omega) \!=\! c) ~ \\
& \quad \quad \text{implies} ~\exists \omega \in \Omega.~ X(\omega) \!=\! a~ \text{and} ~ Y
(\omega)\! =\! b ~ \text{and} ~ Z(\omega) \!= \!c \enspace .
\end{align*}
\end{definition}

In the literature on (in)dependence logics, the semantic clauses for atomic formulas are extended with clauses 
giving meaning to the  logical connectives and quantifiers. A number of inequivalent ways of achieving this appear in the literature~\cite{Vaananen,GV,DHKM}. All share the feature that the resulting logics are exotic. We shall discuss this in more detail in Section~\ref{subsection:relationship-multiteam}.

In this paper, we consider a different approach to embedding the equivalence and conditional independence constructs, with their multiteam semantics, in a full multi-sorted first-order logic. We observe  that the multiteam semantics of the atomic constructs  lives naturally in a certain \emph{atomic sheaf topos}. 
and then we  make use of the standard \emph{internal logic} of the topos, which in the case of an atomic topos is  classical logic.

%%%%%%%%%%%%%%%%%%%%%%%%%%%%%%%%%%%
%%% SECTION
%%%%%%%%%%%%%%%%%%%%%%%%%%%%%%%%%%%%

\section{Atomic sheaves}
\label{section:atomic-sheaves}

%The purpose of the remainder of the paper is to place the classical multiteam semantics from Section~\ref{section:classical-multiteam} in a more general and, we believe, illuminating context. We shall observe that it is a fragment of the 
%\emph{internal logic} of a category of \emph{sheaves} over a category with an \emph{atomic Grothendieck topology}. 
%The Kripke-Joyal semantics in the sheaf category show us how to  extend the semantic definitions of 
%Section~\ref{section:multiteams} to encompass all connectives of first-order logic. 
%Since the category in question is an atomic topos [[REF]], the resulting logic is classical [[MM]]. 

In this section, we define the notion of \emph{atomic sheaf topos}, which is a special kind of \emph{Grothendieck topos}.
We restrict attention to 
presenting the definitions and results we shall make use of, attempting to do so in such a way that they can be understood from first principles given knowledge of core category theory. For further contextualisation,  \cite{MM} is an excellent source.

A \emph{presheaf} on a small category $\sCat{C}$ is a functor 
$P \colon \Op{\sCat{C}} \to \Set$ (note the contravariance). 
The \emph{presheaf category} $\Psh(\sCat{C})$ is the functor category $\Set^{\Op{\sCat{C}}}$. Given a presheaf $P$,
object $Y$ of $\sCat{C}$, element $y \in PY$ and map $f \colon X \rMap{} Y$ in $\sCat{C}$, we write $y \cdot_P f$ for the element $P(f)(y) \in PX$, or simply $y \cdot f$ when $P$ is clear from the context.

\begin{example}[Representable presheaves] For any object $Z \in \sCat{C}$, the \emph{representable presheaf} $\Yon Z := \sCat{C}(-,Z)$ is defined by
\begin{itemize}
\item For any object $X \in \sCat{C}$, define $\Yon Z (X) ~ := ~ \sCat{C}(X,Z)$,
i.e., the hom set.

\item For any map $f \colon Y \rMap{} X$ in $\sCat{C}$ and $g \in (\Yon Z)(X) $,
define $g \cdot f := g \circ f$.
\end{itemize}
The object mapping $Z \mapsto \Yon Z$ extends to a full and faithful functor $\Yon: \sCat{C} \to \Psh(\sCat{C})$, the \emph{Yoneda functor}~\cite{MM}.
\end{example}

\begin{example}[Product presheaves] Let $P_1, \dots, P_n$ be presheaves on $\sCat{C}$.
Define the \emph{product presheaf} $P_1 \times \dots \times  P_n$ in $\Psh(\sCat{C})$ by:
\begin{itemize}
\item For any object $X \in \sCat{C}$, define 
\[
(P_1 \times \dots \times  P_n)(X) ~ := ~ P_1(X) \times \dots \times P_n(X) \enspace ,
\]
i.e., the product of sets.

\item For any map $f \colon Y \rMap{} X$ in $\sCat{C}$ and $(x_1, \dots, x_n) \in (P_1 \times \dots \times  P_n)(X)$,
define
\[
(x_1, \dots, x_n) \cdot f ~ := ~ (x_1 \cdot_{P_1} f, \dots, x_n \cdot_{P_n} f) \enspace ,
\]
% using the actions of $P_1, \dots, P_n$ on the respective components.
\end{itemize}
\end{example}
\noindent
The above definition generalises to infinite products, and further to arbitrary category-theoretic limits and colimits, all of which are defined on presheaves in a similar (pointwise) way, using the corresponding definitions in the category of sets.

The next example is central to this paper.

\begin{example} Let $\Sur$ be (a small category equivalent to) the category whose objects are non-empty finite sets
% of some given infinite set,\footnote{The reason  for restricting to subsets of a given infinite set is so that the category is small.}
 and whose morphisms are surjective functions. For any set $A$, we have a presheaf $\NEls(A)$ in $\Psh(\Sur)$ of 
$A$-valued \emph{nondeterministic variables} (in the sense of Section~\ref{section:multiteams}), defined as follows.
\begin{itemize}
\item For any object $\Omega$ of $\Sur$, define $\NEls(A)(\Omega)$ to be the set of all functions $\Omega \to A$.

\item For any map $p \colon \Omega' \rMap{} \Omega$ in $\Sur$, and $X \in \NEls(A)(\Omega)$ define 
 $X \cdot p$ to be $X \circ p \in \NEls(A)(\Omega')$.

\end{itemize}

\end{example}

Grothendieck introduced a very general notion of what it means for a presheaf $P \in \Psh(\sCat{C})$ to be a 
\emph{sheaf} relative to a \emph{Grothendieck topology} on $\sCat{C}$. A Grothendieck topology $\mathcal{J}$  specifies, for every object $X$, a collection $\mathcal{J}_X$ of families of maps with codomain $X$, in which each family of maps 
$(c_i \colon Y_i \rMap{} X)_{i \in I} \in \mathcal{J}_X$ is deemed to provide a \emph{covering family} (more briefly \emph{cover}) for $X$. A presheaf $P$ is a $\mathcal{J}$-\emph{sheaf} if, for every such cover, every \emph{matching} family of elements $(y_i \in P(Y_i))_{i\in I}$ 
% (for every $d_i \colon Z \rTo Y_i$ and $d_j \colon Z \rTo Y_j$ such that $d_i \circ f_i = d_j \circ f_j$, it holds that $y_i \cdot d_i = y_j \cdot d_j$) 
has a unique \emph{amalgamation} $x \in P(X)$.
% (an element $x \in P(X)$ such that, $y_i = x \cdot c_i$, for all $i \in I$).
The high-level idea is that the \emph{matching} property, which says that the $y_i$ elements agree with each other on overlapping parts of the cover, allows all the $y_i$ to be glued together into a single \emph{amalgamation} $x$, which  is an  element of $P(X)$.
We shall not give the general definitions underlying the emphasised words because, for this paper, it is not necessary to understand the notion of sheaf in its full generality. 
Nonetheless, there is a point about the general definition worth making. The intuition that is usually presented for the general definition is that the
matching condition for the family $(y_i \in P(Y_i))_{i\in I}$ means that the different $y_i$ are compatible with each other, and then the unique amalgamation `glues' these compatible elements together to form a single element $x \in P(X)$, which is possible because the object $X$ is covered by the family $(Y_i)_{i \in I}$. %(via the $c_i$ maps). 
In this paper, we are going to work only with sheaves for  \emph{atomic} Grothendick topologies, for which the usual general intuition for sheaves outlined above is not very helpful. In the case of an atomic topology, covers are single maps $c \colon Y \rMap{} X$, matching families contain only one element $y \in P(Y)$ (it needs to match only with itself, which turns out to be a nontrivial condition) and the amalgamation
$x \in P(X)$ is obtained from $y$ alone, so the usual `gluing' intuition does not apply.
%  (the very word amalgamation is misleading in this case). 
%In spite of the above, we shall use the standard words \emph{matching} and \emph{amalgamation}, since they are firmly established.
% and moreover very relevant to the  general idea of sheaves. 
Instead, we shall use the terminology  \emph{invariant element} in place of matching family, and \emph{descendent} in place of
 amalgamation, since this seems more appropriate in the context of an atomic topology.
%In the following, we  shall define the notions of matching family and amalgamation as they apply in the special case of atomic
 %topologies.  For this purpose, we shall use the terminology  \emph{invariant element} instead of matching family, and \emph{descendent} instead of
 %amalgamation, since this seems more appropriate in the context of an atomic topology where 
 %a cover is (without loss of generality) just a single map and a matching family comprises just one element.
  
  We first introduce the atomic sheaf concept using the example of  $\Sur$, and then follow this with the generalisation to an arbitrary small category $\sCat{C}$.
In the case of $\Sur$, an object $\Omega$ can be thought of as representing a `world' of currently available nondeterministic choices, and
a map $c : \Omega' \rMap{} \Omega$ specifies an extension of the existing nondeterministic choices in $\Omega$  to accommodate the additional 
nondeterminism potentially available in $\Omega'$. Nondeterministic variables form a presheaf  $\NEls(A)$ simply because any nondeterministic variable 
$X \in \NEls(A)(\Omega)$ extends via $c$ to a corresponding $\Omega'$-based nondeterministic variable
$X \cdot c := X \circ c \in  \NEls(A)(\Omega')$. This latter nondeterministic variable is defined for all nondeterministic choices in $\Omega'$, but only makes use of nondeterminism already available in $\Omega$; that is, $(X \cdot c) (\omega') = (X \cdot c) (\omega'')$ for any $\omega',\omega'' \in \Omega'$ for which $c(\omega') = c(\omega'')$. 
Furthermore, %the presheaf $\NEls(A)$ enjoys the following \emph{sheaf property}.
every element $Y \in  \NEls(A)(\Omega')$, that only makes use of nondeterminism in $\Omega$, arises 
as $Y = X \cdot c$ for a unique 
$X \in \NEls(A)(\Omega)$. In other words, it
has a unique representation as a \emph{bona fide}  $\Omega$-based nondeterministic variable $X$.
In order to formulate this technically, we say that a nondeterministic variable 
$Y \in  \NEls(A)(\Omega')$ is \emph{$c$-invariant} if % it only makes use of nondeterminism already available in $\Omega$; i.e.,
$Y(\omega') = Y(\omega'')$ for any $\omega',\omega'' \in \Omega'$ for which $c(\omega') = c(\omega'')$. 
The presheaf $\NEls(A)$ then satisfies:
every $c$-invariant $Y\in  \NEls(A)(\Omega')$ arises as $X \cdot c$ for a unique 
$X \in \NEls(A)(\Omega)$, which we call the \emph{$c$-descendent} of $Y$. 
As we shall see below, the property we have just elucidated asserts that the presheaf $\NEls(A)$ is a \emph{sheaf} for the 
\emph{atomic Grothendieck topology} on the category $\Sur$.

A similar story can be told for any small category $\sCat{C}$ for which 
an object $X \in \sCat{C}$ can be thought of  as a world of current possibilities, and a map $c \colon Y \rMap{} X$ represents a way of extending the current world to another world $Y$ with additional possibilities. Given a presheaf $P$ an element $x \in P(X)$ and map $c \colon Y \rMap{} X$, the element $x \cdot c \in P(Y)$ represents the extension of $x$ to incorporate the new possibilities from $Y$. The extended element $x \cdot c$ enjoys the 
property of \emph{$c$-invariance} (Definition~\ref{definition:c-invariant} below), 
which formalises that $x \cdot c$ does not depend on any of the possibilities in $Y$ beyond those already available in $X$. 
Moreover, for any $y \in P(Y)$ that is {$c$-invariant}, the definition of \emph{atomic sheaf} (Definition~\ref{definition:atomic-sheaf} below) says that  there must exist a unique $x \in P(X)$ that, via the equation
$y = x \cdot c$, makes explicit the true dependency of $y$ only on $X$. 

The main intuition underpinning the  above discussion can be  summarised as follows. 
In the context of a category $\sCat{C}$, for which we think of  maps $c \colon Y \rMap{} X$ as extending the possibilities offered by state $X$ to a more refined set of possibilities offered by state $Y$,

\begin{itemize}
\item the \emph{presheaf} property of $P$ says that we can \emph{extend} any element $x \in P(X)$, defined using the possibilities at $X$, to a corresponding element  $x \cdot c \in P(Y)$ that, although defined at $Y$, does not exploit the potential additional generality of $Y$;
\item and the \emph{atomic sheaf} property says that, for any element $y \in P(Y)$, defined using the possibilities at $Y$ in such a way that $y$ does not exploit the potential greater generality afforded by $Y$ over $X$, there exists a unique corresponding element $x \in P(X)$
that makes explicit the dependency of $y$ only on possibilities offered by $X$.
\end{itemize}

%In order to consider atomic
Since we are  interested only in atomic topologies, we can define the sheaf property  (Definition~\ref{definition:atomic-sheaf} below) directly,
 without needing to introduce the general notion of  Grothendieck topology. 
However, we do need the 
% We aim to define what it means for a presheaf $P \in \Psh(\sCat{C})$ to be an \emph{atomic sheaf}; that is a sheaf with respect to the \emph{atomic (Grothendieck) topology} to exist on $\sCat{C}$.
%Although we provide a brief overview of Grothendieck topologies in general below, in the technical development we shall be exclusively interested in the special case of atomic topologies, which are only relevant to categories $\sCat{C}$ that satisfy a certain condition.
atomic Grothendieck topology  to 
exist on the base category $\sCat{C}$, which happens if and only if the category $\sCat{C}$ is \emph{coconfluent}.
% (which may be summarised as: \emph{every cospan completes to a commuting square}). % in the following sense.

\begin{definition}[Coconfluence] A category $\sCat{C}$ is \emph{coconfluent}\footnote{In~\cite[A 2.1.11(h)]{johnstone}  $\sCat{C}$ is said to satisfy the \emph{right Ore condition}.}
if for any cospan $X \rMap{f} Z \lMap{g} Y$, there exists a span $X \lMap{u} W \rMap{v} Y$ such that $f \circ u = g \circ v$. %[[DIAGRAM]]?
\end{definition}

\begin{proposition} 
\label{proposition:Sur-coconfluent}
$\Sur$ is coconfluent.
\end{proposition}

\begin{proof} Consider any cospan $\Omega_X \rMap{p} \Omega_Z \lMap{q}\Omega_Y$ in $\Sur$, Define 
\[
\Omega_W ~ := ~ \{(x,y) \in \Omega_X \times \Omega_Y \mid p(x) = q(y) \} \enspace .
\]
Then $u := (x,y) \mapsto x$ and $v := (x,y) \mapsto y$ define surjective functions
$\Omega_W \Epi \Omega_X $ and $\Omega_W \Epi \Omega_Y $, hence they are maps in $\Sur$, for which indeed $p \circ u  = q \circ v$.  (More briefly, the pullback in $\Set$ is a commuting square in $\Sur$, though not a pullback in $\Sur$.)
\end{proof}

  %see, e.g., \cite[A 2.1.11(h)]{johnstone}
  Let $P \in \Psh(\sCat{C})$ be a presheaf. 
% on a coconfluent category $\sCat{C}$.
\begin{definition}[Invariant element] 
\label{definition:c-invariant}
Given $c \colon Y \rMap{} X$ and $y \in P(Y)$ we say that 
$y$ is \emph{$c$-invariant} if, for any parallel pair of maps $d, e \colon Z \rMap{} Y$ such that $c \circ d = c \circ e$, it holds that $y \cdot d = y \cdot e $. 
\end{definition}
\begin{definition}[Descendent] Given $c \colon Y \rMap{} X$ and $y \in P(Y)$ we say that 
$x \in P(X)$ is a  \emph{$c$-descendent} of $y$ if $y = x \cdot c$.  
\end{definition}
\noindent
It is easily seen that if $x$ is a  {$c$-descendent} of $y$ then $y$ is $c$-invariant. The notion of sheaf imposes a converse. 
\begin{definition}[Atomic sheaf] 
\label{definition:atomic-sheaf}
A presheaf $P \in \Psh(\sCat{C})$ %on a coconfluent category $\sCat{C}$
 is an \emph{atomic sheaf} if, for every map 
$c \colon Y \rMap{} X$ in $\sCat{C}$, every $c$-invariant $y \in P(Y)$ has a unique $c$-descendent
$x \in P(X)$.%\footnote{We have replaced the usual terminology of \emph{matching} and \emph{amalgamation} with \emph{invariant} and \emph{descendent}, since the latter seem more appropriate in the case of atomic topologies where only a single $y$ is involved in the sheaf condition.}
\end{definition}
\noindent
We shall also have use for the following weakening of the  notion of sheaf.
\begin{definition}[Separated presheaf] 
\label{definition:separated}
A presheaf $P \in \Psh(\sCat{C})$ %on a coconfluent category $\sCat{C}$
 is an \emph{separated} (with respect to the atomic topology)  if, for every map 
$c \colon Y \rMap{} X$ in $\sCat{C}$, every $c$-invariant $y \in P(Y)$ has at most one $c$-descendent
$x \in P(X)$
\end{definition}
\begin{proposition}
\label{proposition:separated-inj}
A presheaf $P \in \Psh(\sCat{C})$ %on a coconfluent category $\sCat{C}$
 is  separated if and only
 if, for all $x,y \in P(X)$ and $q \colon Z \rMap{} X$, it holds that
 $x \cdot q = y \cdot q$ implies $x = y$. 
 \end{proposition}
 \begin{proof}
 Suppose $P$ is separated, and $x,y$ and $q$ are such that $x \cdot q = y \cdot q$. It then holds that $x \cdot q$ is $q$-invariant, and $x$ and $y$ are 
 $q$-descendents of $x \cdot q$. So, by separatedness, $x = y$.
 
 The converse implication, showing that separatedness follows from the statement in the proposition, is easy.
 \end{proof}

Propositions~\ref{proposition:NEls-sheaf} and~\ref{proposition:Yon-is-sheaf}
below illustrate the notion of sheaf in the case of $\sCat{C} = \Sur$.
\begin{proposition}
\label{proposition:NEls-sheaf}
For any set $A$ the presheaf $\NEls(A)$ in $\Psh(\Sur)$ is an atomic sheaf.
\end{proposition}
\begin{proof}
Consider any map $c \colon \Omega' \rMap{} \Omega$ in $\Sur$ and $c$-invariant $Y \in \NEls(A)(\Omega')$, i.e.,
function $Y \colon \Omega' \to A$. Define
\[
\Omega'' := \{(\omega', \omega'') \in \Omega' \times \Omega' \mid c(\omega') = c(\omega'')\} \,,
\]
and $u := (\omega', \omega'') \mapsto \omega' \colon \Omega'' \rMap{} \Omega'$ and
$v := (\omega', \omega'') \mapsto \omega'' \colon \Omega'' \rMap{} \Omega'$.
Clearly $c \circ u = c \circ v$. So, since $Y$ is $c$-invariant, $Y \circ u = Y \cdot u = Y \cdot v = Y \circ v$.
That is,  for any $(\omega' , \omega'') \in \Omega''$, we have $Y (\omega') = Y(\omega'')$; i.e., for any $\omega \in \Omega$, the function $Y$ is constant on $c^{-1}(\omega)$.  
Define $X \in \NEls(A)(\Omega)$, i.e., $X \colon \Omega(A)$ by:
\begin{equation}
\label{equation:X-def-Sur}
X(\omega) ~ := ~ Y(\omega')~\text{where $\omega' \in c^{-1}(\omega)$}\enspace .
\end{equation}
Since $c$ is surjective, this is a good definition by the constancy property remarked above. By definition, 
$Y = X \circ c = X \cdot c$, so $X$ is a $c$-descendent of $Y$.  It is the unique such, because, for any $c$-descendent $X$, the surjectivity of $c$ forces \eqref{equation:X-def-Sur}.
\end{proof}
\begin{proposition}
\label{proposition:Yon-is-sheaf}
For any finite set $\Omega$ the representable presheaf $\Yon(\Omega)$ in $\Psh(\Sur)$ is an atomic sheaf.
\end{proposition}
\noindent 
We omit the proof, which is very similar to the previous. This last proposition asserts that the atomic topology on $\Sur$ is
\emph{subcanonical}.

As a final set of examples, it is standard (and also easily verified) that if $\Sh{P_1}, \dots, \Sh{P_n}$ are sheaves
then the product presheaf $\Sh{P_1} \times \dots \times  \Sh{P_n}$ is also a sheaf, the \emph{product sheaf}. (A similar fact applies more generally to arbitrary category-theoretic limits of sheaves.) In this statement, we introduce a notational convention we shall often adopt. We shall typically use underlined names for  sheaves (as with $\NEls(A)$) in order to emphasise that they are sheaves not just presheaves. 

Assuming the small category $\sCat{C}$ is coconfluent, we write $\ShAt(\sCat{C})$ for the full subcategory 
of  atomic sheaves in $\Psh(\sCat{C})$. While the coconfluence
condition was not actually used in the definition of atomic sheaf above, it nonetheless plays a critical role.
 For the benefit of readers who know the relevant category theory, 
 we reiterate that the coconfluence condtion is equivalent to   the collection of atomic covers in $\sCat{C}$  forming a
 Grothendieck topology, which in turn means that 
   $\ShAt(\sCat{C})$ is a \emph{Grothendieck topos}, and the inclusion functor $\ShAt(\sCat{C}) \to \Psh(\sCat{C})$ has a left adjoint $\AS :  \Psh(\sCat{C}) \to \ShAt(\sCat{C})$, the \emph{associated sheaf} functor~\cite{MM}. Composing with the Yoneda functor, we obtain a functor $\AS\Yon: \sCat{C} \to \ShAt(\sCat{C})$. Because we are working with atomic topologies, every map in $\sCat{C}$ is a \emph{cover}, i.e., it is mapped by $\AS\Yon$ to an epimorphism in $\ShAt(\sCat{C})$. It thus follows from the Yoneda lemma that a necessary condition for
 every 
 representable presheaf to be a  sheaf (i.e., for the atomic topology to be subcanonical) is that all maps in $\sCat{C}$ are epimorphic.

%%%%%%%%%%%%%%%%%%%%%%%%%
%%% SECTION
%%%%%%%%%%%%%%%%%%%%%%%%%

\section{Atomic sheaf logic}
\label{section:atomic-sheaf-logic}

For the next two sections, let $\sCat{C}$ be an arbitrary  coconfluent small category. We present %, from first principles, 
a 
fragment of the internal logic of the topos $\ShAt(\sCat{C})$ of atomic sheaves, which we will extend later
with equivalence and conditional independence formulas. % in the logic.
%of Section~\ref{section:multiteams}.
The fragment we consider is simply multi-sorted first-order logic. Let $\Sort$ be a collection of sorts.
% We assume a set $\Sort$ of basic syntactic \emph{sorts} $\Syn{A}, \Syn{B}, \Syn{C}, \dots$, each
% representing a corresponding sheaf $\Sh{A}, \Sh{B}, \Sh{C}, \dots$. 
We assume a collection 
of primitive relation symbols, where each relation symbol $\Syn{R}$ has an \emph{arity} given as a finite sequence of sorts $\Arity(\Syn{R}) \in \Sort^*$. 
As in Section~\ref{section:multiteams}, variables $\Syn{x^A}$ have explicit sorts. The rules for forming atomic formulas
are: 
\begin{itemize}
\item if $\Arity(\Syn{R}) = \Syn{A}_1 \dots \Syn{A}_n$ and $\Syn{x}_1^{\Syn{A}_1}, \dots, \Syn{x}_n^{\Syn{A}_n}$ is a list of variables of the corresponding sorts, then $\Syn{R}(\Syn{x_1^{{A}_1}}, \dots, \Syn{x}_n^{\Syn{A}_n})$ is a formula; 

\item if $\Syn{x^A}, \Syn{y^A}$ have the same sort then $\Syn{x^A} = \Syn{y^A}$ is a formula.

%\item If $\Syn{X}_1^{\Syn{A}_1}, \dots, \Syn{X}_n^{\Syn{A}_n}$ and 
%$\Syn{Y}_1^{\Syn{A}_1}, \dots, \Syn{Y}_n^{\Syn{A}_n}$ are two lists of variables of the same length with the same sort list, %  $\Syn{A}_1,  \dots, \Syn{A}_n$
%then $\Syn{X}_1^{\Syn{A}_1}, \dots, \Syn{X}_n^{\Syn{A}_n} \Equiv \Syn{Y}_1^{\Syn{A}_1}, \dots, \Syn{Y}_n^{\Syn{A}_n}$ is an atomic formula.

\end{itemize}
\noindent
The grammar for formulas extends atomic formulas with the usual constructs of first-order logic. 
%(Since the  semantics will validate classical logic, the omission of implication is immaterial.)
\begin{align*}
\Phi ~ ::= ~ & \Syn{R}(\Syn{x}_1^{\Syn{A}_1}, \dots, \Syn{x}_n^{\Syn{A}_n}) \mid 
\Syn{x^A} = \Syn{y^A} 
%\\ 
% & \mid \Syn{X}_1^{\Syn{A}_1}, \dots, \Syn{X}_n^{\Syn{A}_n} \Equiv \Syn{Y}_1^{\Syn{A}_1}, \dots, \Syn{Y}_n^{\Syn{A}_n} 
\mid \neg \Phi \mid \Phi \wedge \Phi \mid \Phi \vee \Phi \mid \Phi \to \Phi \mid \exists \Syn{x^A}.\, \Phi \mid \forall \Syn{x^A}.\, \Phi \enspace .
\end{align*}
\noindent
We write $\FV(\Phi)$ for the set of free variables of a formula $\Phi$.
%Note that we have not included the equivalence and conditional  independence formulas from Section~\ref{section:multiteams} in the grammar. These constructs will be addressed in Sections~\ref{section:atomic-equivalence}--\ref{section:Schanuel}.
% is because the independence formulas do not have a canonical interpretation in a general atomic topos. Rather they are interpretable over particular sheaves (e.g., $\NEls(A)$) in particular atomic toposes (e.g., $\ShAt(\Sur)$). This will be the main topic of investigations in Sections [[WHICH]]. 

%\begin{example} 
%\label{example:relational-structure}
%The \emph{logic of a relational structure $\mathcal{A}$} has a single sort $\Syn{A}$, and every
%arity-$n$ relation symbol $\Syn{R}$ in the signature of the structure is included in our
% logic as a relation symbol $\Syn{R}$
% with  $\Arity(\Syn{R}) = \underbrace{\Syn{A}\dots\Syn{A}}_n\,$.
%\end{example}

\begin{definition}[Semantic interpretation] 
\label{definition:semantic-interpretation}
A \emph{semantic interpretation} in $\ShAt(\sCat{C})$ is given by
a function mapping every sort $\Syn{A}$ to an atomic sheaf 
$\Sh{\Syn{A}}$ (i.e., to an object of $\ShAt(\sCat{C})$), and a function mapping
every relation symbol $\Syn{R}$ of arity $\Syn{A}_1 \dots \Syn{A}_n$ to a \emph{subsheaf}
$\Sh{\Syn{R}} ~ \subseteq ~ \Sh{\Syn{A}_1} \times \dots \times \Sh{\Syn{A}_1}$.  
\end{definition}

\begin{definition}[Subpresheaf/subsheaf] For $P, Q \in \Psh(\sCat{C})$, we say that
$Q$ is a \emph{subpresheaf} of $P$ (notation $Q \subseteq P$)  if:
\begin{itemize}
\item for every object $X \in \sCat{C}$, we have $Q(X) \subseteq P(X)$, and

\item for every map $f \colon Y \rMap{} X$ in $\sCat{C}$ and element $x \in Q(X)$, it holds that
$x \cdot_Q f = x \cdot_P f$. 
\end{itemize}
For sheaves $\Sh{P}, \Sh{Q}$ with  $\Sh{Q} \subseteq \Sh{P}$, we say  $\Sh{Q}$ is a \emph{subsheaf} of
$\Sh{P}$.
\end{definition}
\noindent
The following is standard, and also easily verified. 
\begin{proposition} 
\label{proposition:subsheaf}
Given a presheaf $P \in \Psh(\sCat{C})$ and a function $Q$ mapping
every object $X \in \sCat{C}$ to a subset of $P(X)$, the function $Q$ determines a 
(necessarily unique) subpresheaf of $P$ if and only if:
\begin{itemize}
\item for every  $f \colon Y \rMap{} X$ in $\sCat{C}$ and $x \in Q(X)$, it holds that $x \cdot_P f \in Q(Y)$.
\end{itemize}
If the above holds and $P$ is also a sheaf, then the uniquely determined subpresheaf $Q$ is 
itself a sheaf %(and hence a subsheaf of $P$) 
if and only if
\begin{itemize}
\item for every  $f \colon Y \rMap{} X$ in $\sCat{C}$ and $x \in P(X)$, if $x \cdot_P f \in Q(Y)$ then $x \in Q(X)$.
\end{itemize}
\noindent
(This characterisation is valid in the  form above because we are considering only sheaves for the atomic topology.)
\end{proposition}

The  three propositions below  illustrate the notion of subsheaf. 
The first two observe that the relations of equiextension and conditional independence of nondeterministic variables (Definitions~\ref{definition:equiextension} and~\ref{definition:conditional-independence-NV}) form subsheaves, a fact which will enable us to 
extend atomic sheaf logic with equivalence and conditional-independence relations at the end of the present section.
Although the proofs are straightforward, we include them to help give 
readers who are not familiar with sheaves some feeling for the subsheaf property.

\begin{proposition} 
\label{proposition:equiextension-subsheaf}
The subsets
\[
\{(X,Y) \mid X \!\bowtie \!Y \} \subseteq (\NEls(A) \times \NEls(A))(\Omega)
\]
define a subsheaf $~\bowtie_{A} \subseteq \NEls(A) \times \NEls(A)$
 via Proposition~\ref{proposition:subsheaf}. 
\end{proposition}
\begin{proof}
For the subpresheaf property, suppose $(X,Y) \in (\NEls(A) \times \NEls(A))(\Omega)$ are such that
$(X, Y) \in {\,\bowtie_{A}\!(\Omega)}$; i.e., we have equality of images
$X(\Omega) = Y(\Omega)$. Let 
$q \colon \Omega' \rMap{} \Omega$ be a map in $\Sur$.
We need to show that $(X \cdot q, Y\cdot q) \in {\,\bowtie_{A}\!(\Omega')}$. But indeed
\[(X \cdot q)(\Omega') = X(q(\Omega')) = X(\Omega) = Y(\Omega) = Y(q(\Omega')) =(Y \cdot q)(\Omega') \enspace ,\]
where the second and fourth equalities hold because $q$ is surjective.

For the subsheaf property, suppose we have $(X,Y) \in (\NEls(A) \times \NEls(A))(\Omega)$ and 
map $q \colon \Omega' \rMap{} \Omega$ in $\Sur$ such that $(X \cdot q, Y \cdot q) \in {\,\bowtie_{A}\!(\Omega')}$.
By the definition of equiextension,  $X(q(\Omega')) = Y(q(\Omega'))$.
Because $q$ is surjective, $X(\Omega) = Y(\Omega)$. That is, 
$(X, Y) \in {\,\bowtie_{A}\!(\Omega)}$, as required by Proposition~\ref{proposition:subsheaf} to show the subsheaf property. 
\end{proof}

\begin{proposition} 
\label{proposition:indep-subsheaf}
The subsets
\[
\{(X,Y,Z) \mid X \!\Indep \!Y \!\Cond \!Z\} \subseteq (\NEls(A) \times \NEls(B) \times \NEls(C))(\Omega)
\]
define a subsheaf $\Indep_{A,B|C} \subseteq \NEls(A) \times \NEls(B) \times \NEls(C)$
 via Prop.~\ref{proposition:subsheaf}. 
\end{proposition}
\begin{proof}
We leave the subpresheaf property to the reader and verify just the subsheaf property.
Suppose we have 
$(X,Y,Z) \in (\NEls(A) \times \NEls(B) \times \NEls(C))(\Omega)$
 and 
map $q \colon \Omega' \rMap{} \Omega$ in $\Sur$ such that
$(X \cdot q ,\,Y \cdot q,  \,Z \cdot q) \in  {\Indep_{A,B|C}(\Omega')}$; i.e., 
$ X \cdot q \Indep Y \cdot q  \Cond Z \cdot q$. We need to show that 
$(X,Y,Z) \in  {\Indep_{A,B|C}(\Omega)}$; i.e., $X \Indep Y \Cond Z$.

Suppose  that there exists $\omega_1 \in \Omega$ such that
$X(\omega_1) = a$ and $Z(\omega_1) = c$, and there exists
$\omega_2 \in \Omega$ such that
$Y(\omega_2) = b$ and $Z(\omega_2) = c$. Using the surjectivity of $q$, let
$\omega_1', \omega_2' \in \Omega'$ be such that $q(\omega_1') = \omega_1$ and
$q(\omega_2') = \omega_2$. Then $(X \cdot q)(\omega_1') = a$ and 
$(Z \cdot q)(\omega_1') = c$. Similarly
$(Y \cdot q)(\omega_2') = b$ and 
$(Z \cdot q)(\omega_2') = c$.
Because $ X \cdot q \Indep Y \cdot q  \Cond Z \cdot q$,
there exists $\omega' \in \Omega'$ such that
$(X \cdot q)(\omega') = a$ and $(Y \cdot q)(\omega') = b$
and $(Z \cdot q)(\omega') = c$. So $\omega := q(\omega')$ satisfies
$X(\omega) = a$ and 
$Y(\omega) = b$ and $Z(\omega) = c$, showing that indeed
$X \Indep Y \Cond Z$.
\end{proof}

%\noindent
%We omit the straightforward proofs of Propositions~\ref{proposition:equiextension-subsheaf}
%and~\ref{proposition:indep-subsheaf}.

As further interesting examples of subsheaves, we  show how 
subsheaves of the  sheaf $\NEls(A)$  can be defined by using modalities to lift 
 properties $P \subseteq A$ to properties of $A$-valued nondeterministic variables. 
\begin{proposition} 
\label{proposition:R-subsheaf}
For any set $A$ and subset $P \subseteq A$, the definitions
\begin{align*}
\Sh{\Box P}(\Omega)~ := ~ & 
\{ X :  \Omega \to A \mid \forall \omega \in \Omega.\,  X(\omega) \in P \} \\
\Sh{\Diamond P}(\Omega)~ := ~ & 
\{ X :  \Omega \to A \mid \exists \omega \in \Omega.\,  X(\omega) \in P \} 
\end{align*}
define subsheaves
$\Sh{\Box P}$ and $ \Sh{\Diamond P}$ of $\NEls(A)$ in $\ShAt(\Sur)$, by Proposition~\ref{proposition:subsheaf}.
%
%We use the above to give a standard interpretation in $\ShAt(\Sur)$ to the logic,
%defined in  Example~\ref{example:relational-structure},
%of the relational structure $\mathcal{A}$. The single sort $\Syn{A}$ 
%is interpreted as the sheaf
%$\Sh{A} := \NEls(A)$; and each arity-$n$ relational symbol $\Syn{R}$ 
% is interpreted as the
%subsheaf $\Sh{R} := \Sh{\Syn{R}^{\mathcal{A}}} \subseteq (\NEls(A))^n$.
\end{proposition}
\noindent
This time we omit the proof, since the modality subsheaves will not play any further role in the paper. We mention, however, that 
the constructions in Proposition~\ref{proposition:R-subsheaf} can be used as the basis for an interesting modal extension of the first-order atomic sheaf logic of $\ShAt(\Sur)$, in which the modalities mediate between the ordinary first-order logic of variables valued in  $A$ and 
the sheaf logic of nondeterministic variables valued in $\NEls(A)$.

\begin{figure}[t]
\begin{align*}
X \Forces_{\!\underline\rho}  \Syn{R}(\Syn{x}_1^{\Syn{A}_1}, \dots, \Syn{x}_n^{\Syn{A}_n}) ~ \Iff ~ & ~
  (\underline\rho(\Syn{x}_1^{\Syn{A}_1}), \dots, \underline\rho(\Syn{x}_n^{\Syn{A}_n}))  \in \Sh{\Syn{R}}(X)
 \\
X \Forces_{\!\underline\rho}  \Syn{x^A}  = \Syn{y^A} ~ \Iff ~ & ~ \underline\rho(\Syn{x^A})=\underline\rho(\Syn{y^A}) 
\\
%X \Forces_{\!\underline{\rho}}   ~ 
%  \Syn{x}_1^{\Syn{A}_1}, \dots, \Syn{x}_n^{\Syn{A}_n}  \Equiv \Syn{y}_1^{\Syn{A}_1}, \dots, \Syn{y}_n^{\Syn{A}_n}
%  ~ \Iff ~ & ~ 
%   (\,  (\underline\rho(\Syn{x}_1^{\Syn{A}_1}), \dots, \underline\rho(\Syn{x}_n^{\Syn{A}_n})) \, , \, 
%     (\underline\rho(\Syn{y}_1^{\Syn{A}_1}), \dots, \underline\rho(\Syn{y}_n^{\Syn{A}_n}))\,  )
%  \,  \in \, \, {{\Equiv_{\Sh{A_1} \times \dots \times \Sh{A_n}}}(X)}
%\\
X \Forces_{\!\underline\rho} \neg \Phi  ~  \Iff  ~ & ~
    X \not\Forces_{\!\underline\rho} \Phi
\\
X \Forces_{\!\underline\rho} \Phi \wedge \Psi  ~  \Iff  ~ & ~
    X \Forces_{\!\underline\rho} \Phi   ~\text{and}~
    X \Forces_{\!\underline\rho} \Psi
\\
X \Forces_{\!\underline\rho} \Phi \vee \Psi  ~  \Iff  ~ & ~
    X \Forces_{\!\underline\rho} \Phi   ~\text{or}~
    X \Forces_{\!\underline\rho} \Psi
\\
X \Forces_{\!\underline\rho} \Phi \to \Psi  ~  \Iff  ~ & ~
    X \not\Forces_{\!\underline\rho} \Phi   ~\text{or}~
    X \Forces_{\!\underline\rho} \Psi
\\
X \Forces_{\!\underline\rho} \exists \Syn{x^A}.\, \Phi   ~  \Iff  ~ & ~ \exists Y. \;  \exists f \colon Y \rMap{} X. \,  
%\\
%& \qquad 
\exists x  \in \Sh{A}(Y). ~Y  \Forces_{(\underline{\rho} \,\cdot f){[\Syn{x^A}:=x]}}  \Phi  
%  ~ \text{where $\underline{\rho}' :=  (\Syn{z^B} \mapsto \underline{\rho}(\Syn{z^B}) \cdot f$)}
\\
X \Forces_{\!\underline\rho} \forall \Syn{x^A}.\, \Phi   ~  \Iff  ~ & ~ \forall Y. \; \forall f \colon Y \rMap{} X. \,  
%\\
%& \qquad
\forall x  \in \Sh{A}(Y). ~Y  \Forces_{(\underline{\rho}\, \cdot f){[\Syn{x^A}:=x]}}  \Phi   
% ~ \text{where $\underline{\rho}' :=  (\Syn{z^B} \mapsto \underline{\rho}(\Syn{z^B}) \cdot f$)}
%\\[3ex]
%\uncover<1->{\Omega \Forces_{\!\underline\rho}  \neg \Phi  ~  \Iff  & ~~ %\text{not}~
%    \Omega \not\Forces_{\!\underline\rho} \Phi 
%    }
\end{align*}
\caption{Semantics of atomic sheaf logic}
\label{figure:atomic-sheaf-semantics}
\end{figure}

%We can now finally define the semantics of our logic. Recall, we have a collection $\Sort$ of sorts $\Syn{A}$ each with associated sheaf $\Sh{A}$, and 
%a collection of relation symbols, for which each $\Syn{R}$ of arity $\Syn{A}_1 \dots \Syn{A}_n$
%has an associated subsheaf $\Sh{R} \subseteq  \Sh{A_1} \times \dots \times \Sh{A_1}$.

Returning to the general semantic interpretation of atomic sheaf logic in $\ShAt(\sCat{C})$, the semantics of formulas is given by  a forcing relation
\[
X \Forces_{\Sh{\rho}} \Phi \enspace ,
\]
where $\Phi$ is a formula, $X$ is an object of $\sCat{C}$ and 
\[
\Sh{\rho} ~ \in ~ \prod_{\Syn{x}^\Syn{A} \in \{\Syn{x}_1^{\Syn{A}_1}, \dots, \Syn{x}_n^{\Syn{A}_n}\}}  \Sh{A}(X)
\]
 is what we call an \emph{$X$-assignment}: it maps every variable $\Syn{x}^\Syn{A}$
 in a set $\{\Syn{x}_1^{\Syn{A}_1}, \dots, \Syn{x}_n^{\Syn{A}_n}\} \supseteq \FV(\Phi)$ %containing every free variable of $\Phi$ 
 to an element $\Sh{\rho}(\Syn{x}^\Syn{A}) \in \Sh{\Syn{A}}(X)$, where $\Sh{\Syn{A}}$ is the sheaf interpreting the sort $\Syn{A}$ of the variable. 
 
The definition of the forcing relation is presented in Fig.~\ref{figure:atomic-sheaf-semantics}. 
In the quantifier clauses, 
% we use the following manipulation of assignments. Given an
% $X$-assignment $\Sh{\rho}$ and $f \colon Y \rTo X$ in $\sCat{C}$,
we write write $\Sh{\rho} \,\cdot f$ for the $Y$-assignment $\Syn{z^B} \mapsto \Sh{\rho}(\Syn{z^B}) \cdot f$,
where $\Sh{\rho}$ is an $X$-assignment and  $f \colon Y \rMap{} X$ is a map in $\sCat{C}$,

The clauses for the propositional connectives in Fig.~\ref{figure:atomic-sheaf-semantics} 
look remarkably simple-minded. They are, nonetheless,  equivalent to the more involved clauses that appear  in the \emph{sheaf semantics}  for  logic in a sheaf topos~\cite{MM}. The simplification in formulation is possible because we are working in the special case of atomic sheaves. The clauses for the existential and universal quantifier are also taken from sheaf semantics,  and do not admit further simplification. Their non-local nature (they involve a change of world along $f \colon Y \to X$) is the key feature that will give atomic sheaf logic its character, when we later include equivalence and conditional independence formulas.
% in Sections~\ref{section:atomic-equivalence}--\ref{section:Schanuel}.

The next results summarise fundamental  properties of the forcing relation and the logic it induces. The first is very basic, but we include it explicitly because the notion of \emph{locality} it addresses, namely
the dependency of  semantics only on assignments  to the free variables appearing in a formula, 
has been a delicate issue in the context of (in)dependence logics. 

\begin{proposition}[Locality] 
\label{proposition:locality}
For any  formula $\Phi$, object $X$ of  $\sCat{C}$ and
$X$-assignments $\Sh{\rho}, \Sh{\rho}'$ that are defined and coincide on $\FV(\Phi)$.
%all free variables of $\Phi$, 
\[
\text{$X \Forces_{\Sh{\rho}} \Phi$ if and only if $X \Forces_{\Sh{\rho}'} \Phi$.}
\]
\end{proposition}

\begin{proposition}[Sheaf property]
\label{proposition:forcing-sheaf}
For any formula $\Phi$, map 
$f \colon Y \rMap{} X$ in $\sCat{C}$,  and
$X$-assignment $\Sh{\rho}$ defined on $\FV(\Phi)$. %all free variables of $\Phi$, 
\begin{equation}
\label{equation:sheaf-prop}
\text{$X \Forces_{\Sh{\rho}} \Phi$ if and only if $Y \Forces_{\Sh{\rho} \, \cdot f} \Phi$.}
\end{equation}
\end{proposition}

\noindent
Proposition~\ref{proposition:forcing-sheaf} is called the sheaf property because it 
is equivalent to the statement that, for every formula $\Phi$ with
$\FV(\Phi) \subseteq \{\Syn{x}_1^{\Syn{A}_1}, \dots , \Syn{x}_n^{\Syn{A}_n}\}$, it holds that
\begin{equation}
\label{equation:Phi-subsheaf}
\{(x_1, \dots, x_n) \mid X \Forces _{\Syn{x}_i^{\Syn{A}_i} \mapsto x_i} \Phi \} ~ 
\subseteq  ~ (\Sh{A_1} \times \dots \times \Sh{A_n})(X) 
\end{equation}
defines a subsheaf of $\Sh{A_1} \times \dots \!\times \Sh{A_n}$ via Proposition~\ref{proposition:subsheaf}.

Propositions~\ref{proposition:locality} 
and~\ref{proposition:forcing-sheaf} are both proved by induction on the structure of the formula.
We omit the proof of Proposition~\ref{proposition:locality}, which is straightforward.
Proposition~\ref{proposition:forcing-sheaf}
asserts that 
the \emph{monotonicity} and \emph{local character} properties from~\cite[{\S}VI.7]{MM} hold.
In \emph{loc.\ cit.}, these properties are shown to hold for arbitrary Grothendiek topologies, whereas
Proposition~\ref{proposition:forcing-sheaf} concerns just the special case of atomic topologies. 
Nevertheless, we give a direct proof of Proposition~\ref{proposition:forcing-sheaf}, both for the benefit of readers who do not know general sheaf theory, and also to demonstrate the crucial role played by the coconfluence property of $\sCat{C}$.
\begin{proof}[Proof of Proposition~\ref{proposition:forcing-sheaf}]
By induction on the structure of $\Phi$.

In the case that $\Phi$ is an atomic formula of the form $\Syn{R}(\Syn{x}_1^{\Syn{A}_1}, \dots, \Syn{x}_n^{\Syn{A}_n})$,
property~\eqref{equation:sheaf-prop} holds because 
$\Sh{\Syn{R}}$ is a subsheaf of $\Sh{\Syn{A}_1} \times \dots \times \Sh{\Syn{A}_1}$.

If $\Phi$ is an equality $\Syn{x^A}  = \Syn{y^A}$, then the left-to-right implication of~\eqref{equation:sheaf-prop} is immediate. For the right-to-left implication, suppose 
$Y \Forces_{\Sh{\rho} \, \cdot f} \Syn{x^A}  = \Syn{y^A}$; that is, $\Sh{\rho}(\Syn{x^A}) \cdot f = \Sh{\rho}(\Syn{y^A}) \cdot f$.
Since $\Sh{\Syn{A}}$ is a sheaf, hence separated, we have
$\Sh{\rho}(\Syn{x^A}) = \Sh{\rho}(\Syn{y^A})$, by Proposition~\ref{proposition:separated-inj}. That is, 
$X \Forces_{\Sh{\rho}} \Syn{x^A}  = \Syn{y^A}$, as required.

The cases for the propositional connectives are all easy. We note only that, for the cases of negation and implication, in which there are negated clauses in the definition of the forcing relation (Fig.~\ref{figure:atomic-sheaf-semantics}), the induction hypothesis is used in the opposite direction of~\eqref{equation:sheaf-prop} to the implication being proved. 

In the case that $\Phi$ is an existentially quantified formula $\exists \Syn{x^A}.\, \Phi'$, we prove the left-to-right implication of~\eqref{equation:sheaf-prop}. Accordingly, suppose that 
$X \Forces_{\!\underline\rho} \exists \Syn{x^A}.\, \Phi'$. By the forcing clause for the existential quantifier,  there exist  $g: Z \rMap{} X$  and 
$x  \in \Sh{A}(Z)$ such that $Z  \Forces_{(\underline{\rho} \,\cdot g){[\Syn{x^A}:=x]}}  \Phi'$. By coconfluence, there exists a span $Y \lMap{f'} W \rMap{g'} Z$ such that $g \circ g' = f \circ f'$. By the induction hypothesis,
$W  \Forces_{(\underline{\rho} \,\cdot g \cdot g'){[\Syn{x^A}:=x\cdot g']}}  \Phi'$; i.e.,
$W  \Forces_{(\underline{\rho} \,\cdot f \cdot f'){[\Syn{x^A}:=x\cdot g']}}  \Phi'$. Whence, by the forcing clause for the existential quantifier, $Y \Forces_{\!\underline\rho\,\cdot f} \exists \Syn{x^A}.\, \Phi'$, as required.
We leave the easier right-to-left implication of~\eqref{equation:sheaf-prop}, which does not involve coconfluence, to the reader.

The proof for the universal quantifier is similar. (It can also be bypassed, by noting that the forcing interpretation of
$\exists \Syn{x^A}.\, \Phi'$ is equivalent to that for $\neg \exists \Syn{x^A}.\, \neg\Phi'$.)
\end{proof}

It is standard that sheaf semantics, for an arbitrary Grothendieck topology, always validates intuitionistic logic. 
In the special case of an atomic topology, the \emph{law of excluded middle} $\Phi \vee \neg \Phi$ is also validated, hence atomic sheaf logic is classical. In more detail, atomic topologies are special cases of dense Grothendieck topologies, and categories of sheaves for the latter are always boolean, hence classical logic is validated. This whole picture is explained in~\cite{MM}. 
We shall not, however, assume familiarity with this abstract  picture. Accordingly, we give a brief, direct explanation of how atomic sheaf logic
validates classical logic.

A formula $\Phi$ is said to be \emph{true (in $\ShAt(\sCat{C})$) under all assignments} if, for every object $X$ of  $\sCat{C}$ and
$X$-assignments $\Sh{\rho}$ defined on $\FV(\Phi$),%all free variables of $\Phi$,
it holds that $X \Forces_{\Sh{\rho}} \Phi$.
\begin{theorem}[Classical logic]
\label{theorem:classical}
If $\Phi$ is a theorem of (multisorted) classical logic then it is true in $\ShAt(\sCat{C})$ under all assignments.
\end{theorem}
\begin{proof}[Proof (outline)]
It follows trivially from the  definition of the forcing relation Fig.~\ref{figure:atomic-sheaf-semantics} that
every classical propositional tautology (including every  instance of the 
law of excluded middle $\Phi \vee \neg \Phi$) is true under all assignments (assuming, as we do, that we are working in a classical meta-theory).

The verification of the validity of the axioms and inference rules pertaining to quantifiers takes a little more work, but is not difficult. Since we are working in a special case of sheaf semantics, where such facts are anyway well established in far greater generality, we do not go into details. A sceptical reader may enjoy verifying this for themselves, using their preferred 
formulation of the axioms and rules of multi-sorted first-order logic.
\end{proof}

By Theorem~\ref{theorem:classical}, atomic sheaf logic is just multisorted first-order classical logic with a nonstandard semantics. 
The logic includes the equality relation, which is given a canonical interpretation. 
The nonstandard semantics allows 
relation symbols to be interpreted as arbitrary subsheaves of product sheaves. Atomic sheaf categories possess interesting such subsheaves that have no analogue in the standard semantics of first-order logic. Our main examples  of this phenomenon are
the two relations from the title: equivalence and conditional independence. 

To end this section, we observe that, in the case of our running example $\ShAt(\Sur)$, atomic sheaf logic  can incorporate the relations of equivalence and conditional independence from multiteam semantics, as in Section~\ref{section:multiteams}.
Syntactically, we simply extend the logic with equivalence and conditional independence formulas~\eqref{equation:equiv-formula} 
and~\eqref{equation:indep-formula}, as in Section~\ref{section:multiteams}. 
Actually, we can do this simply by including equivalence and conditional independence as particular relation symbols,
so the equivalence and conditional independence formulas are then instances of atomic formulas of the form
$\Syn{R}(\Syn{x}_1^{\Syn{A}_1}, \dots, \Syn{x}_n^{\Syn{A}_n})$.
Specifically, for equivalence, we include relation symbols
$\Equiv_{ \Syn{A}_1\dots \Syn{A}_n}$ with $\Arity(\Equiv_{ \Syn{A}_1\dots \Syn{A}_n}) = \Syn{A}_1\dots \Syn{A}_n \Syn{A}_1\dots \Syn{A}_n$.
Similarly, for conditional independence, 
we include relation symbols $\sIndep_{\Syn{A}_1\dots \Syn{A}_l, \Syn{B}_1\dots \Syn{B}_m |  \Syn{C}_1\dots \Syn{C}_n}$
with $\Arity(\sIndep_{\Syn{A}_1\dots \Syn{A}_l,\Syn{B}_1\dots \Syn{B}_m |  \Syn{C}_1\dots \Syn{C}_n}) = 
\Syn{A}_1\dots \Syn{A}_l \Syn{B}_1\dots \Syn{B}_m \Syn{C}_1\dots \Syn{C}_n$.
 
To interpret the extended logical language, we instantiate the semantic interpretation of Definition~\ref{definition:semantic-interpretation},
in the special case of the category $\ShAt(\Sur)$,
by requiring that every sort $\Syn{A}$ is interpreted by a sheaf of nondeterministic
variables $\Sh{\Syn{A}} := \NEls(\Sem{\Syn{A}})$ for some set $\Sem{\Syn{A}}$. 
We then interpret each relation
$\Equiv_{ \Syn{A}_1\dots \Syn{A}_n}$ 
as the subsheaf 
of $\left(\prod_{i=1}^n \Sh{\Syn{A}_i } \right) \times \left(\prod_{i=1}^n \Sh{\Syn{A}_i}\right)$
that is isomorphic to the subsheaf 
$\bowtie_{(\prod_{i=1}^n \Sem{\Syn{A}_i})}$ of 
$\NEls(\prod_{i=1}^n \Sem{{A}_i}) \times \NEls(\prod_{i=1}^n \Sem{\Syn{A}_i}))$ from
Proposition~\ref{proposition:equiextension-subsheaf}
along the canonical isomorphism between the two product sheaves.
A similar procedure, using $\Indep_{A,B|C}$ from Proposition~\ref{proposition:indep-subsheaf}, defines the semantics of conditional independence formulas as subsheaves.
These rather convoluted definitions are equivalent to simply 
interpreting equivalence and conditional independence formulas directly using the 
conditions  given in Figure~\ref{figure:atomic-semantics}. 
The benefit of the convoluted explanation in terms of subsheaves is that it
presents the extended logic as a special case of general atomic sheaf logic, and in doing so
explains why the meta-logical properties (locality, sheaf property, classical logic)  hold for the extended logic.

%%%%%%%%%%%%%%%%%%%%%%%%%%%%%%%
%%% SECTION
%%%%%%%%%%%%%%%%%%%%%%%%%%%%%%%%

\section{Atomic equivalence}
\label{section:atomic-equivalence}

The interpretation of equivalence formulas at the end of Section~\ref{section:atomic-sheaf-logic} was given only 
for the relation of  equiextension of  nondeterministic variables, interpreted over the sheaves of nondeterministic variables 
in $\ShAt(\Sur)$ using Proposition~\ref{proposition:equiextension-subsheaf}.

Atomic sheaves offer, however, a much more  general perspective on the notion of equivalence. Every category $\ShAt(\sCat{C})$ of atomic sheaves possesses a canonical notion of equivalence, which we call \emph{atomic equivalence}. 
Specifically, for every sheaf $\Sh{P}$, there is an associated subsheaf
${\Equiv_{\Sh{P}}} \subseteq \Sh{P} \times \Sh{P}$ that is an equivalence relation in $\ShAt(\sCat{C})$.
(A subsheaf $\Sh{E} \subseteq \Sh{P} \times \Sh{P}$ is an \emph{equivalence relation in 
$\ShAt(\sCat{C})$} if $\Sh{E}(X) \subseteq \Sh{P}(X) \times \Sh{P}(X)$ is an equivalence relation, for every $X \in \sCat{C}$.)

\begin{theorem}[Atomic equivalence]
\label{theorem:atomic-equiv}
Let $\Sh{P}$ be any sheaf in $\ShAt(\sCat{C})$.
\begin{align*}
\Equiv_{\Sh{P}}(X) ~ : = ~  & 
\{(x,x') \in \Sh{P}(X) \times \Sh{P}(X) \mid 
%\\
%& \qquad
\exists Z, \exists u, u'  \colon Z \rMap{} X. ~ x \cdot u = x' \cdot u' \}
\end{align*}
defines a subsheaf ${\Equiv_{\Sh{P}}} \subseteq \Sh{P} \times \Sh{P}$ via Proposition~\ref{proposition:subsheaf}.
Moreover, this is an equivalence relation in $\ShAt(\sCat{C})$.
\end{theorem}
\noindent
\begin{proof} For the subpresheaf property, suppose  $(x,x') \in {\Equiv_{\Sh{P}}}\,(X)$. Thus, for some $u,u' \colon Z \rMap{} X$, we have $x \cdot u = x' \cdot u'$. Consider any $f \colon Y \rMap{} X$. By coconfluence, 
there exist $g \colon W \rMap{} Z$ and $v \colon W \rMap{} Y$ such that $f \circ  v= u \circ g$. 
Similarly, there exist 
$g' \colon W' \rMap{} Z$ and $v' \colon W' \rMap{} Y$ such that $f \circ v' = u' \circ g'$.
Again by coconfluence, 
there exist $h \colon V \rMap{} W$ and $h' \colon V \rMap{} W'$ such that $g \circ h = g' \circ h'$. Then:
\[x \cdot f \cdot v \cdot h = x \cdot u \cdot g \cdot h = x' \cdot u' \cdot g' \cdot h' = 
x' \cdot f \cdot v' \cdot h'  \enspace .\]
So $v \circ h$ and $v' \circ h' \colon V \rMap{} Y$ show that  $(x \cdot f , x' \cdot f) \in {{\Equiv_{\Sh{P}}}(Y)}$.

For the subsheaf property, consider any $(x, x') \in \Sh{P}(X) \times \Sh{P}(X)$ and $f \colon Y \rMap{} X$ such that
$(x \cdot f, x' \cdot f) \in {\Equiv_{\Sh{P}}\,(Y)}$; i.e., there exist $u,u' \colon Z \rMap{} Y$ such that
$x \cdot f \cdot u =  x' \cdot f \cdot u'$. Thus $f \circ u$ and $f \circ u'  \colon Z \rMap{} X$ show that indeed 
$(x,x') \in {{\Equiv_{\Sh{P}}}(X)}$.

For the equivalence relation property, reflexivity and symmetry are trivial. For transitivity, suppose
$(x,x') \in {\Equiv_{\Sh{P}}}\,(X)$ and $(x',x'') \in {\Equiv_{\Sh{P}}}\,(X)$; i.e., there exist
$u,u' \colon Z \rMap{} X$ such that $x \cdot u = x' \cdot u'$ and 
$v, v' \colon Z' \rMap{} X$ such that $x' \cdot v = x'' \cdot v'$. By coconfluence, there exist
$w \colon W \rMap{} Z$ and $w' \colon W \rMap{} Z'$ such that $u' \circ w = v \circ w'$. Then
\[x \cdot u \cdot w = x' \cdot  u' \cdot w = x' \cdot v \cdot w' = x'' \cdot v' \cdot w'  \enspace .\]
So $u \circ w$ and $v' \circ w'$  show that $(x,x'') \in {\Equiv_{\Sh{P}}}\,(X)$.
\end{proof}

In the special case of sheaves $\NEls(A)$ of nondeterministic variables in $\ShAt(\Sur)$,
the canonical equivalence ${\Equiv_{\NEls(A)}}$ coincides with the equiextension subsheaf $\bowtie_{A}$
defined in Proposition~\ref{proposition:equiextension-subsheaf}.

\begin{proposition} 
\label{proposition:equiv-NEls}
The subsheaf ${\Equiv_{\NEls(A)}} \subseteq \NEls(A)\times \NEls(A)$
in  $\ShAt(\Sur)$ coincides with $\bowtie_A \subseteq \NEls(A)\times \NEls(A)$.
%\[
%\Equiv_{\NEls(A)}(\Omega) ~ := ~
%\{(X,X') \in (\Omega \to A)^2 \mid X(\Omega) = X'(\Omega) \} \enspace .
%\]
\end{proposition}

\begin{proof}
Consider any $X, X' \colon \Omega \to A$.

Suppose there exist $u, u' \colon \Omega' \rMap{} \Omega$ such that
$X \cdot u = X' \cdot u'$; i.e.,  $X \circ u = X' \circ u'$. Then $X \bowtie X'$ because
\[ X(\Omega) ~ = ~ X(u(\Omega')) ~ = ~ X'(u'(\Omega')) ~ = ~ X'(\Omega) \enspace ,\]
using the surjectivity of $u$ and $u'$ for the first and last equalities.

Conversely, suppose $X \bowtie X'$, i.e., $X(\Omega) = X'(\Omega)$. Define $\Omega_A := X(\Omega)$, which is 
a finite nonempty set hence (up to isomorphism) an object of $\Sur$. The functions 
$X$ and $X'$ are surjective from $\Omega$ to $\Omega_A$, hence give morphisms
$X,X' \colon \Omega  \rMap{} \Omega_A$ in $\Sur$. By coconfluence, there exist maps $p,q \colon \Omega' \rMap{} \Omega$ such that $X \circ p = X' \circ q$. But this means that $X \cdot p = X' \cdot q$, hence $(X,X') \in {\Equiv_{\NEls(A)}(\Omega)}\,$.
\end{proof}

Using the notion of atomic equivalence, we give a canonical semantics 
to equivalence formulas~\eqref{equation:equiv-formula} in any atomic sheaf topos.
As at the end of Section~\ref{section:atomic-sheaf-logic}, we include such formulas by considering them as given by
relation symbols $\Equiv_{ \Syn{A}_1\dots \Syn{A}_n}$ with $\Arity(\Equiv_{ \Syn{A}_1\dots \Syn{A}_n}) = \Syn{A}_1\dots \Syn{A}_n \Syn{A}_1\dots \Syn{A}_n$. The general semantic interpretation of sorts and relations  (Definition~\ref{definition:semantic-interpretation})
is then extended to require that each relation symbol $\Equiv_{ \Syn{A}_1\dots \Syn{A}_n}$ is
interpreted as  the subsheaf
\[
\Sh{\Equiv_{ \Syn{A}_1\dots \Syn{A}_n}} ~ := ~ 
   \Equiv_{\Sh{\Syn{A}_1} \times \dots \times \Sh{\Syn{A}_n}} ~ 
    \subseteq ~ (\Sh{\Syn{A}_1} \times \dots \times \Sh{\Syn{A}_n}) \times  (\Sh{\Syn{A}_1} \times \dots \times \Sh{\Syn{A}_n}) \enspace .
\]
The forcing relation $X \Forces_{\!\underline{\rho}}   ~ 
  \Syn{x}_1^{\Syn{A}_1}, \dots, \Syn{x}_n^{\Syn{A}_n}  \Equiv \Syn{y}_1^{\Syn{A}_1}, \dots, \Syn{y}_n^{\Syn{A}_n}$
 is then covered by the general clause for relation symbols $\Syn{R}$ in Figure~\ref{figure:atomic-sheaf-semantics}. This is equivalent to defining:\begin{align*}
X & \Forces_{\!\underline{\rho}}  ~
  \Syn{x}_1^{\Syn{A}_1}, \dots, \Syn{x}_n^{\Syn{A}_n}  \Equiv \Syn{y}_1^{\Syn{A}_1}, \dots, \Syn{y}_n^{\Syn{A}_n}
  ~ \Iff ~
  %\\ &  
  (\,(\underline\rho(\Syn{x}_1^{\Syn{A}_1}),\! \dots, \!\underline\rho(\Syn{x}_n^{\Syn{A}_n})) \, , \, 
     (\underline\rho(\Syn{y}_1^{\Syn{A}_1}), \!\dots, \!\underline\rho(\Syn{y}_n^{\Syn{A}_n})) \,) 
~ \in ~{{\Equiv_{\Sh{A_1} \times \dots \times \Sh{A_n}}}(X)}.
\end{align*}
\noindent
By Proposition~\ref{proposition:equiv-NEls}, the above definition generalises the multiteam interpretation 
of independence as the equiextension relation, in the case $\sCat{C} = \Sur$ and 
$\Sh{\Syn{A}} = \NEls(\Sem{\Syn{A}})$,  that was given in Section~\ref{section:atomic-sheaf-logic}.

\begin{figure}
\begin{gather}
\vec{\Syn{x}} \Equiv \vec{\Syn{x}}
\label{equiv:A}
\\
\vec{\Syn{x}} \Equiv \vec{\Syn{y}}  ~ \to ~ \vec{\Syn{y}} \Equiv \vec{\Syn{x}}  
\label{equiv:B}
\\
\vec{\Syn{x}} \Equiv \vec{\Syn{y}}  ~  \wedge  ~ \vec{\Syn{y}} \Equiv \vec{\Syn{z}}   ~\to ~ 
\vec{\Syn{x}} \Equiv \vec{\Syn{z}} 
\label{equiv:C}
\\[1ex]
\vec{\Syn{x}} \Equiv \vec{\Syn{y}}  ~ \to ~ \pi(\vec{\Syn{x}}) \Equiv \pi(\vec{\Syn{y}}) %~~~ \text{($\pi$ a permutation)}
\label{equiv:D}
\\
\vec{\Syn{x}}, \Syn{x} \Equiv \vec{\Syn{y}}, \Syn{y}  ~ \to ~\vec{\Syn{x}} \Equiv \vec{\Syn{y}} 
\label{equiv:E}
\\[1ex]
\vec{\Syn{x}} \Equiv \vec{\Syn{y}} ~ \wedge ~ \Phi(\vec{\Syn{x}}) ~ \to ~ \Phi(\vec{\Syn{y}})
\label{equiv:F}
% ~~~\text{($\FV(\Phi) \subseteq \{\vec{\Syn{x}}\}$)}
\\
\vec{\Syn{x}} \Equiv \vec{\Syn{x}'} ~ \to ~ \exists \Syn{y}'. ~ 
(\vec{\Syn{x}},\Syn{y}  \Equiv \vec{\Syn{x}'},\Syn{y}')
\label{equiv:G}
\end{gather}
\caption{Axioms for equivalence}
\label{figure:equivalence-axioms}
\end{figure}

We now explore the logic of atomic equivalence, valid in any category of atomic sheaves. 
Fig.~\ref{figure:equivalence-axioms} lists formulas that are valid in our semantics, which we identify as 
axioms for equivalence. In them, we have abbreviated variable sequences by vectors. It is implicitly assumed that the lengths and sorts of the variable sequences match so that the equivalence formulas are legitimate.
Axioms~\eqref{equiv:A}--\eqref{equiv:C} simply state that $\Equiv$ is an equivalence relation. The next two assert structural properties. 
In~\eqref{equiv:D}, $\pi$ is any permutation of the variable sequence, and the axiom asserts that
equivalence is preserved if one permutes variables in the same way on both sides.
By axiom~\eqref{equiv:E}, equivalence is also preserved if one drops identically positioned variables from both sides.
Axiom~\eqref{equiv:F} is more interesting: 
equivalence enjoys a substitutivity property, similar to the substitutivity property of equality.
However, an important restriction is hidden in the notation. It is assumed  that all free variables in $\Phi$ are contained
in a sequence $\vec{\Syn{z}}$ of \emph{distinct} variables matching in length and sorting with $\vec{\Syn{x}}$, and hence also with $\vec{\Syn{y}}$. We then write $\Phi(\vec{\Syn{x}})$ for the substitution $\Phi(\vec{\Syn{x}}/ \vec{\Syn{z}})$,
and similarly for $\Phi(\vec{\Syn{y}})$.
We call \eqref{equiv:F}  the \emph{invariance principle}, as it states that 
properties not  involving extraneous variables are invariant under equivalence.
Axiom \eqref{equiv:G} is called the \emph{transfer principle}. If $\vec{\Syn{x}}$ and $\vec{\Syn{x}'}$ are jointly equivalent, then for any variable $\Syn{y}$ there exists a (necessarily equivalent) variable $\Syn{y}'$ 
such that $\vec{\Syn{x}}, \Syn{y}$ and $\vec{\Syn{x}'}, \Syn{y'}$ are jointly equivalent.

This soundness of axioms~\eqref{equiv:A} to~\eqref{equiv:E} is straightforward.
The soundness of 
the invariance principle \eqref{equiv:F} is a consequence of the following simple lemma.
\begin{lemma}
For any $\Sh{P} \in \ShAt(\sCat{C})$ with subsheaf $\Sh{Q} \subseteq \Sh{P}$. If
$x,x' \in \Sh{P}(X)$ are such that $(x,x') \in {\Equiv_{\Sh{P}} (X)}$ and $x \in \Sh{Q}(X)$ then $x' \in \Sh{Q}(X)$.
\end{lemma}
\begin{proof}
Because $(x,x') \in {\Equiv_{\Sh{P}} (X)}$, we have that there exist $u, u' \colon Y \rMap{} X$ such that
$x \cdot u = x' \cdot u'$. As $x \in \Sh{Q}(X)$ and $\Sh{Q}$ is a subpresheaf, we have $x \cdot u \in \Sh{Q}(Y)$, that is $x' \cdot u' \in \Sh{Q}(Y)$. Hence, since $\Sh{Q}$ is a subsheaf, $x' \in \Sh{Q}(X)$.
\end{proof}
\noindent
The invariance principle follows from the lemma, because $\Phi$ defines a subsheaf of $\Sh{\Syn{A}_1} \times \dots \times \Sh{\Syn{A}_n}$ 
via~\eqref{equation:Phi-subsheaf}, where $\Syn{A}_1, \dots, \Syn{A}_n$ are the sorts of the vector
$\vec{\Syn{x}} = \Syn{x}_1^{\Syn{A}_1}, \dots, \Syn{x}_n^{\Syn{A}_n}$ (and hence also of $\vec{\Syn{y}}$) in \eqref{equiv:F}.

The soundness of the transfer principle \eqref{equiv:G} is a consequence of the lemma below.
\begin{lemma}
Let $\Sh{P}, \Sh{Q}$ be sheaves and let $x,x' \in P(X)$ such that $(x,x') \in {\Equiv_{\Sh{P}} (X)}$.
For any $y \in P(X)$, there exists $p \colon Z \rMap{} X$ and $y' \in \Sh{Q}(Z)$ such that
$((x \cdot p, y \cdot p), (x' \cdot p, y')) \in\, {\Equiv_{\Sh{P} \times \Sh{Q}} (Z)}$.
\end{lemma}
\begin{proof}
Since  $(x,x') \in {\Equiv_{\Sh{P}} (X)}$, there exist maps
$u,u' \colon Y \rMap{} X$ such that $x \cdot u = x' \cdot u'$.  
By coconfluence, let
$v,v' \colon Z \rMap{} Y$ be such that $u \circ v = u' \circ v'$.
Define $p := u \circ v$ and
$y' := y \cdot u \cdot v'$. 
By coconfluence again, let
$w,w' \colon W \rMap{} Z$ be such that $w \circ v = w' \circ v'$.  Then $w, w'$ show that
$((x \cdot p, \,y \cdot p), \,(x' \cdot p, \,y')) \in\, {\Equiv_{\Sh{P} \times \Sh{Q}} (Z)}$, because:
\begin{gather*}
x \cdot p \cdot w  =  x \cdot u \cdot v \cdot w  =  x' \cdot u' \cdot v' \cdot w'  =  x' \cdot u \cdot v  \cdot w' 
~ = ~ x' \cdot p \cdot w' \\
\intertext{and}
y \cdot p \cdot w ~ = ~ y \cdot u \cdot v \cdot w ~ = ~  y \cdot u \cdot v' \cdot w'  ~ = ~ y' \cdot w' \enspace .
\end{gather*}
\end{proof}

%%%%%%%%%%%%%%%%%%%%%%%%%
%%% SECTION
%%%%%%%%%%%%%%%%%%%%%%%%%

\section{Independent pullbacks}
\label{section:independent-pullbacks}

Whereas Section~\ref{section:atomic-equivalence} has given equivalence formulas a canonical interpretation in an arbitrary atomic sheaf topos
$\ShAt(\sCat{C})$,
the interpretation of conditional independence formulas (seemingly) requires additional structure on the generating category $\sCat{C}$. 
Primary amongst this is that 
$\sCat{C}$ possess \emph{independent pullbacks}, as defined below.

\begin{definition}[Independent pullbacks]
\label{definition:independent-pullbacks}A system of 
\emph{Independent pullbacks} on a category $\sCat{C}$ is given by a collection of commuting squares in $\sCat{C}$, called \emph{independent squares}.  A commuting  square
\begin{equation}
\label{equation:independent-square}
\vcenter{\hbox{\includegraphics[scale=0.18]{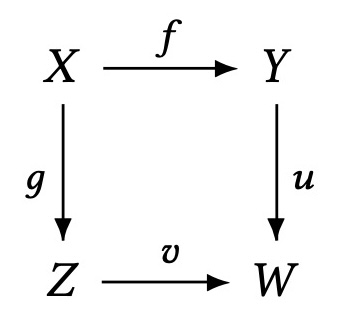}}}
%\begin{diagram} X & \rTo^f & Y   \\
%\dTo^g & & \dTo_u  \\
%Z & \rTo^{v} & W 
%\end{diagram}
\end{equation}
is then defined to be an \emph{independent pullback} if it is independent and it satisfies the usual pullback property restricted to independent squares; i.e., for every independent square
\[
\vcenter{\hbox{\includegraphics[scale=0.18]{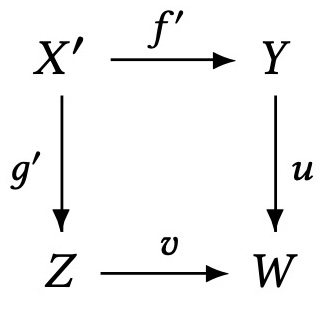}}}
%\begin{diagram} X' & \rTo^{f'}& Y   \\
%\dTo^{g'} & & \dTo_u  \\
%Z & \rTo^{v} & W 
%\end{diagram}
\]
there exists a unique $q: X' \rMap{} X$ such that $f \circ p = f'$ and $g' \circ p = g$.
The assumed collection of independent squares and derived collection of independent pullbacks are together required to satisfy the five conditions below.
\begin{description}
\item[(IP1)] Every commuting square of the form below is independent.
\[
\vcenter{\hbox{\includegraphics[scale=0.18]{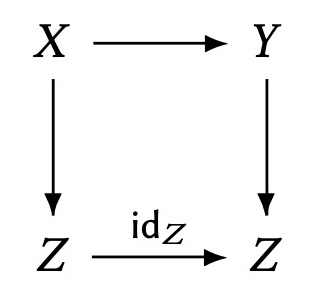}}}
%\begin{diagram} X & \rTo & Y \\
%\dTo & & \dTo \\
%Z & \rTo^{\Id_Z} & Z
%\end{diagram}
\]
\item[(IP2)]
If the left square below is independent then so is the right.
\[
\vcenter{\hbox{\includegraphics[scale=0.18]{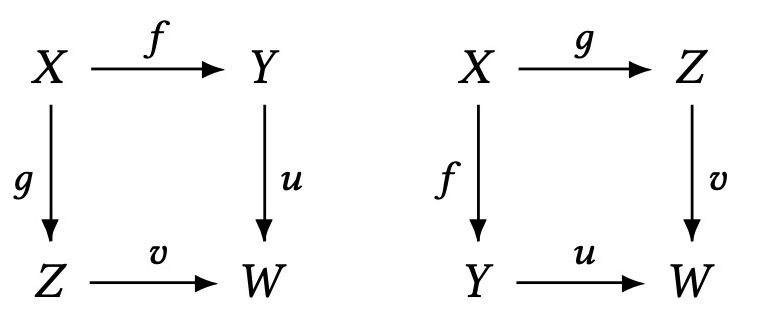}}}
%\begin{diagram} X & \rTo^f & Y  & \quad & X & \rTo^g & Z \\
%\dTo^g & & \dTo_u & & \dTo^f & & \dTo_v \\
%Z & \rTo^{v} & W & & Y & \rTo^u & W
%\end{diagram}
\]
\item[(IP3)]
If $(A)$ and $(B)$ below are independent, then so is the composite rectangle $(AB)$.
\begin{equation}
\label{equation:AB}
\vcenter{\hbox{\includegraphics[scale=0.18]{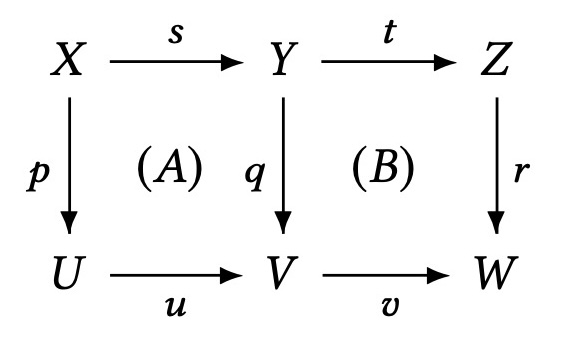}}}
%\begin{diagram}
% X & \rTo^s & Y  &  \rTo^t &  Z \\
%\dTo^p & {(A)~} & \dTo^q & {(B)~} & \dTo_r \\
%U & \rTo_u  & V  & \rTo_v & W
%\end{diagram}
\end{equation}
\item[(IP4)] If the composite rectangle $(AB)$ above is independent and $(B)$ is an independent pullback then $(A)$ is independent.

\item[(IP5)] Every cospan $Y \rMap{u} W \lMap{v} Z$ has a completion to a commuting square \eqref{equation:independent-square} that is an independent pullback. 
\end{description}

\end{definition}

\noindent
It is an easy consequence of  axioms (IP1) and (IP3) that, in any commuting diagram as below, 
if the right square is independent then so is the outer kite.
\begin{equation}
\label{diagram:kite}
\vcenter{\hbox{\includegraphics[scale=0.18]{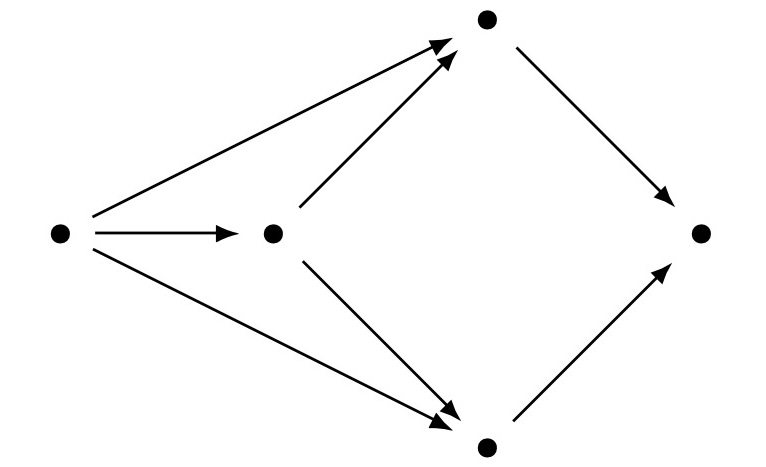}}}
%\begin{diagram}
%&&&& {\bullet} && \\
%&  && \ruTo \ruTo(4,2) && \rdTo  & \\
%{\bullet} & \rTo & {\bullet} &&&& {\bullet} &  \\
%& \rdTo(4,2) && \rdTo &&  \ruTo & \\
%&&&& {\bullet} &&
%\end{diagram}
\end{equation}
A straightforward consequence of this property in turn is that, in any independent pullback square~\eqref{equation:independent-square}, the span $f,g$ is \emph{jointly monic}, i.e., for all
parallel pairs $s,t: V \rMap{} X$, if both $f\circ s = f \circ t$ and $g \circ s = g \circ t$ then $s = t$.

\begin{definition}[Descent property]
We say that independent-pullback structure has the \emph{descent property} if, 
in any commuting diagram of the form~\eqref{diagram:kite} above, 
if the outer kite is independent then so is the right-hand square.
\end{definition}

As a first (trivial) example of independent pullbacks, in any category $\sCat{C}$ with pullbacks the collection of all commuting squares defines an independent pullback structure on $\sCat{C}$ satisfying the descent property, for which the independent pullbacks are exactly the pullbacks.  The category $\Sur$ (which does not have pullbacks) provides a nontrivial example. 

\begin{definition}[Independent square in $\Sur$]
\label{definition:indep-square-sur}
Define a  commuting square in $\Sur$
\[
\vcenter{\hbox{\includegraphics[scale=0.18]{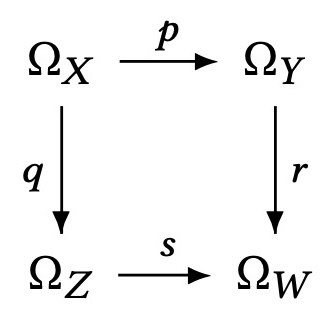}}}
%\begin{diagram}
%{\Omega_X} & \rTo^p & {\Omega_Y} \\
%\dTo^q & & \dTo_r \\
%{\Omega_Z} & \rTo^s & {\Omega_W}
%\end{diagram}
\]
to be \emph{independent} if $p \Indep q \Cond r \circ p$, using conditional independence of nondeterministic variables (Definition~\ref{definition:conditional-independence-NV}). 
\end{definition}
\begin{proposition}
\label{proposition:Sur-indep-pullbacks}
Definition~\ref{definition:indep-square-sur} endows 
$\Sur$ with independent pullback structure satisfying the descent property.
\end{proposition}
\begin{proof}
Because the square is commuting and the maps are surjective, the condition
of Definition~\ref{definition:conditional-independence-NV} simplifies to: for all
$\omega_Y \in \Omega_Y$ and $\omega_Z \in \Omega_Z$, we have
$r(\omega_Y) = s(\omega_Z)$ implies there exists $\omega_X \in \Omega_X$ such that
$p(\omega_X) = \omega_Y$ and $q(\omega_X) = \omega_Z$.

The easy verification of properties (IP1) and (IP2) is left to the reader. 

For (IP3), suppose (A) and (B) in diagram~\eqref{equation:AB} are independent. We show that $t \circ s \Indep p \Cond r \circ t \circ s$, using the characterisation above. 
Accordingly, suppose $\omega_Z \in \Omega_Z$ and $\omega_U \in \Omega_U$
are such that $r(\omega_Z) = v(u(\omega_U))$. We need to find $\omega_X \in \Omega_X$ such that
$t(s(\omega_X)) = \omega_Z$ and $p(\omega_X) = \omega_U$. Because
$r(\omega_Z) = v(u(\omega_U))$,
the independence of (B) gives us $\omega_Y \in \Omega_Y$ such that
$t(\omega_Y) = \omega_Z$ and $q(\omega_Y) = u(\omega_U)$. By the latter equation and the independence of (A),
there exists $\omega_X \in \Omega_X$ such that $s(\omega_X) = \omega_Y$ and $p(\omega_Z) = \omega_U$.
We then have $t(s(\omega_X)) = t(\omega_Y) = \omega_Z$ as required.

For (IP4), we verify the stronger property that if the composite rectangle (AB) in diagram~\eqref{equation:AB} is independent and if $t,q$ are jointly monic then (A) is independent. In the category $\Sur$ the joint monicity of $t,q$
means that, for all $\omega_Y, \omega_Y' \in \Omega_Y$, if both
$t(\omega_Y) = t(\omega'_Y)$ and $q(\omega_Y) = q(\omega'_Y)$ then $\omega_Y = \omega_Y'$. To prove that (A) is independent, suppose $\omega_Y \in \Omega_Y$ and $\omega_U \in \Omega_U$ are such that $q(\omega_Y) = u(\omega_U)$. Then $r(t(\omega_Y)) = v(q(\omega_Y)) = v(u(\omega_U))$. So, by the independence of (AB), there
exists $\omega_X \in \Omega_X$ such that $t(s(\omega_X)) = t(\omega_Y)$ and $p(\omega_X) = \omega_U$. 
We then have $q(s(\omega_X)) = u(p(\omega_X)) = u(\omega_U) = q(\omega_Y)$. It follows, by the joint monicity of 
$t,q$, that $s(\omega_X) = \omega_Y$. Together with the equation $p(\omega_X) = \omega_U$, this verifies the independence of (A).

For (IP5), the construction in the proof of Proposition~\ref{proposition:Sur-coconfluent} completes any cospan to an independent pullback square, as is easily verified.

We leave the straightforward verification of the descent property to the reader.
\end{proof}

\noindent
A more abstract way of describing the independent pullback structure on $\Sur$ is that  a commuting square in $\Sur$ is independent if and only if it is a 
weak\footnote{A \emph{weak limit} is a cone that enjoys the existence property but not necessarily the uniqueness property of a limit.}
pullback in $\Set$, and it is  an independent pullback if and only if it is a pullback in $\Set$. One can use this to give 
a more abstract verification that (IP1)--(IP5) and descent hold.

We end this section with some general consequences of  the definition of 
independent pullback structure. 
The first such consequence is that an analogue of the pullback lemma holds for independent pullbacks.
\begin{lemma}[Independent-pullback lemma] %\leavevmode
Suppose $\sCat{C}$ has independent pullback structure.
\begin{enumerate}
\item \label{ipl:i} If $(A)$ and $(B)$ in \eqref{equation:AB} are both independent pullbacks then so is the composite rectangle $(AB)$.

\item \label{ipl:ii} If  $(B)$ and the composite rectangle $(AB)$  in \eqref{equation:AB} are both independent pullbacks then so is $(A)$.
\end{enumerate}
\end{lemma}
\begin{proof}
The proof has the same structure as that of the standard pullback lemma, with the additional burden of having to verify that various commuting squares are independent. We give the proof of statement~\ref{ipl:i} insofar as it involves independence properties, leaving the standard uniqueness argument and the proof of statement~\ref{ipl:ii} to the reader. 

Suppose (A) and (B) are independent pullbacks. We need to verify that (AB) is an independent pullback. Accordingly, 
suppose that $z : T \rMap{} Z$ and $w : T \rMap{} U$ are such that the top square i the diagram below is independent. We need to show that there exists a unique map $x : T \rMap{} X$ such that $p \circ x = w$ and $t \circ s \circ x = z$.
\[
\vcenter{\hbox{\includegraphics[scale=0.18]{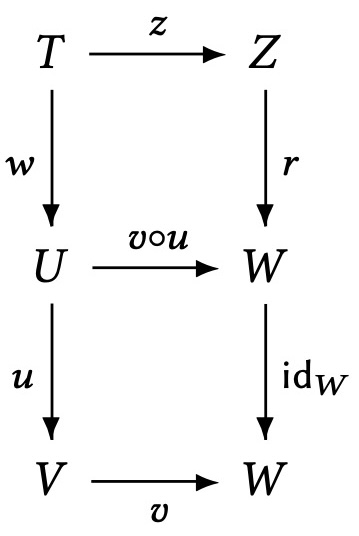}}}
%\begin{diagram} T & \rTo^z & Z   \\
%\dTo^w & & \dTo_r  \\
%U & \rTo^{v\circ u} & W \\
%\dTo^u & & \dTo_{\Id_W} \\
%V & \rTo_v & W
%\end{diagram}
\]
By axioms (IP1) and (IP2), the bottom square above is independent, hence, by (IP2) and (IP3), so is the composite rectange.
Since (B) is an independent pullback, there exists a unique $y : T \rMap{} Y$ such that $t \circ y = z$ and $q \circ y = u \circ w$. This means that the top square in the diagram above factorises as
\[
\vcenter{\hbox{\includegraphics[scale=0.18]{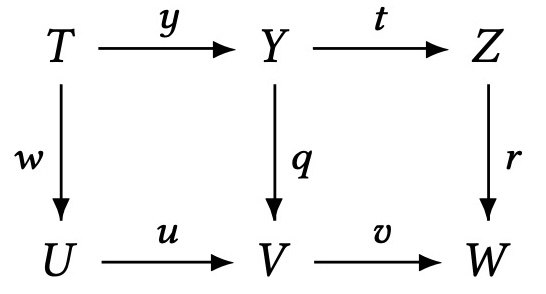}}}
%\begin{diagram} T &  \rTo^y & Y & \rTo^t & Z   \\
%\dTo^w & & \dTo_q & & \dTo_r  \\
%U & \rTo^{u} &V & \rTo^v &  W 
%\end{diagram}
\]
Since the composite rectangle is independent and the right-hand square is (B), which is an independent pullback, 
the left-hand square is independent by (IP4). Since (A) is an independent pullback, there exists a unique 
$x : T \rMap{} X$ such that $p \circ x = w$ and $s \circ x = y$, whence $t \circ s \circ x = t \circ y = z$. 
The proof that $x$ is the unique map satisfying $p \circ x = w$ and $t \circ s \circ x = z$ then proceeds as  usual.
%For uniqueness, suppose 
%$x' : T \rMap{} X$ is  such that $p \circ x' = w$ and $t \circ s \circ x' = z$. Then $q \circ s \circ x' = u \circ p \circ x'  = u \circ w$. By the characterisation of $y$, it follows from $t \circ s \circ x' = z$ and $q \circ s \circ x' = u \circ w$ that
%$s \circ x' = y$. Then, by the characterisation of $x$, we have $x' = x$. 
\end{proof}

By axiom (IP5),  any category with independent pullbacks   is \emph{a fortiori} coconfluent, hence we can consider the category $\ShAt(\sCat{C})$ of atomic sheaves, for small such $\sCat{C}$. The remaining results in this section demonstrate a pleasing interplay 
 between atomic sheaves and  independent pullback structure. They are aimed at readers who are  interested in the general category-theoretic framework. Readers keen to arrive at the atomic sheaf logic of conditional independence may prefer to skip to the next section. 
 
\begin{theorem} 
\label{theorem:sheaf-nice}
Suppose $\sCat{C}$ is a small category with independent pullback structure.
The following are equivalent, for every $P \in \Psh(\sCat{C})$.
\begin{enumerate}

\item 
\label{square-prop:a}
$P$ is an atomic sheaf.

\item 
\label{square-prop:b}
$P$ maps independent squares in $\sCat{C}$ to pullbacks in $\Set$.
\end{enumerate}

\end{theorem}
\noindent
Note that, by contravariance, $P$ maps an independent square of the form \eqref{equation:independent-square} to a pullback square in $\Set$ with apex $PW$.

The proof of Theorem~\ref{theorem:sheaf-nice} is an adaptation to the axiomatic structure of independent pullbacks 
of a standard argument (see, e.g., \cite[A 2.1.11(h)]{johnstone}) that sheaves in the Schanuel topos can be characterised as 
pullback preserving functors from the category $\I$ of finite sets and injective functions to $\Set$.

\begin{proof}
For the \eqref{square-prop:a} $\Rightarrow$ \eqref{square-prop:b} implication, suppose that $P$ is an atomic sheaf. We first show that
$P$ maps independent pullbacks in $\sCat{C}$ to pullbacks in $\Set$. Consider any independent
 pullback of the form~\eqref{equation:independent-square}. We need to show that the square below is a pullback in $\Set$.
 \[
\vcenter{\hbox{\includegraphics[scale=0.18]{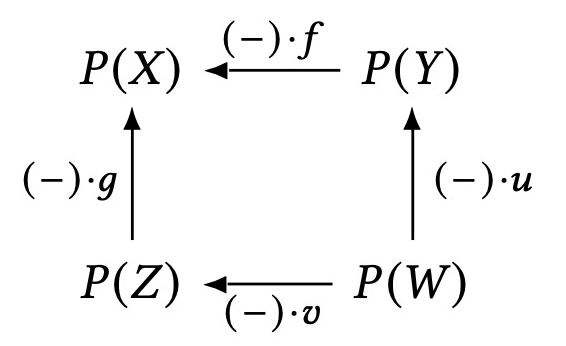}}}
% \begin{diagram} P(X) & \lTo^{(-) \cdot f} & P(Y ) \\
%\uTo^{(-)\cdot g} & & \uTo_{(-)\cdot u}  \\
%P(Z) & \lTo_{(-)\cdot v} & P(W) 
%\end{diagram}
\]
Accordingly, let $y\in P(Y)$ and $z \in P(Z)$ be such that $y \cdot f = z \cdot g$. We need to show that there exists a unique $w \in P(W)$
such that $w \cdot u = y$ and $w \cdot v = z$. 

We show that $z$ is $v$-invariant.  Let $s,t : T \rMap{} Z$ be such that $v \circ s = v \circ t$.
By the independent-pullback lemma, we can construct an independent pullback of $u$ along
$v \circ s = v \circ t$, either by composing the independent pullback  \eqref{equation:independent-square} with the independent
pullback of $g$ along $s$, or by composing \eqref{equation:independent-square} with the independent
pullback of $g$ along $t$. By a straightforward argument, this means the independent pullbacks of $g$ along $s$ and $t$ can be given 
the same left edge $g'$ as in the diagram below, which comprises three independent pullback squares (one with $f$ and $v$, one with $s'$ and $s$ and one with $t'$ and $t$).
\[
\vcenter{\hbox{\includegraphics[scale=0.18]{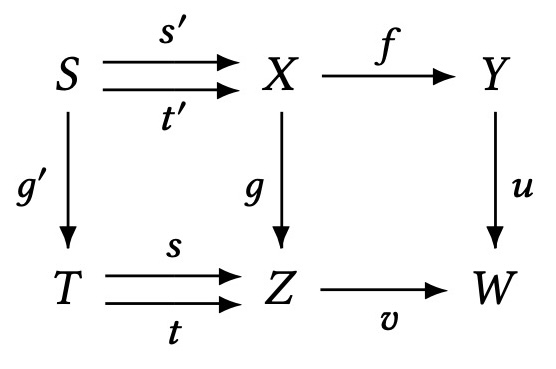}}}
%\begin{diagram} 
%S & \pile{\rTo^{s'} \\ \rTo_{t'}} & X & \rTo^f & Y   \\
%\dTo^{g'} & & \dTo^g & & \dTo_u  \\
%T & \pile{\rTo^{s} \\ \rTo_t} & Z & \rTo_{v} & W 
%\end{diagram}
\]
We have:
\begin{align*}
& z \cdot s \cdot g' = z \cdot g \cdot s' = y \cdot f \cdot s'  = y \cdot f \cdot t'  = z \cdot g \cdot t' = z \cdot t \cdot g' \enspace .
\end{align*}
Since $P$ is separated (Definition~\ref{definition:separated}) it follows that $z \cdot s = z \cdot t$.
Thus $z$ is indeed $v$-invariant.

By the sheaf property there exists $w \in P(W)$ such that $z = w \cdot v$. Then:
\[
w \cdot u \cdot f = w \cdot v \cdot g = z \cdot g = y \cdot f \enspace .
\]
So, by separatedness, we have found $w$ such that $w \cdot u = y$ and $w \cdot v = z$. Such a $w$ is unique by separatedness.

Having established that $P$ maps independent pullbacks in $\sCat{C}$ to pullbacks in $\Set$, we show that it more generally maps all independent squares
to pullbacks. Accordingly, suppose \eqref{equation:independent-square} is an independent square. By taking the independent 
pullback of $u$ along $v$, we can obtain  \eqref{equation:independent-square} as a composite:
\[
\vcenter{\hbox{\includegraphics[scale=0.18]{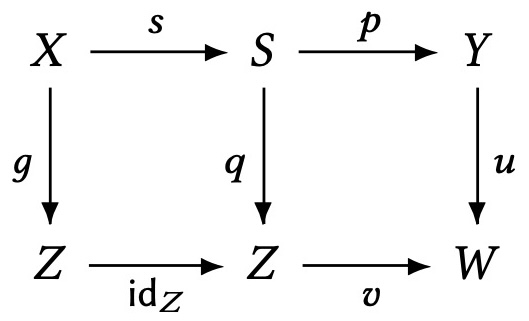}}}
%\begin{diagram} X & \rTo^s & S & \rTo^p & Y   \\
%\dTo^g & & \dTo^q & & \dTo_u  \\
%Z & \rTo_{\Id_Z} & Z  & \rTo_{v} & W 
%\end{diagram}
\]
Since the right-hand square is an independent pullback, it is mapped by $P$ to a pullback in $\Set$. The left-hand square is mapped by $P$ to a commuting square in $\Set$ with an identity in a position that makes it a trivial pullback. Thus $P$ maps the composite square \eqref{equation:independent-square} to a composition of pullbacks, hence to a pullback.

For the \eqref{square-prop:b} $\Rightarrow$ \eqref{square-prop:a} implication, let $y \in P(Y)$ and $r \colon Y \rMap{} X$ in $\sCat{C}$ be such that
$y$ is $r$-invariant. Consider an independent pullback of $r$ along itself
\[
\vcenter{\hbox{\includegraphics[scale=0.18]{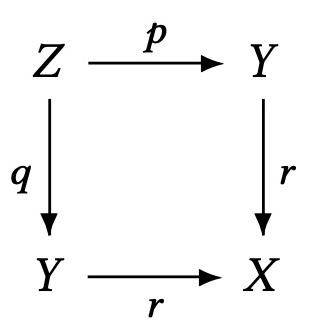}}}
%\begin{diagram} Z & \rTo^p & Y   \\
%\dTo^q & & \dTo_r  \\
%Y & \rTo_r & X 
%\end{diagram}
\]
Because $y$ is $r$-invariant, $y \cdot p = y \cdot q$. By assumption, $P$ maps the above square to a pullback in $\Set$. Hence, there exists a unique $x \in P(X)$ such that $y = x \cdot r$, as required by the sheaf property.
\end{proof}

\begin{corollary}
\label{corollary:push-pull}
The functor $\,\AS\Yon: \sCat{C} \to \ShAt(\sCat{C})$ maps independent squares in $\sCat{C}$ to pushouts
in $\ShAt(\sCat{C})$.
\end{corollary}

\begin{proof} This is a straightforward consequence of Theorem~\ref{theorem:sheaf-nice} on account of the bijections
\[
\ShAt(\AS\Yon(X), \Sh{A})  \cong  \Psh(\Yon(X), \Sh{A}) \cong \Sh{A}(X) \enspace ,
\]
natural in $X$ and $\Sh{A}$, given by the left-adjoint property of the associated sheaf functor and by the Yoneda lemma.

In more detail, consider any independent square in $\sCat{C}$ of the 
form~\eqref{equation:independent-square}. Suppose we have maps $\beta$ and $\gamma$ 
in  $\ShAt(\sCat{C})$ making the outside kite below commute.
\[
\vcenter{\hbox{\includegraphics[scale=0.18]{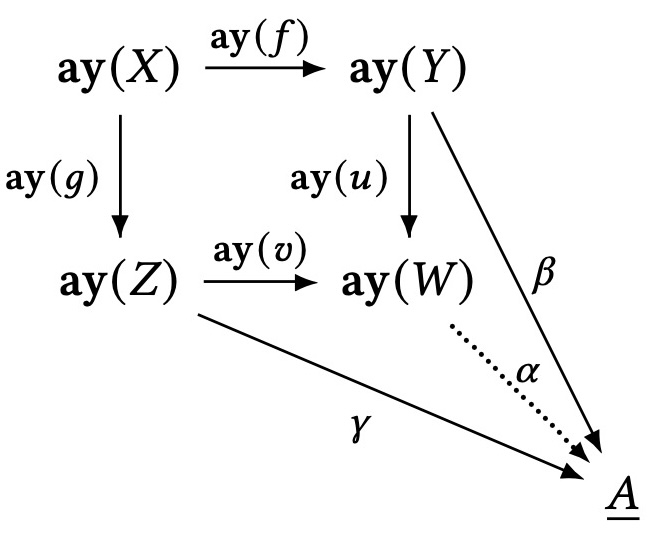}}}
%\begin{diagram} \AS\Yon(X) & \rTo^{\AS\Yon(f)} & \AS\Yon(Y)  & &   \\
%\dTo^{\AS\Yon(g)} & & \dTo^{\AS\Yon(u)}  & \rdTo(2,4)^\beta & \\
%\AS\Yon(Z) & \rTo^{\AS\Yon(v)} & \AS\Yon(W) & & \\
%& \rdTo(4,2)_\gamma & & \rdDotty^{\!\!\alpha} & \\ %
% & & & & \Sh{A}
%\end{diagram}
\]
 The natural bijections above mean that $\beta$ and $\gamma$ correspond respectively to $y \in \Sh{A}(Y)$ and $z \in \Sh{A}(Z)$ satisfying $y \cdot f = g \cdot v$. Since, by 
 Theorem~\ref{theorem:sheaf-nice}, $\Sh{A}$ maps the square~\eqref{equation:independent-square} to a pullback in $\Set$, there exists a unique $w \in \Sh{A}(w)$ such that $w \cdot u = y$ and $w \cdot v = z$. Translating back along the natural bijections, there exists a unique map $\alpha \colon \AS\Yon(W) \rMap{} \Sh{A}$ such that
 $\alpha \cdot \AS\Yon(u) = \beta$ and $\alpha \cdot \AS\Yon(v) = \gamma$, as required.
\end{proof}

\begin{corollary} 
\label{corollary:are-pushouts}
The following are equivalent for a small category $\sCat{C}$ with independent pullbacks.
\begin{enumerate}
\item \label{item:subcanonical} Every representable presheaf %$\Yon(X)$ 
is an atomic sheaf.
\item \label{item:independents-are-pushouts} Every independent square in $\sCat{C}$ is a pushout.
\end{enumerate}
\end{corollary}

\begin{proof} For the \eqref{item:subcanonical} $\implies$ \eqref{item:independents-are-pushouts} direction,
suppose every representable is an atomic sheaf. Then $\AS\Yon$ and $\Yon$ are naturally isomorphic, hence
$\AS\Yon: \sCat{C} \to \ShAt(\sCat{C})$ is full and faithful. As a fully faithful functor, $\AS\Yon$ reflects
(co)limits in general, and so pushouts in particular. Thus independent squares are pushouts in $\sCat{C}$ by
Corollary~\ref{corollary:push-pull}.

For the \eqref{item:independents-are-pushouts} $\implies$ \eqref{item:subcanonical} direction, 
it holds from the definition of  $\Yon(X)$ as $\sCat{C}(-,X)$ that
every representable presheaf $\Yon(X): \Op{\sCat{C}} \to \Set$ maps any  colimit of a  $D$-shaped  diagram in $\sCat{C}$  to 
a limit of the induced $\Op{D}$-shaped diagram in $\Set$. In particular, $\Yon(X)$ maps pushouts in $\sCat{C}$ to
pullbacks in $\Set$. So, if every independent square is a pushout in $\sCat{C}$, then representables 
map independent squares to pullbacks in $\Set$, and it follows from Theorem~\ref{theorem:sheaf-nice} that representables are
atomic sheaves.
\end{proof}

%%%%%%%%%%%%%%%%%%%%%%%%%
%%% SECTION
%%%%%%%%%%%%%%%%%%%%%%%%%

\section{Atomic conditional independence}
\label{section:atomic-conditional-independence}

%Proposition~\ref{proposition:indep-subsheaf} defined a subsheaf $\Indep_{A,B|C} \subseteq \NEls(A)\! \times\! \NEls(B)\! \times \!\NEls(C)$
%representing conditional independence between sheaves of nondeterministic variables in $\ShAt(\Sur)$.

A main goal of this section is to define a canonical subsheaf $\Indep_{\Sh{A},\Sh{B}|\Sh{C}} \subseteq \Sh{A} \times \Sh{B} \times \Sh{C}$
representing a conditional independence relation between sheaves $\Sh{A}, \Sh{B}, \Sh{C}$ in
atomic sheaf toposes $\ShAt(\sCat{C})$.
To achieve this, we shall require that
$\sCat{C}$ have independent pullbacks. We shall also need
to assume that the sheaves $\Sh{A}, \Sh{B}, \Sh{C}$ enjoy the special property of having  \emph{supports}, a notion that we now define. 

%The definition of below, which was influenced by ~\cite{LAST},
%improves on the treatment of supports in 
%the conference version of the paper.

%However, 
%%To do this, we require the conditioning sheaf $\Sh{C}$ to satisfy one further property: that it has \emph{weak supports} as defined below.
% our approach requires some additional assumptions. The first such assumption is that we require 
%all maps in $\sCat{C}$ to be epimorphic. (As discussed at the end of Section~\ref{section:atomic-sheaves}, this is weaker than assuming that all representables are sheaves.) Exploiting this assumption, we identify a useful property of presheaves. 

\begin{definition}[Supports]
\label{definition:supports}
A \emph{representable factorisation} of an element $x \in P(X)$, where $P \in \Psh(\sCat{C})$,
is given by a triple $(Y,q,y)$ such that:
$q : X \rMap{} Y$ is a map in $\sCat{C}$, we have $y \in P(Y)$ and $x = y \cdot q$. A \emph{morphism} from
one representable factorisation $(Y,q,y)$ of $x$ to another $(Y',q',y')$ is given by a map $r \colon Y \rMap{} Y'$ in $\sCat{C}$ such that
$r \circ q = q'$ and $y' \cdot r = y$. A representable factorisation $(Y,q,y)$ is called a \emph{support} for $x$ when it is a terminal object in the category  of representable factorisations of $x$. A presheaf $P \in \Psh(\sCat{C})$ is said to have \emph{supports} if, for every $X \in \sCat{C}$, it holds that every $x \in P(X)$ has a support. 
\end{definition}
For readers familiar with the \emph{category of elements} $\int \! P$ of a presheaf $P$, we remark that a support for $x \in P(X)$ is the same thing as a terminal object in the co-slice category $(X,x) / {\int \! P}$. This elegant formulation is used 
as the definition of support in~\cite{LAST} (there called \emph{minimal support}).

\begin{lemma} 
\label{lemma:weak-supports}
Suppose all maps in $\sCat{C}$ are epimorphic and  that $P \in \Psh(\sCat{C})$ has supports.
Then, for any $x \in P(X)$ and map $Y \rMap{q} X$  in $\sCat{C}$,
a representable factorisation $(Z, t, z)$ of $x$ is a support for $x$ if and only if  $(Z, t \circ q , z)$ is a support for $x \cdot q$.
\end{lemma}
\begin{proof} Suppose $(Z, t, z)$ is a support for $x$.  
Let $(W,u, w)$ be a support for $x \cdot q$.  Because 
$(Z, t \circ q, z)$ is a representable factorisation of $x \cdot q$, there exists a unique map $r : Z \rMap{} W$ that is a morphism from
$(Z, t \circ q, z)$ to $(W,u, w)$. 
Then $(W, r \circ t, w)$ is a representable factorisation of $x$. So there exists a unique map
$s: W \rMap{} Z$ that is a morphism from $(W, r \circ t, w)$ to $(Z, t, z)$. That is, $r$ is the unique map such that $r \circ t \circ q = u$ and $w \cdot r = z$, and $s$ is the unique map such that $s \circ r \circ t = t$ and $z \cdot s = w$. Since $t$ is an epi, 
the equation $s \circ r \circ t = t$ implies $s \circ r = \Id_Z$. Then we have $t \circ q = s \circ r \circ t \circ q = s \circ u$, which means that $s$ is a morphism of $x \cdot q$ factorisations from $(W,u, w)$ to $(Z, t \circ q, z)$. So $r \circ s$ is a morphism from $(W,u, w)$ to itself. Since  $(W,u, w)$ is the terminal $x \cdot q$ factorisation, $r \circ s = \Id_{W}$. Thus $r$ and $s$ are mutual inverses,
and $t$ is an isomorphism of $x \cdot q$ factorisations from $(W,u, w)$ to $(Z, t \circ q, z)$. Hence $(Z, t \circ q, z)$ is also a support for $x \cdot q$.

Conversely, suppose $(Z, t \circ q, z)$ is a support for  $x \cdot q$. Let $(V, v, w)$ be a representable factorisation of $x$.
Then $(V, v \circ q, w)$ is a representable factorisation of $x \cdot q$. So there exists a unique map $s : W \rMap{} Z$ that is a morphism from $(V, v \circ q, w)$ to $(Z, t \circ q, z)$, that is, $s \circ v \circ p = t \circ p$ and $z \cdot s = w$. Since $p$ is an epi, 
$s \circ v  = t$, and so $s$ is a (clearly unique) morphism from $(V, v, w)$ to $(Z, t, z)$. This shows that $(Z, t, z)$ is a support for $x$. 
\end{proof}

We shall also require  presheaves with  supports to be closed under finite products. This 
follows from a further property of the category $\sCat{C}$ (dual to the existence of \emph{$\mathcal{M}$-images} as defined  in \cite[{\S}5.1]{staton}).

\begin{definition}[Pairings]
A \emph{pair factorisation} of a span $Y \lMap{f} X \rMap{g} Z$ in a category $\sCat{C}$ is given by $(X',q', f',g')$ where
$q' : X \rMap{} X'$ and $Y \lMap{f'} X' \rMap{g'} Z$  are maps in $\sCat{C}$  that 
satisfy $f' \circ q' = f$ and $g' \circ q' = g$. 
A \emph{morphism} from a pair factorisation $(X',q', f',g')$  of $f,g$ to another $(X'',q'', f'',g'')$ is a map $r : X' \to X''$ in  $\sCat{C}$ such that 
$r \circ q' = q''$, $~f'' \circ r = f'$ and $g'' \circ r = g'$. A pair factorisation $(X',q', f',g')$  is said to be a \emph{pairing} for $f,g$
if it is a terminal object in the category of pair factorisations of $f,g$. We say that the category $\sCat{C}$ has \emph{pairings} if every span $f,g$ has a pairing.
\end{definition}

\begin{proposition} 
\label{proposition:product-supports}
Suppose all maps in $\sCat{C}$ are epimorphic and
that $\sCat{C}$ has pairings. If ${P}, {Q} \in \Psh(\sCat{C})$ both have  supports, then so does the product
${P} \times {Q}$.
\end{proposition}

\begin{proof}
Consider any element $(x,y) \in ({P} \times {Q})(X)$. Let $(U,u, x')$ be support for $x$ and $(V, v, y')$ support for $y$. 
Let $(W, w, u', v')$ be a pairing for $u,v$. We show that $(W,w, (x' \cdot u', y' \cdot v'))$ is support for $(x,y)$.

Let $(Z,t,(x'',y''))$ be any representable factorisation of $(x,y)$.
Then $(Z,t,x'')$ is a representable factorisation of $x$, so there exists a unique
map $r: Z \rMap{} U$ that is a morphism from $(Z,t,x'')$ to $(U,u, x')$, i.e., such that $r \circ t = u $ and $x' \cdot r = x''$. Similarly, 
there exists a unique map $s: Z \rMap{} V$ such that $s \circ t = v$ and $y' \cdot s = y''$. 
Since $(X, t,r,s)$ is a pair factorisation of $u,v$, there exists a unique $w' : Z \rMap{} W$ such that
$w' \circ t = w$ and $u' \circ w' = r$ and $v' \circ w' = s$. We claim that $w' : Z \rMap{} W$ is the unique
morphism from $(Z,t,(x'',y''))$ to $(W,w, (x' \cdot u', y' \cdot v'))$.
We have seen  that $w' \circ t = w$. Since $t$ is epimorphic, this determines $w'$ uniquely. 
It also holds  that
$x' \cdot u' \cdot w' =  x' \cdot r = x''$ and $y' \cdot v' \cdot w' = y' \cdot s = y''$.
So $w'$ is indeed a morphism of representable factorisations. 
\end{proof}

We explore the above properties in the case of our running example $\ShAt(\Sur)$. \begin{proposition}
\label{prop:NE-supports}
In $\ShAt(\Sur)$ every  sheaf of the form $\NEls(A)$ has supports. 
\end{proposition}
\begin{proof}
Consider any $X \in \NEls(A)(\Omega)$, i.e., $X \colon \Omega \to A$. Factorise $X$ as a composite
$\Omega \rMap{p} \Omega' \rMap{X'} A$ where $p$ is surjective and $X'$ injective. % and $\Omega'$ is an object of $\Sur$.
It is easy to verify that $(\Omega',p,X')$ is a support for $X$.
\end{proof}

\begin{proposition}
\label{proposition:Sur-pairings}
The category $\Sur$ has pairings.
\end{proposition}
\begin{proof}
Consider a span $\Omega_Y \lMap{p} \Omega \rMap{q} \Omega_Z$ in $\Sur$. Factorise the function  
$(p,q) : \Omega \to \Omega_Y \times \Omega_Z$  as $\Omega \rMap{r} \Omega' \rMap{(p',q')} \Omega_Y \times \Omega_Z$ where
$r$ is surjective and $(p',q')$ injective. Then $(\Omega', r, p',q')$ is a pairing of $p,q$.
\end{proof}

Proposition~\ref{prop:NE-supports} is in fact subsumed by a much more general result, which however has a far more involved proof..
\begin{theorem}
\label{theorem:sur-supports}
In $\ShAt(\Sur)$ every sheaf has supports. 
\end{theorem}
\noindent
Because this theorem is not central to the development, we relegate its proof to Appendix~\ref{appendix:supports}

We henceforth impose global assumptions on our category $\sCat{C}$.
\begin{definition}
\label{definition:requisite}
We say that a small category $\sCat{C}$ has the \emph{requisite structure} if: every map in $\sCat{C}$ is an epimorphism, 
it has pairings, and it has independent-pullback structure  satisfying the descent property.
\end{definition}
\noindent
The reason for imposing the assumption that every map is an epimorphism is that it allows us to apply
Lemma~\ref{lemma:weak-supports} and Proposition~\ref{proposition:product-supports}. 
Because the role of $\sCat{C}$ is to serve as the gateway to the category $\ShAt{\sCat{C}}$ of atomic sheaves, this assumption is very mild.
As discussed at the end of Section~\ref{section:atomic-sheaves}, it is weaker than assuming that all representable presheaves are atomic sheaves. Moreover, every atomic sheaf topos is equivalent to
$\ShAt{\sCat{C}}$ for some coconfluent small category $\sCat{C}$ in which every map is an epimorphism.

Since it is obvious that every map in $\Sur$ is epimorphic, Propositions~\ref{proposition:Sur-pairings} 
and~\ref{proposition:Sur-indep-pullbacks} show that the category $\Sur$ has the requisite structure.

\textbf{For the remainder of the present section, let $\sCat{C}$ be a small category with the requisite structure.}

We define a general \emph{atomic conditional independence} relation for atomic sheaves
$\Sh{A}, \Sh{B}, \Sh{C}$ on $\sCat{C}$  with supports. For any $X \in \sCat{C}$, define
\begin{equation}
\label{equation:independence-subsheaf}
{\Indep}_{\Sh{A},\Sh{B}|\Sh{C}}(X) \subseteq (\Sh{A} \times \Sh{B} \times \Sh{C})(X)
\end{equation}
to consist of those triples $(x,y,z) \in (\Sh{A} \times \Sh{B} \times \Sh{C})(X)$ that satisfy the 
condition: there exists an independent square $r \circ p = s \circ q$ in $\sCat{C}$ (as in the diagram below),
and there exist elements $(x',u') \in (\Sh{A} \times  \Sh{C})(X_x)$, and
$(y',v') \in (\Sh{A} \times  \Sh{C})(X_y)$ and $z' \in  \Sh{C}(X_z)$ such that
$x' \cdot p = x$ and $y' \cdot q = y$ and $z' \cdot r = u'$ and $z' \cdot s = v'$ and 
$(X_z,\, r \circ p,\, z')$ is  support for $z$. The data in the condition above is illustrated by the hybrid diagram below,
where  the symbol $\Indep$ indicates that the square is independent.
\begin{equation}
\label{diagram:hybrid}
\vcenter{\hbox{\includegraphics[scale=0.18]{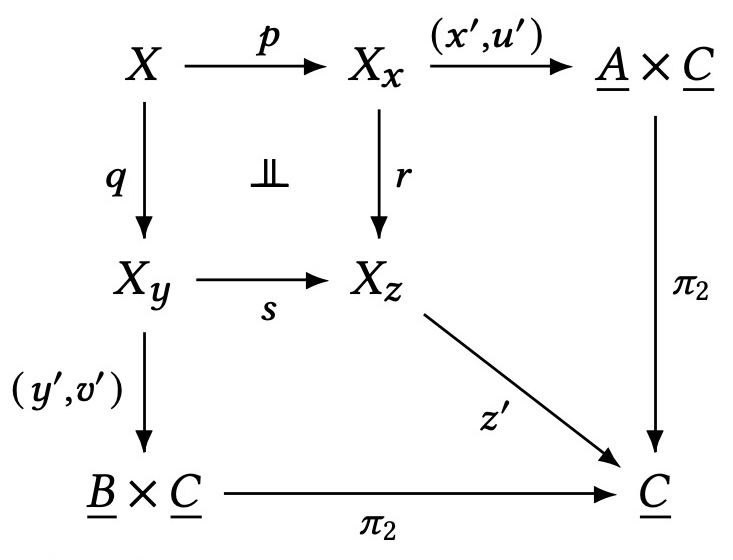}}}
%\begin{diagram}
%X & \rTo^{p} & X_x  & \rTo^{\!\!\!\!\!{(x',u')}\;\;} & \Sh{A} \times \Sh{C} \\
%\dTo^{q} & \Indep &  \dTo_r &  &  \\
%X_y & \rTo_s & X_z & & \dTo_{\pi_2} \\
%\dTo^{(y',v')} & & & \rdTo_{z'} &  \\
%\Sh{B} \times \Sh{C} & & \rTo_{\pi_2} & & \Sh{C}
%\end{diagram}
\end{equation}
The above diagram is \emph{hybrid} in the sense that the arrows in it represent three distinct kinds of entity. Arrows of the form 
$X \rMap{} Y$ between objects of $\sCat{C}$ represent maps in $\sCat{C}$. Arrows of the form $X \rMap{} \Sh{A}$, from an object $X$ of $\sCat{C}$ to a sheaf $\Sh{A}$, represent elements of the set $\Sh{A}(X)$. Arrows of the form  $\Sh{A} \rMap{} \Sh{B}$ between sheaves represent maps in $\ShAt(\sCat{C})$. By the Yoneda lemma, such hybrid diagrams can equivalently be interpreted as 
ordinary diagrams in the presheaf category $\Psh(\sCat{C})$, with objects  $X$ of $\sCat{C}$ being interpreted as representable presheaves $\Yon X$.

%\footnote{The diagram is \emph{hybrid} because it mixes maps in $\sCat{C}$, elements of sheaves and maps in $\ShAt(\sCat{C})$. The symbol $\Indep$ indicates that the square is independent.}

\begin{lemma}
\label{lemma:use-supports}
In the definition of ${\Indep}_{\Sh{A},\Sh{B}|\Sh{C}}(X)$, we can, without loss of generality, choose the data so that
$(X_x, p, (x',u'))$ is support for $(x,z)$ and 
$(X_y, q, (y',v'))$ is support for $(y,z)$.
\end{lemma}
\begin{proof} Suppose we have:
\begin{equation}
\label{diagram:for-support-a}
\vcenter{\hbox{\includegraphics[scale=0.18]{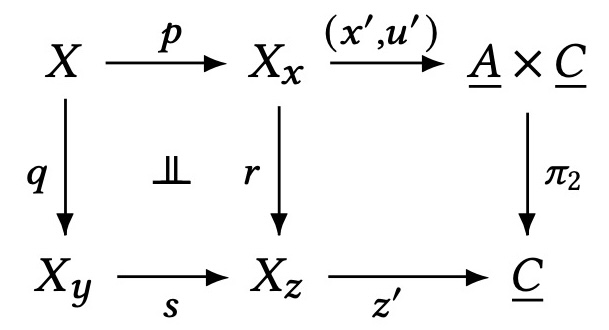}}}
%\begin{diagram}
%X & \rTo^{p} & X_x  & \rTo^{\!\!(x',u')} & \Sh{A} \times \Sh{C} \\
%\dTo^{q} & \Indep &  \dTo^r &  &   \dTo_{\pi_2} \\
%X_y & \rTo_s & X_z &   \rTo_{z'} & \Sh{C} 
%\end{diagram}
\end{equation}
where $(X_z,\, r \circ p,\, z')$ is a support for $z$.
Let $(X'_x, t, (x'',u''))$ be support for $(x',u') \in (\Sh{A} \times \Sh{C})(X_x)$.
Then $(X'_x, t, u'')$ is a representable factorisation of $z' \cdot r$.
By Lemma~\ref{lemma:weak-supports}, $(X_z, r, z')$ is a support for $z' \cdot r$.
So there exists $r': X'_x \to X_z$ such that $r' \circ t = r$ and 
$z' \cdot r' = u''$. We have thus obtained the data in the hybrid diagram below.
\begin{equation}
\label{diagram:for-support-b}
\vcenter{\hbox{\includegraphics[scale=0.18]{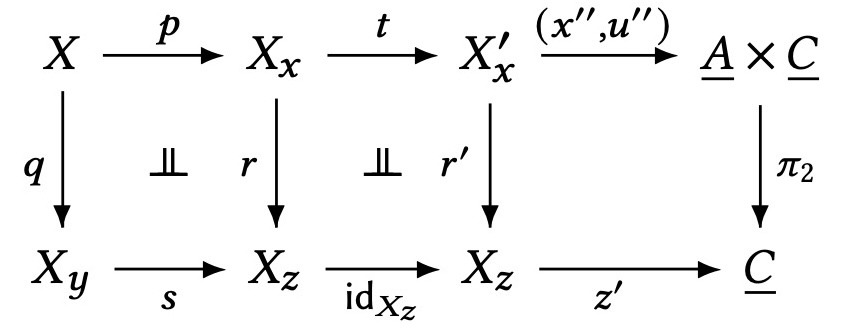}}}
%\begin{diagram}
%X & \rTo^{p} & X_x  & \rTo^t & X'_x & \rTo^{\!\!(x'',u'')} & \Sh{A} \times \Sh{C} \\
%\dTo^{q} & \Indep &  \dTo^r &  \Indep & \dTo^{r'} &   &  \dTo_{\pi_2} \\
%X_y & \rTo_s & X_z & \rTo_{\Id_{X_z}} & X_z &   \rTo_{z'} & \Sh{C} 
%\end{diagram}
\end{equation}
Morover, by Lemma~\ref{lemma:weak-supports}, it holds that
$(X'_x,\, t \circ p, (x'',u''))$ is support for $(x',u') \cdot p = (x,z)$.
We have thus shown how diagram~\eqref{diagram:for-support-a}, gives rise to
diagram~\eqref{diagram:for-support-b}, in which the composite independent square satisfies the desired support property for
$(x,z)$.

By starting with the new diagram and repeating the same argument in a vertical rather than horizontal direction, one similarly satisfies the required 
support property for $(y,z)$.
\end{proof}

\begin{theorem} 
\label{theorem:general-conditional-independence}
Suppose %$\sCat{C}$ has independent pullback structure satisfying the inner square property, and
$\Sh{A}, \Sh{B}, \Sh{C}$ are atomic sheaves with supports. %, where $\Sh{C}$ has weak supports.
Then $\Indep_{\Sh{A},\Sh{B}|\Sh{C}}(X) \subseteq (\Sh{A} \times \Sh{B} \times \Sh{C})(X)$ defines a subsheaf via 
Prop.~\ref{proposition:subsheaf}.
\end{theorem}

\begin{proof}[Proof of Theorem~\ref{theorem:general-conditional-independence}]
We first show that $\Indep_{\Sh{A},\Sh{B}|\Sh{C}}$ is a subpresheaf. Suppose $(x,y,z) \in \Indep_{\Sh{A},\Sh{B}|\Sh{C}}(X)$ and
$t: Y \rMap{} X$ is a map in $\sCat{S}$; that is, we have the data 
in diagram~\eqref{diagram:hybrid} and $(X_z,\, r \circ p,\, z')$ is a support for $z$.
We need to show that $(x\cdot r,\,y \cdot r ,\,z \cdot r) \in \Indep_{\Sh{A},\Sh{B}|\Sh{C}}(Y)$.
This holds on account of the data illustrated below.
\[
\vcenter{\hbox{\includegraphics[scale=0.18]{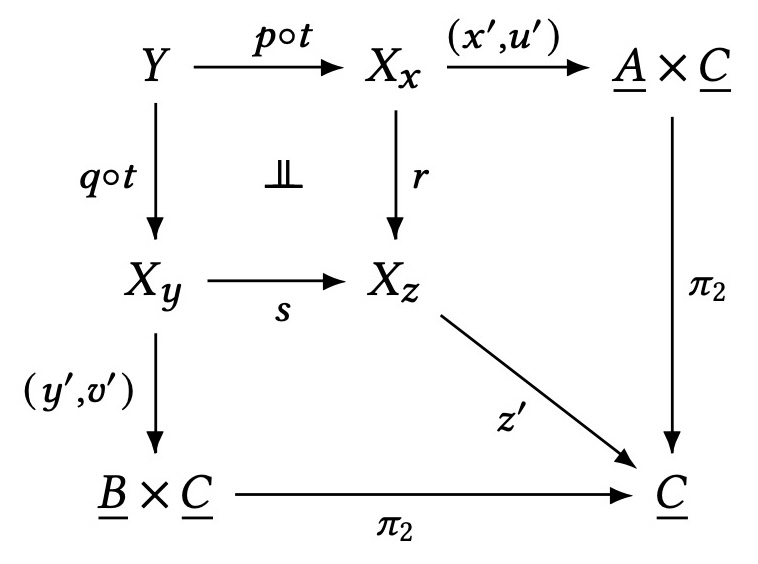}}}
%\begin{diagram}
%Y & \rTo^{p\circ t} & X_x  & \rTo^{\!\!\!\!\!{(x',u')}\;\;} & \Sh{A} \times \Sh{C} \\
%\dTo^{q \circ t} & \Indep &  \dTo_r &  &  \\
%X_y & \rTo_s & X_z & & \dTo_{\pi_2} \\
%\dTo^{(y',v')} & & & \rdTo_{z'} &  \\
%\Sh{B} \times \Sh{C} & & \rTo_{\pi_2} & & \Sh{C}
%\end{diagram}
\]
Indeed, $(X_z,\, r \circ p \circ t,\, z')$ is support for $z \circ r$ on account of Lemma~\ref{lemma:weak-supports}, and the marked square is 
independent  since it is a composition of two independent squares:
\[
\vcenter{\hbox{\includegraphics[scale=0.18]{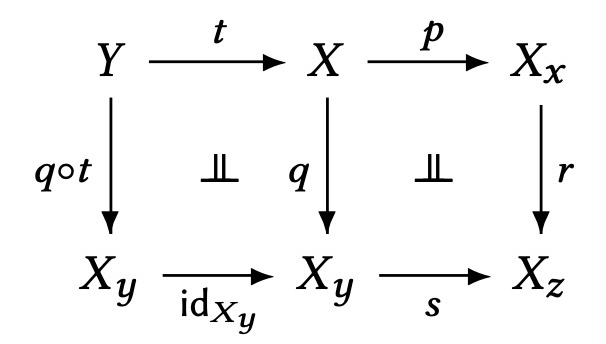}}}
%\begin{diagram}
%Y & \rTo^t & X & \rTo^{p} & X_x  \\
%\dTo^{q \circ t} & \Indep & \dTo^{q} & \Indep &  \dTo_r   \\
%X_y & \rTo_{\Id_{X_y}} & X_y & \rTo_s & X_z 
%\end{diagram}
\]

For the subsheaf property, suppose $(x,y,z) \in (\Sh{A}\times \Sh{B}\times\Sh{C})(X)$ and 
$(x\cdot t,\,y \cdot t ,\,z \cdot t) \in \Indep_{\Sh{A},\Sh{B}|\Sh{C}}(Y)$ where $t: Y \rMap{} X$ is a map in $\sCat{S}$. We need to show that
$(x,y,z) \in \Indep_{\Sh{A},\Sh{B}|\Sh{C}}(X)$.

The assumption gives us the data below
\[
\vcenter{\hbox{\includegraphics[scale=0.18]{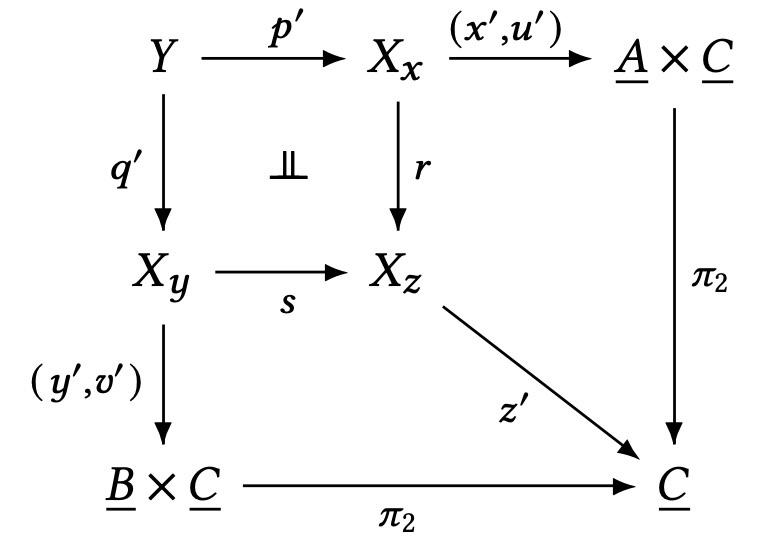}}}
%\begin{diagram}
%Y & \rTo^{p'} & X_{x}  & \rTo^{\!\!\!\!\!{(x',u')}\;\;} & \Sh{A} \times \Sh{C} \\
%\dTo^{q'} & \Indep &  \dTo_r &  &  \\
%X_y & \rTo_s & X_z & & \dTo_{\pi_2} \\
%\dTo^{(y',v')} & & & \rdTo_{z'} &  \\
%\Sh{B} \times \Sh{C} & & \rTo_{\pi_2} & & \Sh{C}
%\end{diagram}
\]
where, 
$x' \cdot p' = x \cdot t$ and $y' \cdot q' = y \cdot t$ and 
$(X_z,\, r \circ p',\, z')$ is support for $z \cdot t$. By Lemma~\ref{lemma:use-supports}, we can assume that
$(X_x, p', (x',u'))$ is support for $(x\cdot t ,z \cdot t)$ and 
$(X_y, q', (y',v'))$ is support for $(y\cdot t,z \cdot t)$.
Since $(X, t, (x,z))$ is a representable factorisation of $(x\cdot t ,z \cdot t)$, we have $p' = p \circ t$ and
$(x,z) = (x'\cdot p, \, u' \cdot p)$,  for some $p \colon X \rMap{} X_x$.
Similarly, $q' = q \circ t$ and
$(y,z) = (y'\cdot q, \, v' \cdot q)$,  for some $q \colon X \rMap{} X_y$.
Then
 \[
r \circ p \circ t = t \circ p' = s \circ q' = s \circ q \circ t \enspace .
\]
Since $t$ is epimorphic, $r \circ p = s \circ q$ is a commuting square, which is 
independent by the descent property. Accordingly, we have 
precisely the data in diagram~\eqref{diagram:hybrid}. Moreover, since
$(X_z,\, r \circ p \circ t,\, z') = (X_z,\, r \circ p',\, z')$ is support for $z \cdot t$, it follows from
Lemma~\ref{lemma:weak-supports} that $(X_z,\, r \circ p,\, z')$ is support for $z$, as required.
\end{proof}

In the special case of sheaves $\NEls(A)$ of nondeterministic variables in $\ShAt(\Sur)$,
the general atomic conditional independence %${\Indep_{\NEls(A),\NEls(B)|\NEls(C)}}$ 
defined above 
coincides with the multiteam conditional independence from
Proposition~\ref{proposition:indep-subsheaf}.

\begin{proposition} 
\label{proposition:indep-NEls}
The subsheaf 
\[{\Indep}_{\NEls(A),\NEls(B)|\NEls(C)} \subseteq \NEls(A)\times \NEls(B) \times \NEls(C)\]
in $\ShAt(\Sur)$ coincides with $\Indep_{A,B|C} \subseteq \NEls(A) \times \NEls(B) \times \NEls(C)$
from Proposition~\ref{proposition:indep-subsheaf}.
%\[
%\Equiv_{\NEls(A)}(\Omega) ~ := ~
%\{(X,X') \in (\Omega \to A)^2 \mid X(\Omega) = X'(\Omega) \} \enspace .
%\]
\end{proposition}

\begin{proof}
Suppose $X : \Omega \to A, \; Y: \Omega \to B,\; Z : \Omega \to C$ are
nondeterministic variables such that $X \Indep Y \Cond Z$, according to
Definition~\ref{definition:conditional-independence-NV}. Define
\begin{align*}
\Omega_X ~ := ~  & \{(x,z) \in A \times C \mid \exists \omega \in \Omega.\; \text{$x = X(\omega)$ and $z = Z(\omega)$}\} 
\\
\Omega_Y ~ := ~  & \{(y,z) \in B \times C \mid \exists \omega \in \Omega.\; \text{$y = Y(\omega)$ and $z = Z(\omega)$}\} 
\\
\Omega_Z ~ := ~ & Z(\Omega)
\end{align*}
Then the hybrid diagram below, shows that $(X,Y,Z)$ belongs to the atomic conditional independence ${\Indep}_{\NEls(A),\NEls(B)|\NEls(C)}(\Omega)$, since $(\Omega_Z, \,Z, \, z\mapsto z)$ is support for $Z$, by the definition of $\Omega_Z$.
\[
\vcenter{\hbox{\includegraphics[scale=0.18]{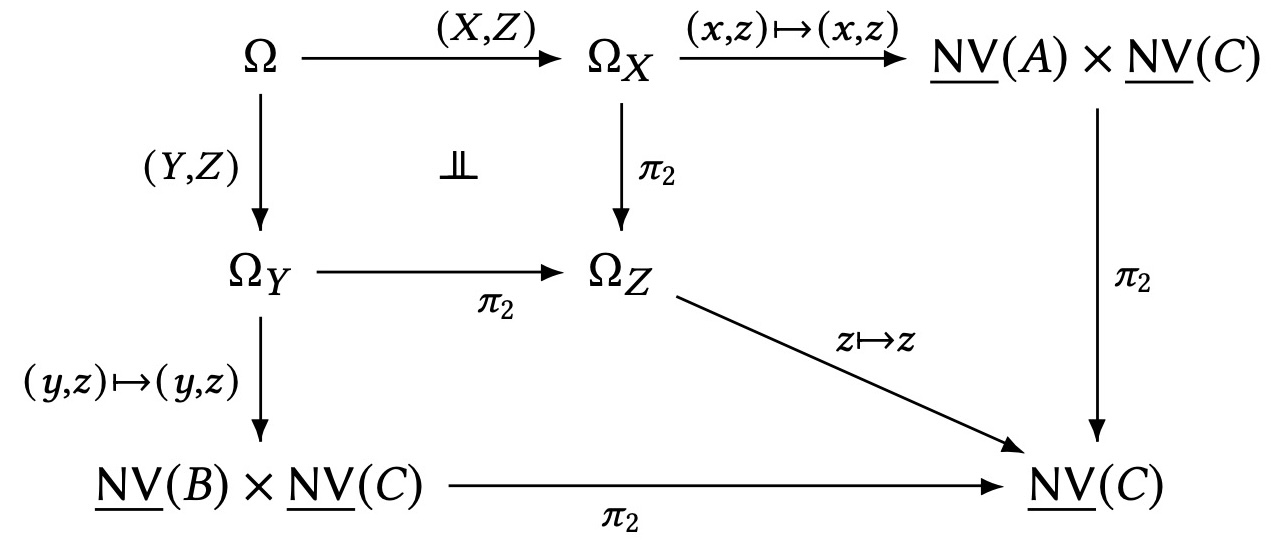}}}
%\begin{diagram}
%\Omega & \rTo^{\!\!\!(X,Z)} & \Omega_X  & \rTo^{\!\!\!(x,z) \mapsto (x,z)} & \NEls(A) \times \NEls(C)\\
%\dTo^{(Y,Z)} & \!\!\!\!\!\!\!\!\!\Indep &  \dTo_{\pi_2} &  &  \\
%\Omega_Y & \rTo_{\pi_2} & \Omega_Z & & \dTo_{\pi_2} \\
%\dTo^{(y,z) \mapsto(y,z)} & & & \rdTo^{z\mapsto z} &  \\
%\NEls(B) \times \NEls(C) & & \rTo_{\pi_2} & & \NEls(C)
%\end{diagram}
\]

Conversely, suppose $(X,Y,Z) \in {\Indep}_{\NEls(A),\NEls(B)|\NEls(C)}(\Omega)$. That is, we have the data in the hybrid diagram
below, where $X' \cdot p = X$ and $Y' \cdot q = Y$ and $(\Omega_Z, \, r \circ p, \, Z')$ is support for $Z$.
\[
\vcenter{\hbox{\includegraphics[scale=0.18]{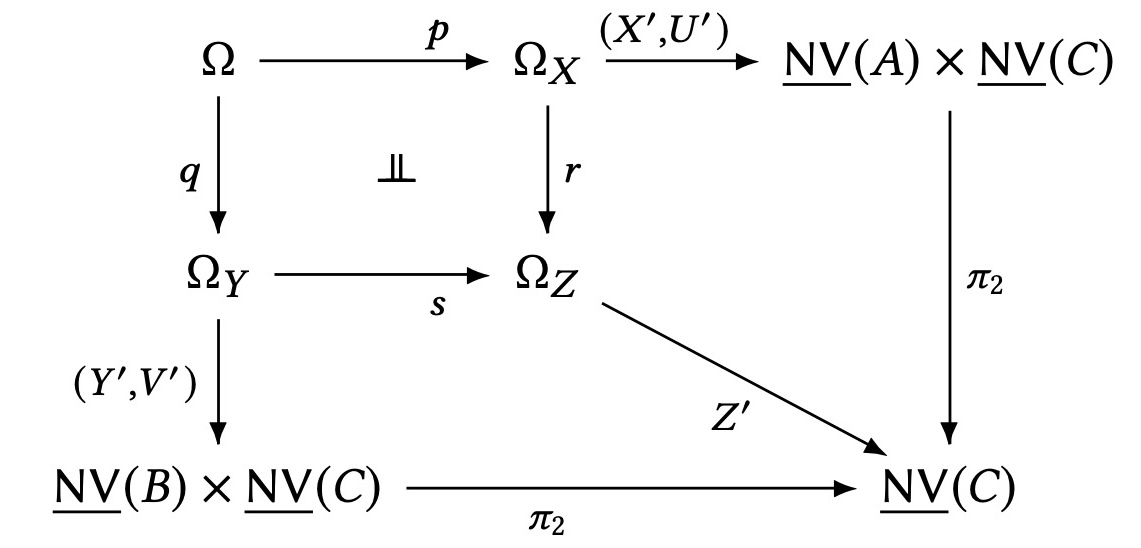}}}
%\begin{diagram}
%\Omega & \rTo^{p} & \Omega_X  & \rTo^{\!\!\!\!\!{(X',U')}\;\;} & \NEls(A) \times \NEls(C) \\
%\dTo^{q} & \!\!\!\!\!\!\!\!\!\!\Indep &  \dTo_r &  &  \\
%\Omega_Y & \rTo_s & \Omega_Z & & \dTo_{\pi_2} \\
%\dTo^{(Y',V')} & & & \rdTo_{Z'} &  \\
%\NEls(B) \times \NEls(C) & & \rTo_{\pi_2} & & \NEls(C)
%\end{diagram}
\]
We show that $X \Indep Y \Cond Z$, according to
Definition~\ref{definition:conditional-independence-NV}.
Suppose we have $\omega', \omega'' \in \Omega$ such that $X(\omega') = a$ and $Z(\omega')= c$ and
$Y(\omega'') = b$ and $Z(\omega'') = c$. Then
\[
Z'(r(p(\omega'))) = Z(\omega') = Z(\omega'') = Z'(s(q(\omega''))) \enspace .
\]
Since $(\Omega_Z, \, r \circ p, \, Z')$ is support for $Z$, the function $Z' \colon \Omega_Z \to C$ is injective,
hence $r(p(\omega')) = s(q(\omega''))$. Since the top-left square is independent, there exists $\omega \in \Omega$
such that $p(\omega) = p(\omega')$ and $q(\omega) = q(\omega'')$. Then
$X(\omega) = X'(p(\omega)) = X'(p(\omega')) = X(\omega') = a$. Similarly, $Y(\omega) = b$, and $Z(\omega) = c$.
\end{proof}

We now turn to the extension of the atomic sheaf logic of Sections~\ref{section:atomic-sheaf-logic} and~\ref{section:atomic-equivalence} with
conditional independence formulas~\eqref{equation:indep-formula}.
Once again, we view this extension as being obtained by including a family of relation symbols.
In this case we add relations
$\sIndep_{\vec{\Syn{A}}, \vec{\Syn{B}} | \vec{\Syn{C}}}$, and require that each such relation  is interpreted as the subsheaf
\[{\Indep}_{\Sh{\vec{\Syn{A}}}, \Sh{\vec{\Syn{B}}} | \Sh{\vec{\Syn{C}}}}
 ~ \subseteq ~ \Sh{\vec{\Syn{A}}} \times \Sh{\vec{\Syn{B}}} \times  \Sh{\vec{\Syn{C}}} \enspace ,
\]
where we write, e.g., $\Sh{\vec{\Syn{A}}}$ for the product $\prod_{i = 1}^n \Sh{\Syn{A}_n}$, where
$\vec{\Syn{A}}$ is the vector of sorts $\Syn{A}_1, \dots, \Syn{A}_n$.
To ensure that ${\Indep}_{\Sh{\vec{\Syn{A}}}, \Sh{\vec{\Syn{B}}} | \Sh{\vec{\Syn{C}}}}$ is well defined, 
we require that every sort $\Syn{A}$ is interpreted as a sheaf 
$\Sh{\Syn{A}}$ with supports.

\begin{figure}
\begin{gather}
\vec{\Syn{x}} \sIndep \vec{\Syn{y}} \Cond \vec{\Syn{z}}
~ \to ~ \pi(\vec{\Syn{x}}) \sIndep \pi'(\vec{\Syn{y}}) \Cond \pi''(\vec{\Syn{z}})
\label{indep:P}
\\
\vec{\Syn{x}} \sIndep \vec{\Syn{y}} \Cond \vec{\Syn{y}}
\label{indep:A}
\\
\vec{\Syn{x}} \sIndep \vec{\Syn{y}} \Cond \vec{\Syn{z}} ~ \to ~ 
\vec{\Syn{y}} \sIndep \vec{\Syn{x}} \Cond \vec{\Syn{z}}
\label{indep:B}
\\
\vec{\Syn{x}} \sIndep \vec{\Syn{y}},\vec{\Syn{z}} \Cond \vec{\Syn{w}} ~ \to ~ 
\vec{\Syn{x}} \sIndep \vec{\Syn{y}} \Cond  \vec{\Syn{w}}
\label{indep:W}
\\
\vec{\Syn{x}} \sIndep \vec{\Syn{y}},\vec{\Syn{z}} \Cond \vec{\Syn{w}} ~ \to ~ 
\vec{\Syn{x}} \sIndep \vec{\Syn{y}} \Cond \vec{\Syn{z}}, \vec{\Syn{w}}
\label{indep:C}
\\
\vec{\Syn{x}} \sIndep \vec{\Syn{y}} \Cond \vec{\Syn{z}}, \vec{\Syn{w}}
~ \wedge ~ \vec{\Syn{x}} \sIndep  \vec{\Syn{z}}  \Cond \vec{\Syn{w}}
~ \to ~ 
\vec{\Syn{x}} \sIndep \vec{\Syn{y}}, \vec{\Syn{z}}\Cond \vec{\Syn{w}}
\label{indep:D}
\\[1ex]
%\vec{\Syn{x}} \sIndep \vec{\Syn{y}} \Cond  \vec{\Syn{w}}
%~ \wedge ~ \vec{\Syn{x}} \sIndep \vec{\Syn{z}} \Cond  \vec{\Syn{w}}
%~ \wedge ~ \vec{\Syn{y}}, \vec{\Syn{w}} \Equiv \vec{\Syn{z}}, \vec{\Syn{w}}
%~ \to ~ \vec{\Syn{x}}, \vec{\Syn{y}}, \vec{\Syn{w}} \Equiv \vec{\Syn{x}},\vec{\Syn{z}}, \vec{\Syn{w}}
%\label{indep:Y}
% \\
\exists \vec{\Syn{y}}.~( \vec{\Syn{y}}, \vec{\Syn{w}} \Equiv \vec{\Syn{x}}, \vec{\Syn{w}} ~ \wedge ~
 \vec{\Syn{y}} \sIndep \vec{\Syn{z}} \Cond  \vec{\Syn{w}})
 \label{indep:Z}
% ~~~\text{($\FV(\Phi) \subseteq \{\vec{\Syn{x}}\}$)}
\end{gather}
\caption{Axioms for conditional independence}
\label{figure:independence-axioms}
\end{figure}

Figure~\ref{figure:independence-axioms} lists formulas valid in this semantics that we single out
as a suitable list of axioms for reasoning about conditional  independence.
Axiom~\eqref{indep:P} asserts that conditional independence is preserved under permutations within each of the three lists of variables involved. This axiom, together with axioms~\eqref{indep:A}--\eqref{indep:D} are all standard axioms for conditional independence, appearing in closely related forms in~\cite[Theorem 3.1 and Lemmas 4.1--4.3]{dawidA}, in \cite[Theorem~1]{spohn} and in the work of Pearl, Paz and Geiger~\cite{PP,GPP,GP} (in which only conditional independence statements of the restricted form
$\vec{\Syn{x}} \sIndep \vec{\Syn{y}} \Cond \vec{\Syn{z}}$ for three disjoint sets of variables $\vec{\Syn{x}}$, $\vec{\Syn{y}}$ and $\vec{\Syn{z}}$ are considered). The axioms appear more explicitly in their present form in Dawid's axioms for the notion of \emph{separoid}~\cite{dawidB}.
We leave the straightforward verification of the soundness of axioms~\eqref{indep:P}--\eqref{indep:W} to the reader. 
The soundness of axioms~\eqref{indep:C} and~\eqref{indep:D} is more technical. To avoid encumbering the main development with these technical proofs, they are given in Appendix~\ref{appendix:CD}.

Whereas axioms~\eqref{indep:P}--\eqref{indep:D} concern conditional independence in isolation,
axiom~\eqref{indep:Z} captures the interaction between conditional independence and atomic equivalence.
%Axiom~\eqref{indep:Y}, which we call the \emph{principle of independent equivalence}, states that, if we vary 
% $\vec{\Syn{y}}$ over tuples that are independent from $\vec{\Syn{x}}$ conditional on $\vec{\Syn{w}}$, then the joint conditional (on $\vec{\Syn{w}}$) `distribution' of $\vec{\Syn{x}}, \vec{\Syn{y}}$ is determined by the conditional `distribution' 
% of $\vec{\Syn{y}}$ alone.\footnote{We here refer to `distribution' informally, to convey a probabilistic intuition for the stated property. We moreover talk about  the \emph{conditional equivalence} of $\vec{\Syn{y}}$ and $\vec{\Syn{z}}$ given $\vec{\Syn{w}}$ as an intuitive way of understanding the joint equivalence $\vec{\Syn{y}}, \vec{\Syn{w}} \Equiv \vec{\Syn{z}}, \vec{\Syn{w}}$.}
Axiom~\eqref{indep:Z} makes essential use of the existential quantifier of atomic sheaf logic to
capture a key first-order  property:
given variables $\vec{\Syn{x}}, \vec{\Syn{z}}, \vec{\Syn{w}}$ one can always find variables $\vec{\Syn{y}}$
that are conditionally  independent from $\vec{\Syn{z}}$ given $\vec{\Syn{w}}$, but such that
$\vec{\Syn{y}},\vec{\Syn{w}}$ is jointly
equivalent to
$\vec{\Syn{x}},\vec{\Syn{w}}$.
% (one way of understanding this latter property is:
%$\vec{\Syn{y}}$ is equivalent to
%$\vec{\Syn{x}}$ conditionally on $\vec{\Syn{w}}$).
 We call this property the \emph{independent existence principle}: independent  variables with any desired distribution always exist.
The validity of the principle 
% of independent equivalence~\eqref{indep:Y} and
of independent existence~\eqref{indep:Z} is established by 
Lemma~\ref{lemma:validity-Z} below.

\begin{lemma} 
\label{lemma:validity-Z}
Given $x \in \Sh{A}(X)$, $\,z \in \Sh{B}(X)$ and $w \in \Sh{C}(X)$,
there exist $p \colon  Y \rMap{} X$ and $y \in \Sh{A}(Y)$ such that
\begin{align}
\label{equation:blob}
((y,\, w \cdot p), \, (x \cdot p,\,  w \cdot p)) & \in \,{\Equiv_{\Sh{A} \times \Sh{C}}(Y)} \\
\label{equation:blobi}
(y, \, z \cdot p, \, w \cdot p) & \in \, {{\Indep}_{\Sh{A},\Sh{B} \Cond\Sh{C}}}(Y) \enspace .
\end{align}
\end{lemma}
\begin{proof}
Let $(Z, s,w')$ be support for $w$, and consider the independent pullback of $s: X \to Z$ along itself:
\[ 
\vcenter{\hbox{\includegraphics[scale=0.18]{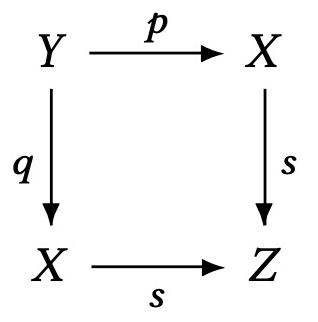}}}
%\begin{diagram}
%Y & \rTo^p & X \\
%\dTo^{q} & & \dTo_s \\
%X & \rTo_s & Z
%\end{diagram}
\]
We have:
\[ w \cdot p = w' \cdot s \cdot p = w' \cdot s \cdot q = w \cdot q  \enspace .\]
 Define $y := x \cdot q$.
 
 By the independent pullback above, there is a unique map $t : Y \to Y$ such that $p \circ t = q$ and $q \circ t = p$.
 So:
  \[
 (y,\, w \cdot p)  =  (x \cdot q,\, w \cdot q) ~ = ~ (x \cdot p \circ t,\, w \circ p \circ t) \enspace .
 \]
 Thus the pair $\Id_Y, t \colon Y \rMap{} Y$ shows that~\eqref{equation:blob} holds. 
  
 For the independence statement, we have: 
 \[
\vcenter{\hbox{\includegraphics[scale=0.18]{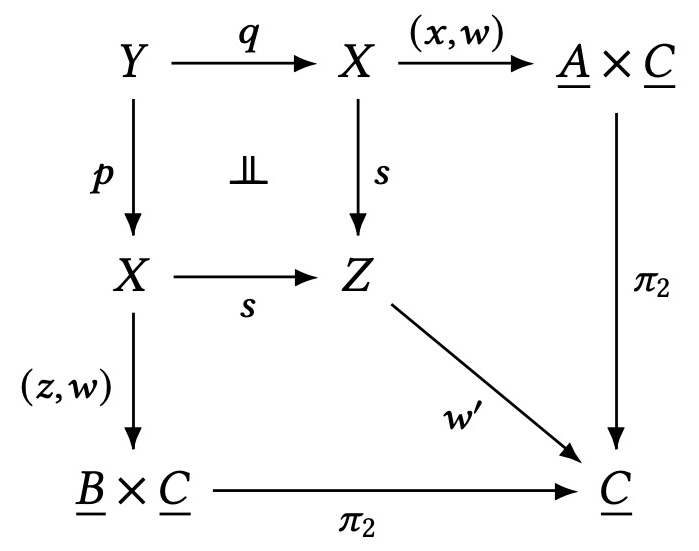}}}
%\begin{diagram}
%Y & \rTo^{q} & X  & \rTo^{\!\!\!\!\!{(x,w)}\;\;} & \Sh{A} \times \Sh{C} \\
%\dTo^{p} & \Indep &  \dTo_s &  &  \\
%X & \rTo_s & Z & & \dTo_{\pi_2} \\
%\dTo^{(z,w)} & & & \rdTo_{w'} &  \\
%\Sh{B} \times \Sh{C} & & \rTo_{\pi_2} & & \Sh{C}
%\end{diagram}
  \]
 The first component of the top side is $x \cdot q  = y \in \Sh{A}(Y)$. The first component of the left side
 is $z \cdot p \in \Sh{B}(Y)$. Moreover, by Lemma~\ref{lemma:weak-supports},
 $(Z, \, s\circ p, \, w')$ is support for $w \cdot p$. Thus we indeed have \eqref{equation:blobi}.
\end{proof}

As an interesting consequence of the axioms, we prove that existence properties  are preserved under conditional independence, in the sense of the result below. This provides a first-order reasoning principle for conditional independence, whose scope potentially extends beyond atomic sheaf logic to more general contexts in which there is a conditional independence relation but no analogue of the
relation $\Equiv$ of atomic equivalence. 
\begin{theorem}[Existence preservation]
The schema below follows from the axioms in Figs.~\ref{figure:equivalence-axioms}  and~\ref{figure:independence-axioms}.
\begin{align*}
& (\exists \vec{\Syn{y}}.~ \Phi(\vec{\Syn{x}},\vec{\Syn{y}},\vec{\Syn{w}})) ~ \to ~
% \\ & \qquad
\forall \vec{\Syn{z}}. ~ (\, \vec{\Syn{x}} \sIndep \vec{\Syn{z}} \Cond \vec{\Syn{w}}
~ \to ~ \exists \vec{\Syn{y}}. ~(\,\vec{\Syn{x}}, \vec{\Syn{y}} \sIndep \vec{\Syn{z}} \Cond \vec{\Syn{w}} ~ \wedge ~
      \Phi(\vec{\Syn{x}},\vec{\Syn{y}},\vec{\Syn{w}})))
 \end{align*}
 Here we adopt the same convention as in the invariance principle. In $\Phi(\vec{\Syn{x}},\vec{\Syn{y}},\vec{\Syn{w}})$ every free variable in $\Phi$ has been substituted by one of the variables in $\vec{\Syn{x}},\vec{\Syn{y}},\vec{\Syn{w}}$.
\end{theorem}
\begin{proof}
Let  $\vec{\Syn{y}}$ be such that 
\begin{equation}
\label{ep:A}
\Phi(\vec{\Syn{x}},\vec{\Syn{y}},\vec{\Syn{w}})\, .
\end{equation}
Consider any $\vec{\Syn{z}}$. 
By the independent existence principle~\eqref{indep:Z}, there exists $\vec{\Syn{y}'}$ such that
\begin{equation}
\label{ep:B}
\vec{\Syn{y}'}, \vec{\Syn{x}}, \vec{\Syn{w}} \Equiv \vec{\Syn{y}}, \vec{\Syn{x}}, \vec{\Syn{w}}
\end{equation}
and 
\begin{equation}
\label{ep:C}
\vec{\Syn{y}'} \sIndep \vec{\Syn{z}} \Cond  \vec{\Syn{x}}, \vec{\Syn{w}} \, .
\end{equation}
Suppose 
\begin{equation}
\label{ep:D}
\vec{\Syn{x}} \sIndep \vec{\Syn{z}} \Cond \vec{\Syn{w}} \, .
\end{equation}
Then \eqref{ep:C} and \eqref{ep:D} combine to give 
%By~\eqref{indep:B} and~\eqref{indep:D}, 
$\vec{\Syn{x}},\vec{\Syn{y}'} \sIndep \vec{\Syn{z}} \Cond   \vec{\Syn{w}}$, by the axioms for conditional independence.

Further, \eqref{ep:A} and \eqref{ep:B} combine to give 
$\Phi(\vec{\Syn{x}},\vec{\Syn{y}'},\vec{\Syn{w}})$, by the invariance principle \eqref{equiv:F}.
 \end{proof}

%%%%%%%%%%%%%%%%%%%%%%%%%%%%%%%%%%%%%%%%%
%%% SECTION
%%%%%%%%%%%%%%%%%%%%%%%%%%%%%%%%%%%%%%%%%

\section{Probability sheaves}
\label{section:probability-sheaves}

In this long section, we present another instance of our axiomatic structure:  atomic sheaves over  \emph{standard Borel probability spaces}.
The idea is that such spaces take the role of  {sample spaces}, and random variables over such sample spaces collectively form an atomic sheaf. More precisely, for any standard Borel space $A$, we shall obtain a sheaf $\SRV(A)$ of all
$A$-valued random variables.  
For this aim, the standard-Borel assumption
 serves three purposes. Firstly, it is sufficiently general that it encompass both discrete and continuous probability.
Secondly, it 
provides a \emph{small} category of sample spaces to build atomic sheaves over.
Finally, it also provides useful technical machinery (such as \emph{disintegrations} of random variables), which would be unavailable in general if arbitrary probability and measurable spaces were used. This machinery is essential in showing that the category of sample spaces has independent pullback structure. 
When interpreted over the sheaves  of 
random variables $\SRV(A)$, 
atomic sheaf logic provides logical
principles governing  the relations of almost sure equality,
of equality in distribution and of conditional independence with its standard  probabilistic meaning, since these  three relations are respectively encapsulated as equality, atomic equivalence and atomic conditional independence in the logic.

In order to fully understand the technical development in the present section, it is necessary to have some background in probability and measure theory. Nevertheless, we try to also explain the main ideas informally, so help readers without the necessary background to follow the line of development at a high level. 

Standard Borel spaces will be the value spaces of random variables, and they will also be the structures over which we build sample spaces.
\begin{definition}[Standard Borel space]
\label{definition:sbs}
A \emph{standard Borel space} (SBS) is a measurable space $(A, \mathcal{B}_A)$ where $A$ is a Borel subset  of a Polish space $T$ (i.e., a complete separable metric space) and $\mathcal{B}_A$  is 
the $\sigma$-algebra $\{S \cap A \mid \text{$S \subseteq T$ is Borel}\}$.  A \emph{morphism} of standard Borel spaces from 
$(A ,\mathcal{B}_A)$ to $(B , \mathcal{B}_{B})$ is a function $f: A \to B$  that is \emph{measurable}, i.e.,
$f^{-1}(S) \in \mathcal{B}_A$ for all $S \in \mathcal{B}_{B}$.
\end{definition}
\noindent
When $(A ,\mathcal{B}_A)$ is a standard Borel space, we shall refer to the sets in $\mathcal{B}_A$ as the \emph{Borel subsets} of $A$, which is justified because $A$ can always itself be given a Polish topology in which  $\mathcal{B}_A$ is the Borel $\sigma$-algebra. As is well known,  the image $f(C)$ of a Borel subset $C \subseteq A$ under a measurable function  $f: A \to B$, where
$(B , \mathcal{B}_{B})$ is also standard Borel, need not itself be a Borel subset of $B$, but $f(C)$ is always an \emph{analytic} subset of $B$.

On the one hand, the collection of standard Borel spaces is very rich, as it incorporates most measurable spaces that arise naturally in mathematics. On the other, it is  also limited, since there are only two types of standard Borel spaces: (i) spaces $(A,\mathcal{P}(A))$, where $A$ is a \emph{countable} (possibly finite) set with its full powerset $\mathcal{P}(A)$ as the $\sigma$-algebra; and (ii)
spaces $(A, \mathcal{B}_A)$ that are isomorphic to the real numbers with the Borel $\sigma$-algebra $(\R, \mathcal{B})$. As a consequence of this classification, every standard Borel space has a measurable embedding into the interval $[0,1]$ with the Borel $\sigma$-algebra $\mathcal{B}_{[0,1]}$.

Standard Borel probability spaces will act as our sample spaces. As such, they will provide the objects of the category of sample spaces over which we shall consider atomic sheaves. 
\begin{definition}[Standard Borel probability space]
\label{definition:sbps}
 A \emph{standard Borel probability space} (SBPS) is a triple 
$(\Omega ,\mathcal{B}_\Omega, P_\Omega)$ where $(\Omega, \mathcal{B}_\Omega)$ is an SBS  and $P_\Omega \colon \mathcal{B}_\Omega \to [0,1]$ is a probability measure. A \emph{morphism} of standard Borel probability spaces from 
$(\Omega ,\mathcal{B}_\Omega, P_\Omega)$ to $(\Omega' , \mathcal{B}_{\Omega'}, P_{\Omega'})$ is an SBS morphism
$q$ from
$(\Omega, \mathcal{B}_\Omega)$ to $(\Omega' , \mathcal{B}_{\Omega'})$  that \emph{preserves measure}; i.e., $q_*(P_\Omega) = P_{\Omega'}$, where $q_*(P)$ is the \emph{pushforward measure} $S \mapsto P_{\Omega}(q^{-1}(S)) : \mathcal{B}_{\Omega'} \to [0,1]\,$. 
\end{definition}
As with standard Borel spaces, standard Borel probability spaces include the most common probability spaces that one  naturally encounters in mathematics. Any standard Borel probability space $(\Omega ,\mathcal{B}_\Omega, P_\Omega)$ can be decomposed uniquely into 
its \emph{discrete} and \emph{continuous} parts, moreover the continuous part has a very constrained form.
In detail, there exist unique Borel measures
$\delta, \mu : \mathcal{B}_\Omega \to [0,1]$ such that $P_\Omega = \delta + \mu$, the measure $\delta$ is \emph{discrete} (i.e.,
$\delta(B) = \sum_{x \in B} \delta(\{x\})$ for every $B \in \mathcal{B}_\Omega$), and 
either $\mu = 0$ or $(\Omega, \mathcal{B}_\Omega, \mu)$ is isomorphic, via measure-preserving functions,  to the interval $([0,c], \mathcal{B}_{[0,c]}, \lambda)$,
where $c := P_\Omega(\Omega)$,  with the
Borel $\sigma$-algebra $\mathcal{B}_{[0,c]}$ and the (Borel restriction of) Lebesgue measure $\lambda : \mathcal{B}_{[0,c]} \to [0,c]$.

In probability theory, a {random variable} is a measurable function from a probability space, called the {sample space}, to a measurable space, the {value space}. In this paper, we restrict ourselves to the case in which these spaces are both standard Borel. This is broad enough to incorporate both the discrete and continuous random variables arising most commonly in mathematics.
\begin{definition}[Random variable]
If $\Omega$ is an SBPS and $A$ is an SBS
(for notational convenience we here and henceforth abbreviate  $(A ,\mathcal{B}_A)$ as  $A$ and
$(\Omega ,\mathcal{B}_\Omega, P_\Omega)$ as $\Omega$),
a \emph{random variable} $X : \Omega \to A$ is a measurable function from
$(\Omega, \mathcal{B}_\Omega)$ to $(A,  \mathcal{B}_A)$. The SBPS $\Omega$ is called the \emph{sample space} of $X$, and
the SBS $A$ is called the \emph{value space}.
\end{definition}

We next define the three main relations between random variables we shall be interested in: 
\emph{almost-sure equality}, \emph{equidistribution} and \emph{conditional independence}.

In general, we say that a property of elements $\omega \in \Omega$ holds 
\emph{for $P_\Omega$-almost-all $\omega$} if there exists  $S \in \mathcal{B}_\Omega$ with $P_\Omega(S) = 1$
such that the property holds for every $\omega \in S$.
\begin{definition}[Almost-sure equality]
Two random variables $X,Y: \Omega \to A$ are \emph{almost surely equal} (notation $X \EqAE Y$) if
 $X(\omega) = Y(\omega)$ holds for $P_\Omega$-almost-all $\omega$.
 (Since $A$ is a standard Borel space, the set 
 $\{\omega \in \Omega \mid X(\omega) = Y(\omega)\}$ is measurable, and the above condition is 
 equivalent to asking that 
 $P_\Omega(\{\omega \in \Omega \mid X(\omega) = Y(\omega)\}) = 1$.)
 \end{definition}
 
The \emph{distribution} (or \emph{law}) of a random variable $X: \Omega \to A$ is the probability measure
$P_X : \mathcal{B}_A \to [0,1]$ defined as the pushforward $P_X := X_*(P_\Omega)$.
\begin{definition}[Equidistribution] Two random variables 
$X,Y : \Omega \to A$ are \emph{equidistributed} (notation $X \eqdist Y$) if 
$P_X = P_Y$. 
\end{definition}

An important consequence of only considering random variables  between standard Borel spaces is that random variables have \emph{disintegrations}. We state this property as Fact~\ref{fact:disintegrations} below. A proof of can be found in~\cite{DM}. We mention also that an equivalent statement to Fact~\ref{fact:disintegrations}  appears as Theorem~6 of~\cite{DDGK}.
\begin{fact}
\label{fact:disintegrations}
Every random variable  $X  \colon \Omega \to A$ has a \emph{disintegration}; that is, 
a Markov kernel $D_X : A \times \mathcal{B}_{\Omega} \to [0,1]$
\[(x, S)~ \mapsto~ P_{X^{-1}(x)}(S) ~: ~ A \times \mathcal{B}_{\Omega} \to [0,1] \]
%$
%D_X \colon A \times \mathcal{B}_{\Omega} \to [0,1]
%$
satisfying the two properties below. %(we use the suggestive notation $P_{\Omega} (S \mid X = x)$ for $D_X(x, S)$)\,:
\begin{description}
\item[(D1)] $P_{X^{-1}(x)} (X^{-1}(x)) = 1$ for $P_X$-almost all $x\in A$, and
%$P_{\Omega} (X^{-1}(x) \Cond X = x) = 1$ for $P_X$-almost all $x\in A$, and
\item[(D2)]
for every $S \in \mathcal{B}_{\Omega}$,
\[
  P_{\Omega} (S) ~ = ~ \int P_{X^{-1}(x)} (S)\;  \mathrm{d}P_X (x)\enspace .
\]
\end{description}
%Moreover, if $D_X, D'_X$ are two disintegrations then $D_X(x,-) = D'_X(x,-) : \mathcal{B}_{\Omega} \to [0,1]$
%for $P_X$-almost all $x\in S$.
\end{fact}
\noindent
By the Markov kernel property, the function 
$S \mapsto P_{X^{-1}(x)}(S)$ is  a probability measure $P_{X^{-1}(x)} : \mathcal{B}_\Omega \to [0,1]$, for every
$x \in A$. By (D1),  $P_{X^{-1}(x)}$ can be thought of as a probability measure on the fibre set $X^{-1}(x) \in \mathcal{B}_\Omega$, which, by (D2), represents the conditional probability distribution on $\omega \in \Omega$ under the condition 
 $X(\omega) = x$. Properties (D1) and (D2) together characterise the mapping $x \mapsto P_{X^{-1}(x)}$ up to
 $P_\Omega$-almost-sure equality.

%In the sequel, we shall similarly use the notation $P_{\Omega} ({-} \Cond X = x)$ for the probability measure 
%$S \mapsto P_{\Omega} (S \Cond X = x)$.

Exploiting disintegrations, we  give a definition of conditional independence that is a transparent generalisation of the elementary probabilistic definition of unconditional independence.

\begin{definition}[Conditional independence] 
\label{definition:conditional-independence-RV}
For random variables $X : \Omega \to A$, $~Y: \Omega \to B$ and $Z : \Omega \to C$,
we say that $X$ and $Y$ are \emph{conditionally independent given $Z$} (notation $X \Indep Y \Cond Z$) if,  for 
every $S \in \mathcal{B}_A$ and $T \in \mathcal{B}_B$, and for
$P_Z$-almost all $z \in C$, 
\begin{align*}
%& P_{\Omega} (X^{-1}(S)  \cap Y^{-1}(T) \Cond Z = z)  ~=  ~
%P_{\Omega} (X^{-1}(S)  \Cond Z = z) \cdot P_{\Omega} (Y^{-1}(T) \Cond Z = z) \enspace .
& P_{Z^{-1}(z)} (X^{-1}(S)  \cap Y^{-1}(T))  ~=  ~
P_{Z^{-1}(z)} (X^{-1}(S)) \cdot P_{Z^{-1}(z)} (Y^{-1}(T)) \enspace .
\end{align*}
\end{definition}

Our goal in this section is to recover the  three principal relations between random variables (almost-sure equality, equidistribution and conditional independence) 
as the relations of equality, atomic  equivalence and atomic conditional independence in a suitable atomic sheaf topos. 
In order to be able to construct sheaves of random variables, the category over which sheaves will be taken is a category of sample spaces. In fact we consider two such categories.

\begin{definition}[The categories $\sBP$ and $\sBP_0$]
We write  $\sBP$ for a small category of standard Borel probability spaces, that contains every such space up to isomorphism.
We write $\sBP_0$ for the quotient category, with the same objects, in which morphisms are equivalence classes $[p]$ of maps modulo
almost-sure equality $\EqAE$. 
%(Note that every map $ q :  \Omega \rMap{} \Omega'$ is a fortiori a random variable with value space $\Omega'$.)
\end{definition}

\noindent
It is an interesting fact that one can take the category of atomic sheaves over either category, $\sBP$ or $\sBP_0$, and in doing so one obtains equivalent categories of sheaves. Sheaves for the atomic topology on $\sBP$ were 
introduced in~\cite{simpsonB} as \emph{probability sheaves}.
In the present  paper, it will be convenient to  instead take  atomic sheaves over $\sBP_0$. 
Since the two categories of sheaves are equivalent, we shall continue to use the name \emph{probability sheaves}. The equivalence of the two categories will be shown in a separate paper. 

An important advantage of working with $\sBP_0$ is the property below, which fails for $\sBP$.
\begin{proposition}
\label{proposition:sBPZ-epi}
Every morphism in $\sBP_0$ is an epimorphism.
\end{proposition}
\begin{proof}
We first observe that every map $q : \Omega \rMap{} \Omega'$ in $\sBP$ is \emph{almost surjective} in the sense that,
for any $S \in \mathcal{B}_{\Omega}$ with $P_{\Omega}(S) = 1$, 
there exists $T \subseteq q(S)$ such that $T \in \mathcal{B}_{\Omega'}$ and $P_{\Omega'}(T) = 1$.
This holds because the image $q(S)$ is an analytic subset of $\Omega'$ with outer measure $1$. Since all analytic sets are measurable with respect to the completion of the Borel measure $P_{\Omega'}$, the image $q(S)$ also has inner measure $1$, meaning that there exists $T \subseteq q(\Omega)$ with the required properties.

To prove that every morphism in $\sBP_0$ is epimorphic, suppose we have $[q]: \Omega \rMap{} \Omega'$ and 
$[r],[r']: \Omega' \rMap{} \Omega''$ such that $[r] \circ [q] = [r'] \circ [q]$; i.e., $r \circ q \EqAE r' \circ q$. 
Let $S \subseteq \Omega$ be Borel such that  $P_\Omega(S) = 1$ and
$(r \circ q) \! \restriction_S = (r' \circ q) \! \restriction_S$. 
By the almost surjectivity of $q$, let $T \subseteq q(S)$ be such that $T \in \mathcal{B}_{\Omega'}$
and $P_{\Omega'}(T) = 1$. Then $r  \! \restriction_T = r'  \! \restriction_T$; i.e.,
$r \EqAE r'$. Equivalently $[r] = [r']$ as required.
\end{proof}

\begin{proposition}
\label{proposition:sBP-pairing}
The category $\sBP_0$ has pairings.
\end{proposition}
\begin{proof}
Given any  span $\Omega_Y \lMap{[p]}  \Omega_X  \rMap{[q]} \Omega_Z$ in $\sBP_0$,
its pairing is given by $(\Omega, \, [(p,q)],\, \pi_1, \pi_2)$, where
$
\Omega := (\Omega_Y \times \Omega_Z, \mathcal{B}_{\Omega_Y \times\Omega_Z}, P_{(p,q)})\,$,
using the product standard Borel space and the probability distribution of the paired random variables $p$ and $q$.
The properties of a pairing are easily verified, using Proposition~\ref{proposition:sBPZ-epi} for uniqueness.
\end{proof}

\begin{definition}[Independent square in $\sBP_0$]
\label{definition:indep-square-sbp}
Define a commuting square in $\sBP_0$
\begin{equation}
\label{equation:square-sbp}
\vcenter{\hbox{\includegraphics[scale=0.18]{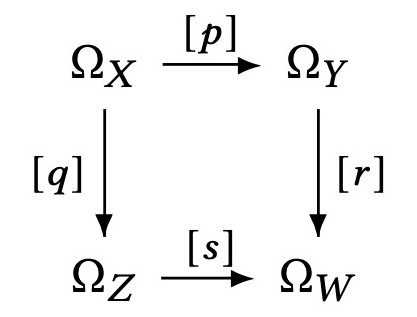}}}
%\begin{diagram}
%{\Omega_X} & \rTo^{[p]} & {\Omega_Y} \\
%\dTo^{[q]} & & \dTo_{[r]} \\
%{\Omega_Z} & \rTo^{[s]} & {\Omega_W}
%\end{diagram}
\end{equation}
to be \emph{independent} if $p \Indep q \Cond r \circ p$, using conditional independence of random variables (Definition~\ref{definition:conditional-independence-RV}). 
\end{definition}

\begin{proposition}
\label{proposition:sBP-indep-pullbacks}
Definition~\ref{definition:indep-square-sbp} endows $\sBP_0$ with independent pullback structure satisfying the descent property. %, all its maps are epimorphic 
%and it has pairings.
\end{proposition}

The proof of Proposition~\ref{proposition:sBP-indep-pullbacks}, which is intricate, can be found in Appendix~\ref{appendix:sbp}. In the present section, we content ourselves with exhibiting the construction
needed to complete any cospan
$\Omega_Y   \rMap{[r]}   \Omega_W  \lMap{[s]}  \Omega_Z$ 
to an independent pullback.
Using the disintegrations for $r$ and $s$, we endow  
the standard Borel product $(\Omega_Y \times \Omega_Z, \mathcal{B}_{\Omega_Y \times \Omega_Z})$
with the probability measure $P$ defined as:
\begin{equation}
\label{equation:pullback-measure}
  U ~  \mapsto ~ \int (P_{r^{-1}(\omega)} \otimes 
      P_{s^{-1}(\omega)}) (U) ~ \mathrm{d}P_{\Omega_W}(\omega)\, ,
\end{equation}
where $P_{r^{-1}(\omega)} \otimes 
      P_{s^{-1}(\omega)}$ is the product probability measure. Then 
\[
(\Omega_Y \times \Omega_Z, \mathcal{B}_{\Omega_Y \times\Omega_Z}, P)
\]
together with the two projections, which are measure preserving, gives the required independent pullback.
We write the resulting independent pullback square as
\[
\vcenter{\hbox{\includegraphics[scale=0.18]{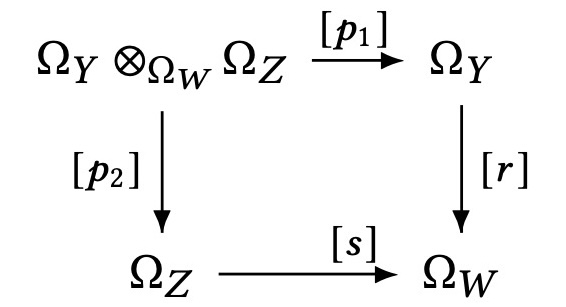}}}
%\begin{diagram}
%{\Omega_Y \otimes_{\Omega_W}\! \Omega_Z} & \rTo^{[p_1]} & {\Omega_Y} \\
%\dTo^{[p_2]} & & \dTo_{[r]} \\
%{\Omega_Z} & \rTo^{[s]} & {\Omega_W}
%\end{diagram}
\]

In combination, Propositions~\ref{proposition:sBPZ-epi}, \ref{proposition:sBP-pairing} and~\ref{proposition:sBP-indep-pullbacks} show that the category $\sBP_0$ has the requisite structure (Definition~\ref{definition:requisite}).

We next define the anticipated sheaves of random variables, first by defining them as  presheaves, 
and then subsequently verifying the atomic sheaf property.
\begin{definition}[Presheaf of random variables $\SRV(A)$]
Let $A$ be a standard Borel space.
Define a presheaf
$\SRV(A) \in \Psh(\sBP_0)$  of $A$=valued random variables (modulo $\EqAE$) by:
\begin{itemize}
\item $\SRV(A)(\Omega) :=$ equivalence classes of random variables $X : \Omega \to A$  modulo $\EqAE$. 

\item For $[X] \in \SRV(A)(\Omega)$ and $[q] \colon \Omega' \to \Omega$, define
$[X] \cdot [q ]:= [X \circ q]$.
\end{itemize}
\end{definition}

\noindent
We remark that a similar definition can be used to define a presheaf of $A$-valued random variables
modulo $\EqAE$ over the base category $\sBP$. In the case that $\sBP$ is used as the base category, one can also define 
an alternative presheaf of random variables, in which random variables are not quotiented modulo $\EqAE$, an option which is not available when $\sBP_0$ is used as the base category. The $\sBP$ -presheaf of unquotiented $A$-valued random variables is not, however, an atomic sheaf. In contrast, irrespective of the choice of base category, $\sBP$ or $\sBP_0$, the presheaf of random variables modulo $\EqAE$
 does form  a sheaf. We prove this in the case  of our chosen base category, $\sBP_0$.

\begin{proposition}
For any standard Borel space $A$, it holds that $\SRV(A)$ is an atomic sheaf.
\end{proposition}
\begin{proof}
Suppose $[Y] \in \SRV(A)(\Omega')$ is $[q]$-invariant where $\Omega' \rMap{q} \Omega$.
is a map in  $\sBP_0$. Consider the independent pullback square
\[
\vcenter{\hbox{\includegraphics[scale=0.18]{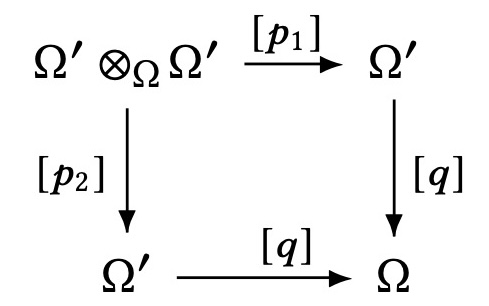}}}
%\begin{diagram}
%{\Omega' \otimes_{\Omega}\! \Omega'} & \rTo^{[p_1]} & {\Omega'} \\
%\dTo^{[p_2]} & & \dTo_{[q]} \\
%{\Omega'} & \rTo^{[q]} & {\Omega}
%\end{diagram}
\]
By $[q]$-invariance, $[Y] \cdot [p_1] = [Y] \cdot [p_2]$, i.e., $Y \circ p_1 \EqAE Y \circ p_2$.
That is, the measure of \[U ~ := ~  \{(\omega'_1,\omega'_2) \in \Omega' \times \Omega' \mid Y(\omega_1) = Y(\omega_2)\}\]
in $\Omega' \otimes_{\Omega} \Omega'$ is $1$. Equivalently, using~\eqref{equation:pullback-measure},
\[
\int (P_{q^{-1}(\omega)} \otimes 
      P_{q^{-1}(\omega)}) (U) ~ \mathrm{d}P_{\Omega}(\omega) ~ = ~ 1 .
\]
So, for $P_{\Omega}$-almost all $\omega \in \Omega$, we have
\[ (P_{q^{-1}(\omega)} \otimes 
      P_{q^{-1}(\omega)}) (U) ~ = ~ 1 \enspace .\]
 For any such $\omega$, by the definition of product measure,
 \[
 \int \int \, \One_{U}(\omega'_1,\omega'_2) \, \mathrm{d}P_{q^{-1}(\omega)}(\omega'_1)\, \mathrm{d}P_{q^{-1}(\omega)}(\omega'_2) ~ = ~ 1 \enspace  \enspace ,
 \]
 where $\One_U$ is the indicator function for the set $U$.
 So for $P_{q^{-1}(\omega)}$-almost all $\omega'_1$ and $P_{q^{-1}(\omega)}$-almost all $\omega'_2$, we have $(\omega'_1,\omega'_2) \in U$, i.e., $Y(\omega'_1) = Y(\omega'_2)$. 
 By arguing using the decomposability property of $P_{q^{-1}(\omega)}$ discussed beneath Definition~\ref{definition:sbps},
 it follows there exists a Borel subset $C_\omega \subseteq \Omega'$ with 
 $P_{q^{-1}(\omega)}(C_\omega) = 1$
 such that  $Y$ is constant on $C_\omega$.
 By the first property of disintegrations, $P_{q^{-1}(\omega)}(q^{-1}(\omega)) = 1$.
 Defining $D_\omega := C_\omega \cap q^{-1}(\omega)$, it holds that 
 $P_{q^{-1}(\omega)}(D_\omega) = 1$, the function $q$ has constant value $\omega$ on $D_\omega$, and
 $Y$ is also constant on $D_\omega$. Let $d_\omega$ be the constant value of $Y$ on $D_\omega$. Note that we have obtained such  $d_\omega$ and $D_\omega$, for $P_\Omega$-almost-all $\omega$.
 
Next we show that there exists a measurable function $X : \Omega \to A$ such that $X(\omega) = d_\omega$, for 
$P_\Omega$-almost all $\omega$. We first show this in the special case that $A \subseteq \R$ is a closed bounded interval,
so all $A$-valued random variables are integrable with their integrals taking values in $A$. Using integrability, we define
\begin{align}
\label{eqn:X-as-d}
X(\omega) ~ & : = ~ \int  Y(\omega') \,\mathrm{d}P_{q^{-1}(\omega)}(\omega') \enspace .
\end{align}
For $P_\Omega$-almost all $\omega$, we have 
\begin{equation}
\label{eqn:Y-as-d}
\int Y(\omega') \,\mathrm{d}P_{q^{-1}(\omega)}(\omega')  ~ = ~ d_\omega \enspace ,
\end{equation}
because $Y(\omega') = d_\omega$, for $P_{q^{-1}(\omega)}$-almost all $\omega' \in D_\omega$.
So we indeed have the required measurable function $X$ in the case of a closed bounded interval $A$.
In the case of an arbitrary standard Borel space $A$, one takes some measurable embedding of $A$ into $[0,1]$
(see the discussion after Definition~\ref{definition:sbs}), and then the definition of $X$ given above can be used to obtain a measurable function  $\Omega \to [0,1]$ that lands with probability $1$ in the 
image of the embedding of $A$ in $[0,1]$, meaning that it restricts (modulo redefining it on a null set) to the required map $X : \Omega \to A$.

We next verify that  $X \circ q \EqAE Y : \Omega' \to A$. Consider the Borel set
$E := \{\omega' \in \Omega' \mid X(q(\omega')) = Y(\omega')\}$. 
We claim that, for $P_\Omega$-almost-every $\omega$, it holds that  $D_\omega \subseteq E$.
Indeed, for $P_\Omega$-almost-all $\omega$, we have that
$\omega' \in D_\omega$ implies both $q(\omega') = \omega$ and $Y(\omega') = d_\omega$,
hence $X(q(\omega')) = Y(\omega')$ follows, i.e., $\omega' \in E$. 
Because $D_\omega \subseteq E$, we have
\[
P_{q^{-1}(\omega)}( E) 
~ = ~ P_{q^{-1}(\omega)}(D_\omega) 
 % & & \text{because $P_{\Omega'}( D_\omega \Cond q = \omega) =1$}
 % \\
% & 
~ = 1 \enspace .
 \]
By the definition of disintegrations,
\begin{align*}
P_{\Omega'}(E) ~ & = ~ \int  
    P_{q^{-1}(\omega)}( E) ~ \mathrm{d}P_{\Omega}(\omega) ~ = ~ \int 1 \, \mathrm{d}P_{\Omega}(\omega)
    ~ = ~ 1 \enspace .
%    P_{\Omega'}( E \cap D_\omega  \Cond q = \omega) \, \mathrm{d}P_{\Omega}(\omega) \\
%& = ~ \int_{\Omega} \int_{\Omega'}  
%\One_E(\omega') \cdot \One_{D_\omega}(\omega')\, 
%    \mathrm{d}P_{\Omega'}( {-} \Cond q = \omega)(\omega')\, \mathrm{d}P_{\Omega}(\omega) \\
\end{align*}
So indeed $X \circ q \EqAE Y : \Omega' \to A$. That is, $[X] \cdot [q] = [Y]$. So
$[X]$ is a $[q]$-descendent of $[Y]$.

That $[X]$ is the unique $[q]$-descendent of $[Y]$ holds because $q$ is almost surjective, as in the proof of 
Proposition~\ref{proposition:sBPZ-epi}.
\end{proof}

\begin{corollary}
For any SBPS $\Omega$ the representable presheaf $\Yon\Omega$ is an atomic sheaf,.
\end{corollary}
\begin{proof}
For any SBPS $\Omega'$, we have that $(\Yon\Omega)(\Omega') \subseteq \SRV(\Omega)(\Omega')$; indeed it is the subset
of measure-preserving functions. It is then easily verified using Proposition~\ref{proposition:subsheaf} that
$\Yon\Omega$ is a subsheaf of  $\SRV(\Omega)$. In particular, $\Yon\Omega$ is a sheaf.
\end{proof}

We end this section by showing as promised that the three atomic forms of atomic  formula of our general atomic sheaf logic are, in the case that sorts are interpreted as sheaves of random variables,
correctly interpreted as the expected probabilistic relations between random variables. Firstly, that equality in the logic corresponds to almost sure equality of random variables is immediate from the definition of the sheaf $\SRV(A)$, in which random variables are explicitly identified modulo $\EqAE$. Secondly, Proposition~\ref{proposition:sbs-equivalence} below shows that atomic equivalence is interpreded as the equidistribution relation $\eqdist\,$.

\begin{proposition}
\label{proposition:sbs-equivalence}
For any SBS $A$, the atomic equivalence subsheaf 
$\Equiv_{\SRV(A)}\, \subseteq\,  \SRV(A) \!\times\! \SRV(A)$
 from Theorem~\ref{theorem:atomic-equiv}  
satisfies:
\[
\Equiv_{\SRV(A)}(\Omega)  = 
\{([X],[X']) \in (\SRV(A)\times \SRV(A))(\Omega) \mid X \eqdist X' \} \enspace .
\]
\end{proposition}
\begin{proof}
Consider any $[X],[X'] \in \SRV(A)(\Omega)$.

Suppose we have $[u],[u'] \colon \Omega' \rMap{} \Omega$ with
$[X] \cdot [u] = [X'] \cdot [u']$, i.e., $X \circ u  \EqAE X' \circ u'$. Then $(X \circ u)_*(P_{\Omega'}) = (X' \circ u')_*(P_{\Omega'})$.
Whence 
\[
X_* (P_\Omega) ~ = ~ X_*(u_*(P_{\Omega'})) ~ = ~ X'_*(u'_*(P_{\Omega'}))  ~ = ~ X'_* (P_\Omega) \enspace ,
\]
which shows $X \eqdist X'$.

Conversely, suppose $X \eqdist X'$; i.e., $X_*(P_\Omega)  = X'_*(P_\Omega)$. We write $\Omega_A$ for the
SBP space given by $A$ together with the probability measure
$P_S := X_*(P_\Omega)$. With this probability measure, the functions $X \colon \Omega \rMap{} \Omega_A$ and $X' \colon \Omega \rMap{} \Omega_A$ are morphisms in $\sBP$. By coconfluence, there exist $p,q \colon \Omega' \rMap{} \Omega$ such that $X \circ p \EqAE X' \circ q$, which implies  $[X] \cdot [p] = [X'] \cdot [q]$. So indeed $([X],[X']) \in\,  {{\Equiv_{\SRV(A)}}(\Omega)}$.
\end{proof}

The remaining form of atomic formula in our logic is conditional independence. Proposition~\ref{proposition:sbps-ci} below
shows that atomic conditional independence is indeed interpreted as the probabilistic relation of conditional independence
(Definition~\ref{definition:conditional-independence-RV}). Before this, in order to be able to make sense of the relation of atomic conditional independence, we need to verify that the sheaves $\SRV(A)$ have supports (Definition~\ref{definition:supports}).
\begin{proposition}
For any standard Borel space $A$, it holds that $\SRV(A)$ has supports.
\end{proposition}
\begin{proof}
Consider any $[X] \in \SRV(A)(\Omega)$. Define a standard Borel probability space by
\[\Omega_X ~ := ~  \text{$A$ with probability measure $P_{\Omega_X} := X_*(P_\Omega)$}\,.\]
It is easily checked that 
$(\Omega_X, \,[X], \, [x\mapsto x])$ is a support for $[X]$, using the almost surjectivity of 
$[X]: \Omega \rMap{} \Omega_X$, as in the proof of Proposition~\ref{proposition:sBPZ-epi}, for uniqueness.
\end{proof}

\begin{proposition}
\label{proposition:sbps-ci}
For any SBSs $A,B,C$, the atomic conditional independence subsheaf 
\[{\Indep}_{\SRV(A),\SRV(B)|\SRV(C) }\, \subseteq\,  \SRV(A) \!\times\! \SRV(B)  \!\times\! \SRV(C )\]
from  Theorem~\ref{theorem:general-conditional-independence}  satisfies:
\begin{align*}
& {\Indep}_{\SRV(A),\SRV(B)|\SRV(C) }(\Omega) ~  = ~ 
%\\ & \quad 
\{([X],[Y],[Z]) \in (\SRV(A)\! \times \! \SRV(B) \! \times \! \SRV(C)) (\Omega)\mid X \Indep Y \Cond Z \} \enspace .
\end{align*}
\end{proposition}
\begin{proof}
Suppose $[X] \in \SRV(A)(\Omega),\, [Y] \in \SRV(B)(\Omega)$ and $[Z] \in \SRV(C)(\Omega)$ are
such that $X \Indep Y \Cond Z$ according to
Definition~\ref{definition:conditional-independence-RV}. Define
\begin{align*}
\Omega_X ~ := ~  & \text{$A \times C$ with probability measure $P_{\Omega_X} := (X,Z)_*(P_\Omega)$}
\\
\Omega_Y ~ := ~  & \text{$B \times C$ with probability measure $P_{\Omega_Y} := (Y,Z)_*(P_\Omega)$}
\\
\Omega_Z ~ := ~ & \text{$C$ with probability measure $P_{\Omega_Z} := Z_*(P_\Omega)$}\,.
\end{align*}
Then the hybrid diagram below, shows that the triple $([X],[Y],[Z])$ belongs to the atomic conditional independence 
relation ${\Indep}_{\SRV(A),\SRV(B)|\SRV(C) }(\Omega)$.
\[
\vcenter{\hbox{\includegraphics[scale=0.18]{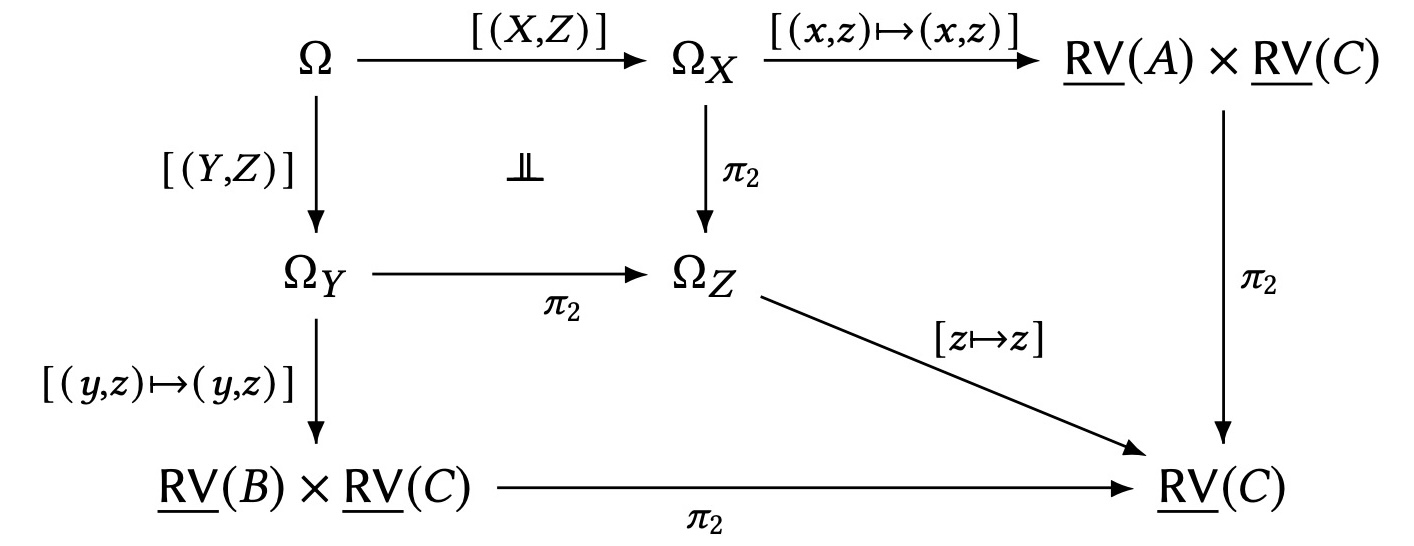}}}
%\begin{diagram}
%\Omega & \rTo^{\!\!\!\!\![(X,Z)]~~~~~} & \Omega_X  & \rTo^{\!\!\![(x,z) \mapsto (x,z)]~~~~~} & \SRV(A) \times \SRV(C)\\
%\dTo^{[(Y,Z)]} & \!\!\!\!\!\!\!\!\!\Indep &  \dTo_{\pi_2} &  &  \\
%\Omega_Y & \rTo_{\pi_2} & \Omega_Z & & \dTo_{\pi_2} \\
%\dTo^{[(y,z) \mapsto(y,z)]} & & & \rdTo^{[z\mapsto z]} &  \\
%\SRV(B) \times \SRV(C) & & \rTo_{\pi_2} & & \SRV(C)
%\end{diagram}
\]
In this diagram, $(\Omega_Z, \,[Z], \, [z\mapsto z])$ is support for $[Z]$, by the definition of $\Omega_Z$, and the top-left square is independent, because $(X,Z) \Indep (Y,Z)  \Cond Z$ holds, which follows from $X \Indep Y \Cond Z$.

Conversely, suppose $([X],[Y],[Z]) \in {\Indep}_{\SRV(A),\SRV(B)|\SRV(C) }(\Omega)$. Defining $\Omega_X, \Omega_Y$ and $\Omega_Z$ as above, we have that
$(\Omega_X, \, [(X,Z)] , \, [(x,z) \mapsto (x,z)])$ is support for $[(X,Z)]$ and
$(\Omega_Y, \, [(Y,Z)] , \, [(y,z) \mapsto (y,z)])$ is support for $[(Y,Z)]$ 
 $(\Omega_Z, \, [Z], \, [z \mapsto z])$ is support for $[Z]$. So, by Lemma~\ref{lemma:use-supports}, these supports fit into the hybrid diagram above. Since the top-left square is independent, we have
 $(X,Z) \Indep (Y,Z)  \Cond Z$. From this, $X \Indep Y \Cond Z$ follows, as required.
%That is, we have the data in the hybrid diagram
%below, where $[X' ] \cdot [p] = [X]$ and $[Y'] \cdot [q] = [Y]$ and $(\Omega_Z, \, [r \circ p], \, [Z'])$ is support for $Z$.
%\[
%%\vcenter{\hbox{\includegraphics[scale=0.15]{diagM.jpg}}}
%\begin{diagram}
%\Omega & \rTo^{\!\!\![p]~~~} & \Omega_X  & \rTo^{\!\!\!\!\!{[(X',U')]}\;\;} & \SRV(A) \times \SRV(C) \\
%\dTo^{[q]} & \!\!\!\!\!\!\!\!\!\!\Indep &  \dTo_{[r]} &  &  \\
%\Omega_Y & \rTo_{[s]} & \Omega_Z & & \dTo_{\pi_2} \\
%\dTo^{(Y',V')} & & & \rdTo_{[Z']} &  \\
%\SRV(B) \times \SRV(C) & & \rTo_{\pi_2} & & \SRV(C)
%\end{diagram}
%\]
%We show that $X \Indep Y \Cond Z$, according to
%Definition~\ref{definition:conditional-independence-NV}.
%Suppose we have $\omega', \omega'' \in \Omega$ such that $X(\omega') = a$ and $Z(\omega')= c$ and
%$Y(\omega'') = b$ and $Z(\omega'') = c$. Then
%\[
%Z'(r(p(\omega'))) = Z(\omega') = Z(\omega'') = Z'(s(q(\omega''))) \enspace .
%\]
%Since $(\Omega_Z, \, r \circ p, \, Z')$ is support for $Z$, the function $Z' \colon \Omega_Z \to C$ is injective,
%hence $r(p(\omega')) = s(q(\omega''))$. Since the top-left square is independent, there exists $\omega \in \Omega$
%such that $p(\omega) = p(\omega')$ and $q(\omega) = q(\omega'')$. Then
%$X(\omega) = X'(p(\omega)) = X'(p(\omega')) = X(\omega') = a$. Similarly, $Y(\omega) = b$, and $Z(\omega) = c$.
\end{proof}

%%%%%%%%%%%%%%%%%%%%%%%%%%%%%%%%%%%
%%% SECTION
%%%%%%%%%%%%%%%%%%%%%%%%%%%%%%%%%%%%

\section{The Schanuel topos}
\label{section:Schanuel}

We give a very condensed outline, without proofs, of one more example in which we have an atomic sheaf logic of equivalence and conditional independence: the Schanuel topos, which is equivalent to the category of nominal sets
of Gabbay and Pitts~\cite{GabP,pitts}.

Let $\I$ be (a small version of) the category whose objects are finite sets and whose morphisms are injective functions.
We consider the topos of atomic sheaves over the category $\Iop$. Since all maps in $\I$ are obviously monomorphic, all maps in $\Iop$ are epimorphic.
\begin{proposition}
\label{theorem:Iop-indep-pullbacks}
The category $\Iop$ carries independent pullback structure satisfying the descent property. %, all its maps are epimorphic 
and it has pairings.
\end{proposition}
\begin{proof}[Description of structure] 
Define a commuting square in $\Iop$ to be \emph{independent} if the associated square (with opposite orientation) 
of functions in $\I$ is a pullback
in $\I$ (or equivalently in $\Set$). A commuting square in $\Iop$ is then an independent pullback if and only if the associated square of functions in $\I$ is a pushout in $\Set$ (but not necessarily in $\I$). Every cospan in $\Iop$ completes to an independent pullback by taking the pushout in $\Set$ of the associated 
span of functions in $\I$. 
\end{proof}

\begin{proposition}
\label{theorem:Iop-pairing}
The category $\Iop$ has pairings.
\end{proposition}
\begin{proof}[Description of structure] 
A span in $\Iop$ gives rise to a cospan of functions in $\I$. The pairing in $\Iop$ is given by the pushout in $\Set$ of the pullback in $\I$ (or $\Set$) of this cospan of functions.
\end{proof}

A presheaf $P \in \Psh(\Iop)$ is just a covariant functor $P \colon \I \to \Set$. 
The description of independent squares above, means that Theorem~\ref{theorem:sheaf-nice},
in the case of $\sCat{C} = \Iop$, specialises to the well-known characterisation that a presheaf $P \in \Psh(\Iop)$ is an atomic sheaf 
if and only if the  covariant functor $P \colon \I \to \Set$ preserves pullbacks (see, e.g.,  \cite[A 2.1.11(h)]{johnstone}). This property enables the result below to be established 
by constructing supports in $\Iop$ as a multiple pullbacks in $\I$ over all representable factorisations, of 
which there are only finitely many.

\begin{proposition}
\label{proposition:schanuel-supports}
Every atomic sheaf in $\ShAt(\Iop)$ has supports.
\end{proposition}

\noindent For a sheaf $\Sh{A}$ in $\ShAt(\Iop)$, the support of an element $x \in \Sh{A}(X)$ corresponds to a
 smallest subset $\mathsf{supp}(x) \subseteq X$ for which there exists $y \in \Sh{A}(\mathsf{supp}(x))$ such that 
$x = y \cdot i$, where $i: X \rMap{} \mathsf{supp}(x)$ in $\Iop$ is given by the inclusion function $\mathsf{supp}(x) \to X$.
Proposition~\ref{proposition:schanuel-supports} is well known.
For example, it plays a key role in Fiore's presentation of $\ShAt(\Iop)$ as a Kleisli category~\cite{menni,FM}. 
An analogous property is also
prominent in presentations of 
the equivalent category of nominal sets~\cite{GabP,pitts}.

\begin{proposition}
\label{proposition:schanuel-equivalence}
For any  $\Sh{A}$ in $\ShAt(\Iop)$, the atomic equivalence subsheaf 
$\Equiv_{\Sh{A}}\, \subseteq\,  \Sh{A} \!\times\! \Sh{A}$
 from Theorem~\ref{theorem:atomic-equiv}  
satisfies:
\[
\Equiv_{\Sh{A}}(X)  = 
\{(x,y) \in (\Sh{A} \times \Sh{A})(X) \mid  \exists X \rMap{i} X .\;\;\; y = x \cdot i\}  \enspace .
\]
\end{proposition}

\begin{proposition}
\label{proposition:schanuel-conditional-independence}
For any  $\Sh{A},\Sh{B},\Sh{C}$ in $\ShAt(\Iop)$, the atomic conditional independence subsheaf 
\[{\Indep}_{\Sh{A},\Sh{B}|\Sh{C}}\, \subseteq\,  \Sh{A} \!\times\! \Sh{B}  \!\times\! \Sh{C}\]
from Theorem~\ref{theorem:general-conditional-independence}  satisfies:
\begin{align*}
& {\Indep}_{\Sh{A},\Sh{B}|\Sh{C}}(X) ~  = ~ 
%\\ & \quad 
\{(x,y,z) \in (\Sh{A} \times \Sh{B}\times \Sh{C}) (X)\mid \mathsf{supp}(x) \cap \mathsf{supp}(y) \subseteq \mathsf{supp}(z)\ \} \enspace .
\end{align*}
\end{proposition}

%%%%%%%%%%%%%%%%%%%%%%%%%%%%%%%%%%%
%%% SECTION
%%%%%%%%%%%%%%%%%%%%%%%%%%%%%%%%%%%%

\section{Discussion and related work}
\label{section:related}

\subsection{Relationship with (multi)team semantics}
\label{subsection:relationship-multiteam}

Our main running example throughout the paper has been the category of atomic sheaves over the category $\Sur$, in which the 
interpretations of atomic equivalence and conditional independence, when applied to the
sheaves $\NEls(A)$ of nondeterministic variables, coincide with the multiteam interpretations of those relations 
from the \emph{(in)dependence logics} of~\cite{Vaananen,GV,DHKM}.
For our logic, we use  the canonical 
internal logic of an atomic sheaf topos, whose semantics is provided by the forcing relation
of Figure~\ref{figure:atomic-sheaf-semantics}, and whose underlying logic is ordinary classical logic.

In our route to atomic sheaf logic 
in Sections~\ref{section:multiteams}--\ref{section:atomic-sheaf-logic},
the use of multiteams 
 seems essential. Indeed, it is the presentation of multiteams as finite-fibred functions in Section~\ref{section:multiteams}
that forms the basis for the connection with the category $\Sur$, whence with atomic sheaves. This contrasts with the 
majority of work on (in)dependence logic, from \cite{Vaananen,GV} onwards, which is largely  based on teams
rather than on multiteams. It is accordingly worth observing, that it is possible to reformulate the atomic sheaf logic of 
Figure~\ref{figure:atomic-sheaf-semantics} directly in terms of teams. 
To see this, note that any finite team trivially gives rise to a canonical finite multiteam, in which every assignment  has multiplicity $1$. 
Conversely, the support of any finite multiteam is a team. Under the correspondence between $\mathcal{V}$-multiteams,
and $\mathcal{V}$-assignments of nondeterministic variables, discussed in Section~\ref{section:multiteams}, we can reformulate these two statements in the following way. Every finite $\mathcal{V}$-team gives rise to 
$\Sh{\rho} \colon \mathcal{V} \to ( \Omega  \to A)$ enjoying the \emph{team property}: 
for all $\omega,\omega' \in \Omega$, if 
$\Sh{\rho} (\Syn{x})(\omega) = \Sh{\rho}_S(\Syn{x})(\omega')$, for all $ \Syn{x} \in \mathcal{V}$,  
then $\omega = \omega'$. Moreover, for every $\mathcal{V}$-multiteam 
$\Sh{\rho'} \colon \mathcal{V} \to ( \Omega'  \to A)$ there exists a unique up to isomorphism $q : \Omega' \to \Omega$ and
$\Sh{\rho} \colon \mathcal{V} \to ( \Omega  \to A)$ such that $\Sh{\rho}$ satisfies the team property and
%$\Sh{\rho'}(\Syn{x})(\omega') = \Sh{\rho}(\Syn{x})(q(\omega))$, for all  $ \Syn{x} \in \mathcal{V}$  and $\omega' \in \Omega'$.
$\Sh{\rho'} = \Sh{\rho} \cdot q$. It thus follows from the sheaf property of forcing (Proposition~\ref{proposition:forcing-sheaf}) that
the behaviour of the relation $\Omega \Forces_{\Sh{\rho}} \Phi$ in $\ShAt(\Sur)$, for any formula $\Phi$, is determined
entirely by its behaviour on teams $\Sh{\rho}$. Moreover, it is easy to unwind the clauses in 
Figure~\ref{figure:atomic-sheaf-semantics} and to reformulate them directly in terms of ordinary teams qua sets of assignments.
Thus atomic sheaf logic over $\Sur$ could equivalently be presented in terms of teams rather than multiteams.

If one carries out such a reformulation in the case of conjunction and of the existential quantifier, one obtains the
standard team interpretation of the former~\cite{Vaananen}, and the \emph{lax} interpretation of the latter, which is 
often the preferred team interpretation~\cite{HK,GV}. The clauses for the other connectives and for the universal quantifier are different however. Whereas the clauses in Figure~\ref{figure:atomic-sheaf-semantics} validate the laws of classical logic, it is well known
that the standard team semantics of the other connectives and the universal quantifier leads to
some logically exotic behaviour. For example, disjunction is not an idempotent operation. 
Abramsky and V\" a\" an\" anen~\cite{AV} provide an illuminating explanation for such behaviour, by showing that 
the dependence logic connectives and quantifiers can be naturally understood as fitting into the framework of 
Pym and O'Hearn's  \emph{logic of bunched implications (BI)}~\cite{OHP,POHY}. We now review this perspective and then discuss how it might be adapted to atomic sheaf logic. 

The approach of~\cite{AV}  is based on Lawvere's notion of  \emph{hyperdoctrine}~\cite{lawvere,pittsCL}.
Recall that the contravariant poswerset functor $\mathsf{P}$ on sets, can be viewed as a functor $\mathsf{P} \colon \Op{\Set} \to \mathbf{Pos}$, where $\mathbf{Pos}$ is the category of partially ordered sets and monotone functions. Specifically,  $\mathsf{P}$  maps any set $X$ to its
set of subsets  partially  ordered by subset inclusion. The functor
$\mathsf{P} \colon \Op{\Set} \to \mathbf{Pos}$ is then a hyperdoctrine. 
  Propositional logic for propositions over a set $X$ is modelled by the boolean algebra structure on $\mathsf{P}(X)$. 
For any function $f \colon  X \to Y$, the \emph{reindexing function}
$\mathsf{P}(f) := f^{-1} \colon  \mathsf{P}(Y) \to  \mathsf{P}(X)$ preserves the boolean algebra structure.
The quantifiers $\exists
: \mathsf{P}(X \times Y) \to \mathsf{P}(X)$ and $\forall
: \mathsf{P}(X \times Y) \to \mathsf{P}(X)$, quantifying over a set $Y$, are modelled as left and right adjoints respectively to the monotone function (considered qua functor) $\pi_1^{-1} : \mathsf{P}(X) \to  \mathsf{P}(X \times Y)$, where $\pi_1 \colon X \times Y \to X$ is the projection map. 

The main construction in~\cite{AV}, adapts the above hyperdoctrine for classical logic
to team semantics, by composing $\mathsf{P} \colon \Op{\Set} \to \mathbf{Pos}$
with the functor $\mathcal{L}\colon \mathbf{Pos} \to \mathbf{Pos}$ given by the 
operation $\mathcal{L}$ that maps any partial order $B$ to its 
lattice $\mathcal{L}(B)$ of down-closed sets. The composite functor $\mathcal{L} \mathsf{P} \colon \Op{\Set} \to \mathbf{Pos}$
then has the following properties. For every set $X$, the fibre poset $\mathcal{L} \mathsf{P}(X)$ is, in a canonical way, a \emph{BI algebra}, that is an algebraic model of the \emph{logic of bunched implications} BI~\cite{OHP,POHY}. In the case 
$X = A^\mathcal{V}$, the elements of $\mathcal{L} \mathsf{P}(A^\mathcal{V})$ are precisely down-closed (in the subset ordering) sets of 
$A$-valued teams with variable set $\mathcal{V}$. Each connective of BI is modelled algebraically as a function of appropriate arity 
on $\mathcal{L} \mathsf{P}(A^\mathcal{V})$. For example, the  \emph{multiplicative conjunction} $\otimes$, is modelled as a certain canonically generated function
$\otimes : \mathcal{L} \mathsf{P}(A^\mathcal{V}) \times \mathcal{L} \mathsf{P}(A^\mathcal{V}) \to \mathcal{L} \mathsf{P}(A^\mathcal{V})$.
Writing $\Phi$ and $\Psi$ for elements of $\mathcal{L} \mathsf{P}(A^\mathcal{V})$ (which can be thought of as an abstract set of propositions), and writing $S \Forces \Phi$ to mean $S \in \Phi$, the function $\otimes$ can be characterised  by
\[
S \Forces \Phi \otimes \Psi ~ \Iff ~ \exists T, U, ~ S = T \cup U ~ \text{and}~ T \Forces \Phi ~ \text{and} ~ U \Forces \Psi \, .
\]
This is exactly the semantic clause for the \emph{disjunction} connective of team semantics.
The exotic behaviour of 
the disjunction of dependence logic is thus nicely explained 
as a manifestation of the expected behaviour of the multiplicative conjunction of BI, whose multiplicative connectives have a natural resource-sensitive interpretation. 
A further consequence of the hyperdoctrine construction in~\cite{AV} is that the embedding of dependence logic in BI  enriches the former with additional logical connectives, such as both additive (intuitionistic) and multiplicative implications. 
Lastly, the hyperdoctrine formulation of dependence logic provides an elegant explanation for the 
team semantics interpretation of the quantifiers $\exists$ and $\forall \colon  \mathcal{L} \mathsf{P}(A^{\mathcal{V}\uplus \{x\}}) \to 
\mathcal{L}\mathsf{P}(A^{\mathcal{V}})$, which are characterised in the desired way~\cite{lawvere,pittsCL} as respectively left and right adjoints  to
$\mathcal{L} \mathsf{P}(\rho \mapsto \rho|_{\mathcal{V}}) : \mathcal{L}\mathsf{P}(A^{\mathcal{V}}) \to 
\mathcal{L} \mathsf{P}(A^{\mathcal{V}\uplus \{x\}})$.

The above hyperdoctrine construction from~\cite{AV} works for the original dependence logic~\cite{Vaananen}, but not for 
independence logic~\cite{GV}, because teams satisfying independence atoms are not down-closed in the subset order. This means that the $\mathcal{L}$ functor cannot be used to interpret formulas involving independence. An alternative is to combine the contravariant powerset functor 
$\mathsf{P}$ with the covariant powerset functor $\mathsf{P}_{!}$ (with direct image as its functorial action). It turns out that if one considers the  composition in the order 
$\mathsf{P} \mathsf{P}_{!} \colon \Op{\Set} \to \mathbf{Pos}$, then the left and right adjoints to the monotone function
$\mathsf{P} \mathsf{P}_{!}(\rho \mapsto \rho|_{\mathcal{V}}) : \mathsf{P} \mathsf{P}_{!}(A^{\mathcal{V}}) \to 
\mathsf{P} \mathsf{P}_{!}(A^{\mathcal{V}\uplus \{x\}})$ correspond respectively to the existential and universal quantifier with 
(the team version of) the forcing clauses from Figure~\ref{figure:atomic-sheaf-semantics}. Further, the boolean algebra structure on 
$\mathsf{P} \mathsf{P}_{!}(A^{\mathcal{V}})$ corresponds to (the team version of) the forcing clauses 
for the propositional connectives in Figure~\ref{figure:atomic-sheaf-semantics}, and this structure is  preserved by
all \emph{reindexing maps} $\mathsf{P} \mathsf{P}_{!}(f)$. The hyperdoctrine $\mathsf{P} \mathsf{P}_{!} \colon \Op{\Set} \to \mathbf{Pos}$ thus recovers the team version of atomic sheaf logic
as in~Figure~\ref{figure:atomic-sheaf-semantics}. It would be interesting to investigate  this construction in more detail, for example to explore how independence and equivalence formulas interact with the hyperdoctrine formulation, and also  the extent to which the logic BI logic is relevant in this picture.  Both points are potentially subtle. The standard hyperdoctrine desideratum 
that logical structure should be 
preserved by reindexing maps  provides a constraint on which atomic primitives are admissible. Moreover, the 
relevance of BI logic is less \emph{a priori} apparent than in~\cite{AV}, because the switch in the order of composition ($\mathsf{P} \mathsf{P}_{!}$ has the covariant functor as the inner functor, whereas $\mathcal{L} \mathsf{P}$ has its covariant functor as the outer functor) means that the outermost functor is no longer given by a canonical BI-algebra construction. 

A different source of exotic behaviour in (in)dependence  logics concerns
interaction between the universal quantifier and (in)dependence atoms. 
One particularly striking example is provided by the sentence below.
\begin{equation}
\label{eqn:forall-example}
 \forall \Syn{x^A},  \forall\Syn{y^B}. \;\; (\Syn{x^A} \sIndep \Syn{y^B}) 
\end{equation}
According to the usual team semantics 
of the universal quantifier, the above sentence is valid. Nevertheless, one can easily exhibit example teams $S$ for which it is not the case that 
$S \Forces \Syn{x^A} \sIndep \Syn{y^B}$, and rightly so, because there would be little point in independence logic if independence were a universally valid relation. 
We view the validity of \eqref{eqn:forall-example} (and other examples like it) as showing that if one is to use (in)dependence logic
 as a basis for reasoning about (in)dependence properties then
 the associated rules of inference will have to be unusual.
 
 Nevertheless,  independence logics and their team semantics have been successfully applied in the direction of
reasoning about dependence and conditional independence. For example, Hannula and Kontinen axiomatise the valid implications involving
\emph{inclusion}  and \emph{embedded multivalued dependencies} in database theory in terms of
inclusion and
conditional independence formulas with their team semantics~\cite{HK}. 
An interesting observation about this work is that it takes place in the fragment 
of independence logic comprising conjunction and (lax) existential quantification as the only logical operators.
Since these are exactly the logical operators for which the semantic interpretations in independence logic and
atomic sheaf logic coincide, the same development can be imported verbatim into atomic sheaf logic
in $\ShAt(\Sur)$ extended with the inclusion relation (which indeed defines a subsheaf
of $\NEls(A) \times \NEls(A)$). One advantage of such a reformulation is that
the axiomatised rules of inference in~\cite{HK}
can be expressed as individual formulas, using the general implication connective of
atomic sheaf logic, rather than left as entailments.
For example, the rule of \emph{inclusion introduction}, which concerns the inclusion relation, has an obvious (derivable) analogue for the equivalence (equiextension) relation, namely: if one has already derived an equivalence
formula
$\vec{\Syn{x}} \Equiv \vec{\Syn{x'}}$
then one can infer the formula
$
\exists \Syn{y'^{\,A}}. ~  (\vec{\Syn{x}},\Syn{y^{\,A}} \Equiv \vec{\Syn{x'}},\Syn{y'^{\,A}})
$.
In atomic sheaf semantics, this rule can be formulated as an implication. Indeed, it is none other than  the 
\emph{transfer principle}~\eqref{equiv:G} from Figure~\ref{figure:equivalence-axioms}, valid in any 
atomic sheaf topos.
The same transfer principle can also be found in mainstream probability theory. The interpretation of~\eqref{equiv:G} in the category $\ShAt(\sBP_0)$
 of probability sheaves   is very close to the \emph{transfer theorem} of \cite[Theorem 5.10]{kallenberg}, and arguably captures the essence of that theorem in logical form. 
 
 The interpretation of atomic sheaf logic in $\ShAt(\sBP_0)$ also connects with a body of work on adapting team semantics to probability-based scenarios. For example, an $A$-valued  \emph{measure team} in~\cite{HPV} is a measurable map $\Omega \to (\mathcal{V} \to A)$, for some  probability space $\Omega$ and set of variables $\mathcal{V}$. This can  
 equivalently be presented as a map $\mathcal{V} \to (\Omega \to A)$, which  is almost the same thing as 
 a variable assignment in  atomic sheaf logic over $\sBP_0$, i.e., a mapping from variables to elements of $\SRV(A)(\Omega)$.
 There are however two key differences: 
 random variables in  $\SRV(A)$ are identified up to almost sure equality, and objects in $\sBP_0$
 are restricted to probability spaces $\Omega$ that are standard Borel.
 Although these differences may seem minor, they are crucial to the interpretation of atomic sheaf logic in $\ShAt(\sBP_0)$. For example, it is because of the restriction to standard Borel spaces that the category  $\ShAt(\sBP_0)$ is coconfluent.  The failure of coconfluence for general probability spaces makes it difficult to
 extend the measure-team semantics of atomic formulas in~\cite{HPV} to include the logical connectives and atoms of independence logic. In the literature, such extensions have   been given only for probabilistic teams based on discrete probability~\cite{DHKMV}.  It is worth remarking that discrete probability fits in equally well with the approach of the present paper. One can consider atomic sheaves over the category of finite probability spaces, or alternatively over the category of countable probability spaces, both of which are full subcategories of $\sBP_0$. Such examples further substantiate our thesis  that atomic sheaf categories provide a unifying framework  configurable to diverse settings for conditional independence.  It would be interesting to compare our approach with the semiring-based framework
 of~\cite{BHKPV}, which provides a different unifying approach to varieties of team semantics, which encompasses both ordinary teams and discrete probabilistic teams.

\subsection{Computer science applications}
\label{subsection:cs-applications}

In this section we outline possible computer science applications for atomic sheaf logic. Rather than trying to be comprehensive, we instead focus on a few illustrative examples, beginning with reasoning about probabilistic programs. 

An almost surely terminating imperative probabilistic program $C$ can be modelled as a probabilistic map between states, that is  a function $\Sem{C}_S : \mathsf{State} \to \mathcal{D}(\mathsf{State})$, where $\mathcal{D}(\mathsf{State})$ is the set of probability distributions over states. Alternatively, but equivalently, it can be viewed as a transformation
$\Sem{C}_T: \mathcal{D}(\mathsf{State}) \to \mathcal{D}(\mathsf{State})$ mapping a probability distribution on initial states to the induced probability distribution on final states~\cite{Kozen}. 
There is also a third related possibility. One can view the program as a transformation $\Sem{C}_R$ mapping an initial \emph{random state} $\Sigma : \Omega \to \mathsf{State}$, for some sample space $\Omega$, to a final random state $\mathrm{T}$~\cite{JS}. However,
because the program $C$ may make use of randomness not present in
$\Omega$, the sample space for $\mathrm{T}$ has to be, in general,  an \emph{extension} of $\Omega$, meaning that 
 $\mathrm{T} :  \Omega' \to \mathsf{State}$  for some suitable sample space $\Omega'$ equipped with a probability preserving map
$q: \Omega' \to \Omega$. While the idea of modelling programs as random-state transformers is very natural, some
 careful bookkeeping is required to deal with the change of sample space. Such bookkeeping can be avoided entirely if one uses the alternative approach of defining the random-state-transformer semantics in the atomic sheaf logic 
 of $\ShAt(\sBP_0)$. Under this approach $\Sem{C}_R$ is formulated as a relation
 $\Sem{C}_R \subseteq \SRV(\mathsf{State}) \times \SRV(\mathsf{State})$ satisfying:
 for any random initial state $\Sigma$ on which $C$ terminates, there exists a random final state
 $\mathrm{T}$ such that $\Sigma\Sem{C}_R\mathrm{T}$ and, for any
 random state $\mathrm{T}'$, it holds that $\Sigma\Sem{C}_R\mathrm{T}'$  implies $\Sigma,\mathrm{T} \Equiv \Sigma,\mathrm{T}'$.
  The key point here is that no sample spaces need to be specified, because, from the viewpoint of atomic sheaf logic, sample spaces are implicit, and the extension of sample spaces is  likewise taken care of implicitly by the semantics of the existential quantifier. 
  Not only is such an implicit-sample-space  style of manipulating random variables  intuitive, it also avoids the bookkeeping required when dealing with explicit sample space extensions. For example, in~\cite{JS}, a property called \emph{relative tightness} is identified as  useful property of probabilistic Hoare-triple-like specifications. Such a  specification $\{\Phi\} C \{\Psi\}$ asserts that, if the precondition $\Phi$ holds for a random initial state $\Sigma$, and if $C$ terminates from $\Sigma$, then the postcondition $\Psi$ holds for the induced random final state 
  $\mathrm{T}$. The property of relative tightness asserts that the probabilistic behaviour of the random state $\mathrm{T}$ on 
  the variables $\mathsf{FV}(\Psi)$ relevant to determining the truth of $\Psi$, depends only on the value of the initial state $\Sigma$ on $\mathsf{FV}(\Phi)$. This can be formulated in a simple way as the statement about conditional independence on the left below
\[
\mathrm{T}_{\mathsf{FV}(\Psi)} \sIndep \Sigma  \, \Cond \, \Sigma_{\mathsf{FV}(\Phi)}
\qquad \qquad
\mathrm{T}_{\mathsf{FV}(\Psi)} \sIndep \Sigma \circ q \, \Cond \, \Sigma_{\mathsf{FV}(\Phi)} \circ q \, ,
\]
where $\Sigma_{\mathsf{FV}(\Psi)}$ and $\mathrm{T}_{\mathsf{FV}(\Psi)}$ denote the initial and final random states restricted to the specified variable
sets. For contrast, we include on the right above the statement of relative tightness that appears in~\cite{JS}, which shows the need for bookkeeping (in this case, composition with $q$)
%, which is required in order to for all random variables to be over the same sample space) 
when the standard mathematical formulation of random variables with explicit sample spaces is used. For a  more involved example of the efficiency afforded by the implicit-sample-space approach of atomic sheaf logic, we consider how the while statement on the left below is 
approximated by iterating the conditional statement on the right.
\[
\mathsf{while}~B~\mathsf{do}~C \qquad \qquad \mathsf{if}~B~\mathsf{then}~C~\mathsf{else}~\mathsf{skip}
\]
Working within atomic sheaf logic, suppose the while statement terminates in random final state $\mathrm{T}$ from a random initial state $\Sigma$. Then defining $\Sigma_0 = \Sigma$ and letting $\Sigma_{n+1}$ be such that $\Sigma_n\,\Sem{\mathsf{if}~B~\mathsf{then}~C~\mathsf{else}~\mathsf{skip}}_R\,\Sigma_{n+1}$, we obtain a sequence 
$(\Sigma_n)_{n \geq 0}$ of random states that converges almost surely to the random state $T$. 
The resulting convergence property  $\Sigma_n \to \mathrm{T}$ is used in ongoing work extending~\cite{JS} 
to prove the correctness of a partial correctness rule for while loops in a probabilistic program logic.
The formulation of the same convergence statement with explicit sample states is unwieldy as it involves a sequence 
$\Omega_0 \xleftarrow{q_0} \Omega_1 \xleftarrow{q_1} \Omega_2 \xleftarrow{q_2} \cdots$ of sample space extensions for the random states $(\Sigma_n)_n$,
as well as a cone (in the category-theoretic sense) $(\Omega_n \xleftarrow{r_n} \Omega')_n$ for the sequence, where $\Omega'$ is the
sample space for $\mathrm{T}$. With this scaffolding in place, the convergence property can be stated as $\Sigma_n \circ r_n  \to \mathrm{T}$.

% and the composition with $q$ is needed to place all random variables over the same sample space $\Omega'$.\footnote{The property of relative tightness can also be formulated directly in terms $\Sem{C}_S$  or $\Sem{C}_T$, but only as a property of the semantic functions themselves, not as a property applicable to a given pair of a distribution on initial states with the corresponding distribution on final states.}

We have outlined above how atomic sheaf logic might be applied to formulate a
random-state-based operational  semantics for imperative  probabilistic programs.
Another potential application is to the assertion logics of Hoare-like program logics for probabilistic programs, in particular to 
\emph{probabilistic separation logic (PSL)}. 
PSL was first introduced in~\cite{BHL} as an approach to verifying probabilistic programs using  a version of the \emph{separating conjunction} of separation logic~\cite{OHRY,YOH} to reason about
probabilistic  independence. The modular style  of reasoning is supported by a version of the \emph{frame rule} of separation logic, which, in the case of probabilistic separation logic, allows  certain statements about probabilistic independence to be inferred.
The paper~\cite{BHL}  presents several applications to the verification of cryptographic protocols. Subsequent work has extended the approach to reason about {negative dependencies}~\cite{BGHT}, adapted it to a probabilistic functional language~\cite{LAS} and
incorporated conditional independence~\cite{BDHS,LAS}.
In all the aforementioned works, the assertion logic has been given as an instance of the \emph{logic of bunched implications (BI)} 
 with a Kripke-style semantics defined over a partially ordered \emph{resource monoid}~\cite{POHY}. This leads to an intuitionistic but not classical assertion logic. 
 It seems likely  that one can obtain a classical assertion logic, by replacing the  Kripke-style semantics of BI
 in a partially ordered resource monoid with a category-based semantics utilising the forcing clauses of 
atomic sheaf logic.\footnote{One version of such a classical assertion logic appears in~\cite{JS}. 
However, the very simple setting of abstract \emph{semantic assertions} with no explicit quantifiers in \emph{op.\ cit.}, enables the  category-theoretic genesis of the logic to be hidden. Its one remaining trace is the set of \emph{footprint variables}, which corresponds to the notion of support in the present paper.}

 Another connection with the logic of bunched implications comes from a fact that we have not developed in the present paper:
 every category $\sCat{C}$ with independent pullbacks and terminal object is symmetric monoidal, and its category
 $\ShAt(\sCat{C})$ of atomic sheaves carries, in addition to its cartesian closed structure,
 a second symmetric monoidal closed structure $\otimes_{\mathsf{Sh}}$ derived, using the methods of~\cite{day}, from the modoidal structure  of $\sCat{C}$. Categories with two such closed structures are category-theoretic models of BI~\cite{OHP}. In the case of $\ShAt(\sCat{C})$, the monoidal structure is furthermore \emph{affine}, hence it has projections $\Sh{A} \xleftarrow{} \Sh{A} \otimes_{\mathsf{Sh}} \Sh{B}
 \xrightarrow{} \Sh{B}$. In the case that $\Sh{A}$ and  $\Sh{B}$ have supports, then the projections are jointly monic and the resulting
 monomorphism
 \[ 
 \vcenter{\hbox{\includegraphics[scale=0.18]{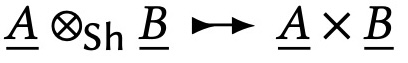}}}
% \begin{diagram}
% \Sh{A} \otimes_{\mathsf{Sh}} \Sh{B} & \rMono &  \Sh{A} \times \Sh{B}
% \end{diagram}
 \]
 is in fact isomorphic to ${\Indep}_{\Sh{A},\Sh{B}} \subseteq \Sh{A} \times \Sh{B}$ given by the unconditional 
 version of~\eqref{equation:independence-subsheaf} (i.e.,  in which $C$ is a terminal object). That is, the unconditional independence relation of the present paper is recovered as an instance of  monoidal structure. 
 This connection will be elaborated in a follow-up paper, where also
 the relationship with the monoidal category setting of~\cite{simpsonA} will be discussed. Indeed the 
 notion of 
\emph{local independence structure with local independent products} in \emph{op.\ cit.}\ is 
equivalent to the independent pullback structure of Section~\ref{section:independent-pullbacks}, but with a much more involved axiomatisation in terms of fibred monoidal structure. 
The monoidal structure of $\sCat{C}$ provides another connection between the work of the present paper and 
varieties of separation logic including probabilistic separation logic,
as elaborated by Li \emph{et.\ al.}~\cite{LAST}.
In their work, the Day monoidal product on presheaves~\cite{day} is used to model the separation of state into independent segments,  whose probabilistic independence can be superimposed using a probability monad. As in our work, the notion of sheaf with supports, which was introduced independently in~\cite{LAST},  plays a crucial role.

The category $\mathbf{Nom}$ of nominal sets of Gabbay and Pitts~\cite{GabP,pitts} has found applications to reasoning about \emph{names} in computer science. The monograph~\cite{pitts} presents many 
examples of such applications, together with pointers to the literature.
One prominent application area is reasoning about abstract syntax for languages involving 
operators that bind variables. 
 
As mentioned in Section~\ref{section:Schanuel}, the category $\mathbf{Nom}$ is equivalent to the Schanuel topos, and so the relations of equivalence and conditional independence defined in Section~\ref{section:Schanuel} can 
be transferred to $\mathbf{Nom}$. 
In $\mathbf{Nom}$, the atomic equivalence relation of Proposition~\ref{proposition:schanuel-equivalence}
is  the equivalence relation of being in the same orbit. The special case of Proposition~\ref{proposition:schanuel-conditional-independence} corresponding to the relation of unconditional independence $x \Indep y$ is the relation of \emph{separatedness} ($\mathsf{supp}(x) \cap \mathsf{supp}(y) = \emptyset$), which is a central relation of interest in the literature on nominal sets.
The full conditional independence relation $x \Indep y \Cond z$ is  then a \emph{relative} notion of separatedness ($\mathsf{supp}(x) \cap \mathsf{supp}(y) \subseteq \mathsf{supp}(z)$), which 
first appeared in~\cite{simpsonA}.
We believe that the atomic logic of equivalence and conditional independence developed in the present paper may,
when transported to $\mathbf{Nom}$,
provide a convenient setting for reasoning  about syntax with variable binding. Let us illustrate this using  the untyped $\lambda$-calculus as an example.

There are several approaches to reasoning about syntax with variable binding. 
The first is to reason about \emph{raw terms}, in which, for example, $\lambda x.\, x$ is 
distinguished from $\lambda y.\, y$ because the variable name differs. This leads to an awkward definition of substitution $M[x :=N]$ that involves a non-canonical choice of bound-variable renaming, and  does not provide a good foundation on which to base structured reasoning principles. Some abitrariness can be avoided by imposing canonicity on bound-variables names, for example using
de Bruijn indices. However, syntactic manipulations then involve arithmetic operations on indices, which means that proofs of syntactic properties are entangled with arithmetic proofs that are an artefact of the choice of representation and have no intrinsic connection to the syntactic properties being proved. An alternative, favoured in many informal expositions of syntax, is to work with \emph{equivalence classes} of terms modulo $\alpha$-equivalence instead of raw terms. This leads to a canonical definition of substitution, but has 
two drawbacks that are particularly significant if one wishes to formalise proofs. The first drawback is that
 all term manipulations need to be proved compatible with the equivalence relation. Such proofs are often omitted  from informal expositions, but of course need to be given in a formal setting. The second drawback is  that
  one loses the structural-induction principle on terms that is derived from the inductive definition of raw terms..
 These two issues can be given a very elegant solution by defining syntax in the category 
 of nominal sets. There is a functor called \emph{name abstraction} that can be used to give a direct inductive definition of the nominal set of terms modulo $\alpha$-equivalence. This definition comes with an associated principle of structural induction for reasoning about
 terms modulo $\alpha$-equivalence, and a principle of structural recursion that allows one to define functions
 that are automatically well-defined on $\alpha$-equivalence classes.
 This approach 
 is  more fully described in the monograph~\cite{pitts}, which also contains pointers to the wider literature.
 It seems fair to say, however, that this approach does not solve all 
the practical difficulties of  reasoning about  binding operators. For example,
 the structural induction and recursion principles can be cumbersome to work with, due to their side conditions involving 
 concepts such as separatedness and freshness.

We propose here an alternative approach to reasoning about syntax with binding operators in the category of nominal sets. 
The idea is to reason directly about  raw terms rather than about $\alpha$-equivalence classes of terms, but to use 
properties of the atomic-sheaf-logic equivalence and conditional independence relations to enable definitions and reasoning to be carried out in an elegant structural way.
To illustrate the proposal, let us consider  untyped $\lambda$-terms  presented in the form $\Gamma \vdash M$, where $\Gamma$ is a finite sequence of distinct names
that are treated as free variables in term $M$. The rules for generating such terms are:
\[
\vcenter{\hbox{\includegraphics[scale=0.18]{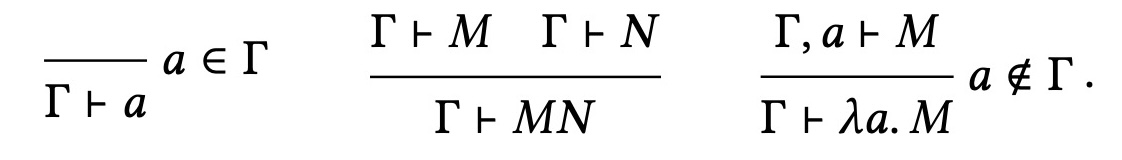}}}
%\begin{prooftree}
%\justifies 
%\Gamma \vdash a
%\using{a \in \Gamma}
%\end{prooftree}
%\qquad
%\begin{prooftree}
%\Gamma \vdash M \quad \Gamma \vdash N
%\justifies 
%\Gamma \vdash M N
%\end{prooftree}
%\qquad
%\begin{prooftree}
%\Gamma, a \vdash M 
%\justifies 
%\Gamma \vdash \lambda a.\, M
%\using{a \not\in \Gamma}
%\end{prooftree}
%\, .
\]
Then the $\Gamma$-indexed relation $\{\equiv_\Gamma \, \subseteq \mathsf{Term}_\Gamma \times \mathsf{Term}_\Gamma\}_\Gamma$ of $\alpha$-equivalence, where 
\[\mathsf{Term}_\Gamma := \{ \Gamma \vdash M \mid \text{$\Gamma \vdash M$ is a term}\}\, ,\] can be defined as the
smallest $\Gamma$-indexed congruence relation containing atomic equivalence 
$\{\sim_\Gamma \, \subseteq \mathsf{Term}_\Gamma \times \mathsf{Term}_\Gamma\}_\Gamma$ (i.e., orbit equality).\footnote{This characterisation depends on the use of terms with explicit contexts and on the restriction to contexts in which all names are distinct.} 
Substitution $\Gamma \vdash M[a := N]$ can
be specified as a function defined on any pair of terms $\Gamma, a  \vdash M$ and $\Gamma \vdash N$ for which 
the conditional independence (i.e., relative separation property) 
\[\Gamma,a \vdash \!M ~\sIndep ~\Gamma \!\vdash \! N ~\mid ~ \Gamma\, \] 
holds, by simple structural recursion on the structure of the raw term $\Gamma , a \vdash M$. 
Of course, one would like substitution to be defined on \emph{all} suitable terms, not just on sufficiently separated ones.
This is achieved, by defining substitution as a ternary \emph{relation} $\mathsf{Sub}_{\Gamma,a} \subseteq 
\mathsf{Term}_{\Gamma, a}  \times \mathsf{Term}_\Gamma  \times \mathsf{Term}_\Gamma$, by specifying
that
\[\mathsf{Sub}_{\Gamma,a} (\Gamma \!, \!a \!\vdash \!M, \, \Gamma \!\vdash \!N, \, \Gamma \!\vdash\! L) \]
holds precisely when there exists $\Gamma \vdash N'$ such that $\Gamma\!\vdash \!N' \sim \Gamma\!\vdash \!N$ and
$\Gamma\!,\!\{a\}\!\vdash\!M \sIndep \Gamma\!\vdash\!N' \mid \Gamma$ and $L = M[a := N']$. By the  independent existence principle~\eqref{indep:Z},
this relation is total in the sense that, for any $M, N$ (for brevity we omit the contexts) there exists $L$ 
such that $\mathsf{Sub}(M,N,L)$. The relation is also single-valued up to equivalence:
if  $\mathsf{Sub}(M,N,L)$ and $\mathsf{Sub}(M,N,L')$ then
it holds that $\Gamma\!\vdash\!L \sim \Gamma\!\vdash\!L'$. Preservation of $\alpha$-equivalence, then follows in the form:
if $M \equiv_{\Gamma,a} M'$ and $N \equiv_{\Gamma} N'$
and $\mathsf{Sub}(M,N,L)$ and $\mathsf{Sub}(M',N',L')$ then $L \equiv_{\Gamma} L'$, which can be established elegantly and abstractly using the 
characterisation of $\alpha$-equivalence given above.

A  high-level summary of the above outlined approach is that one reasons  with raw terms, making use of atomic sheaf logic and its  equivalence and conditional independence relations to systematically subsume the necessary renaming of bound variables
as instances of general logical principles. 
%For example, 
%in proving the totality of the substitution relation, the existence of $N'$ such that $\Gamma\!\vdash \!N' \sim \Gamma\!\vdash \!N$ and
%$\Gamma\!,\!a\!\vdash\!M \sIndep \Gamma\!\vdash\!N' \mid \Gamma$ implements a version of \emph{Barendregt's variable convention} that variables can always be renamed to avoid clashes~[[REFS]]. In our setting, this is justified by the independent existence principle~\eqref{indep:Z}.

The fact that atomic sheaf logic applies both to nominal sets (via the equivalence with $\ShAt(\Iop)$) and to 
probability (via $\ShAt(\sBP_0)$) means that one can compare the two approaches to nominal syntax, the standard one in which terms are $\alpha$-equivalence classes and the proposed one using raw terms, using an analogy with probability theory.
When terms (with explicit context) are considered as $\alpha$-equivalence classes, they are, in particular, equated up to atomic equivalence (orbit equality). In the probabilistic setting of  $\ShAt(\sBP_0)$, atomic equivalence is equality in distribution. So reasoning with $\alpha$-equivalence classes is analogous to doing probability with 
probability distributions. In contrast, our proposal to reason with raw terms and make use of the atomic 
equivalence and conditional independence relations is analogous to, in probability theory,
reasoning with random variables and exploiting the relations of equality in law and conditional independence between them.
Certainly, in mainstream probability theory, reasoning with random variables is usually considered 
 more convenient than reasoning with probability distributions. It therefore seems worth investigating  whether our  proposed approach to reasoning about syntax will have similar practical
advantages over the $\alpha$-equivalence-class-based approach. 
It is intended to carry out some case studies in this direction as future research.

\subsection{Further work}

We end the paper with two questions   for potential further investigation on the theory side, of which the second was  suggested by one of the journal referees. 
The first is to obtain a completeness theorem for the logic of equivalence and conditional independence valid in atomic toposes.
The second is to investigate whether atomic sheaf logic enjoys a similar relationship to second-order logic as that enjoyed by dependence logic~\cite{Vaananen}.

\section*{Acknowledgements}

I thank Angus Macintyre for drawing my attention to dependence logic, and Andr\'e Joyal, Paul-Andr\'e Melli\`es, Dario Stein and the anonymous reviewers of both conference and journal versions for  helpful suggestions. I also thank 
Terblanche Delport, Willem Fouch\'e and Paul and Petrus Potgieter for their hospitality in Pretoria in January 2023, where half this paper was written. Paul Taylor's diagram macros were used.

\bibliographystyle{plain}
\bibliography{EquivIndepArXiV}

%\end{document}

% RGE REMAINDER IS KEPT FOR FUTURE USE

\appendix

\section{Proof of Theorem~\ref{theorem:sur-supports}}
\label{appendix:supports}

Recall that Theorem~\ref{theorem:sur-supports} states that every sheaf 
in $\ShAt(\Sur)$ has supports. 
The main tool needed to  prove this is Theorem~\ref{theorem:push-pull} below.

\begin{theorem}
\label{theorem:push-pull}
Every atomic sheaf $\Sh{P} \in \ShAt(\Sur)$ maps pushouts in $\Sur$ to pullbacks in $\Set$.
\end{theorem}

\noindent
The proof of Theorem~\ref{theorem:push-pull} given below builds on Theorem~\ref{theorem:sheaf-nice}. 
When I discussed this work with André Joyal, he told me that he already knew  Theorem~\ref{theorem:push-pull}, and he kindly showed me  his own proof, which is somewhat different in structure from the argument given below.

Observe first that the category $\Sur$ has pushouts, and that these are defined as in $\Set$.
Observe also that, in any commuting diagram in $\Sur$ of the form below, the outer kite is a pushout if and only if the right-hand square is a pushout (because all maps in $\Sur$ are epimorphic).
\[
\vcenter{\hbox{\includegraphics[scale=0.18]{diag06.jpg}}}
%\begin{diagram}
%&&&& {\bullet} && \\
%&  && \ruTo \ruTo(4,2) && \rdTo  & \\
%{\bullet} & \rTo & {\bullet} &&&& {\bullet} &  \\
%& \rdTo(4,2) && \rdTo &&  \ruTo & \\
%&&&& {\bullet} &&
%\end{diagram}
\]

\begin{lemma}
\label{lemma:very-easy}
Suppose we have a commuting diagram as above in $\Sur$. Let  $P \in \Psh(\Sur)$ be a separated presheaf.
Then $P$ maps the right-hand square to a pullback in $\Set$ if and only if it maps the outer kite to a pullback in $\Set$.
\end{lemma}
\begin{proof} Easy. 
\end{proof}

\begin{proof}[Proof of Theorem~\ref{theorem:push-pull}]
%Given finite nonempty sets $\Omega_X$ and $\Omega_Y$, 
A relation $R \subseteq \Omega_X \times \Omega_Y$ is 
said  to be \emph{bitotal} if:
\[\forall \omega_X \in \Omega_X.\, \exists \omega_Y \in \Omega_Y,\; \omega_XR \omega_Y~\text{and}~\forall \omega_Y \in \Omega_Y.\, 
\exists \omega_X \in \Omega_X,\; xRy \enspace .\]

Let $R \subseteq \Omega_X \times \Omega_Y$ be a bitotal relation. Then the projections 
 $\Omega_X \lMap{r} R \rMap{r'} \Omega_Y$ form a span in $\Sur$. Construct the pushout
 \begin{equation}
 \label{diagram:useful}
 \vcenter{\hbox{\includegraphics[scale=0.18]{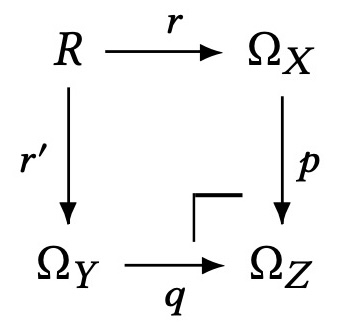}}}
% \begin{diagram}
% R & \rTo^{r} & \Omega_X \\
% \dTo^{r'} & & \dTo_p \\
% \Omega_Y & \rTo_q &\NWpbk  \Omega_Z
% \end{diagram}
\end{equation}

For any $n \geq 0$ define  $R_n  \subseteq \Omega_X \times \Omega_X$ and 
 $S_n \subseteq \Omega_X \times \Omega_X$ 
 by: $R_0 := R$ and $S_n = R^{-1} \circ R_n$ and $R_{n+1} := R \circ S_n$.
 Let $r_n \colon R_n \rMap{} \Omega_X$ and $r'_n \colon R_n \rMap{} \Omega_Y$ be the first and second projections
 and similarly for $s_n \colon S_n \rMap{} \Omega_X$ and $s'_n \colon S_n \rMap{} \Omega_X$ .
 Alternatively, we can formulate this in diagrammatic terms, taking pullbacks for both top-left squares below,
 \begin{equation*}
 \label{diagram:R-S-pullbacks}
 \vcenter{\hbox{\includegraphics[scale=0.18]{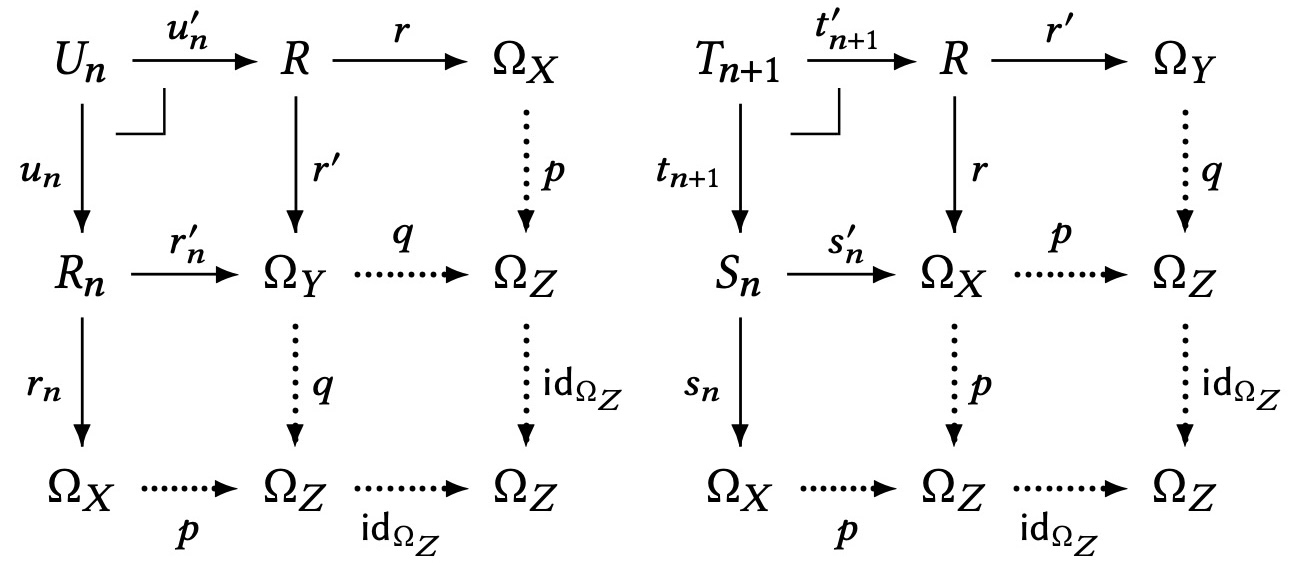}}}
% \begin{diagram}
% U_n \SEpbk & \rTo^{u_n'} & R & \rTo^r & \Omega_X & \quad & T_{n+1} \SEpbk & \rTo^{t'_{n+1}} & R & \rTo^{r'} & \Omega_Y \\
% \dTo^{u_n} & & \dTo_{r'} & & \dDotty_p & & \dTo^{t_{n+1}} & & \dTo_{r} & & \dDotty_q \\
% R_n & \rTo^{r'_n} & \Omega_Y & \rDotty^q & \Omega_Z  & & S_n & \rTo^{s'_n} & \Omega_X& \rDotty^p &  \Omega_Z \\
% \dTo^{r_n} & & \dDotty_q & & \dDotty_{\Id_{\Omega_Z}} & &  \dTo^{s_n} & &  \dDotty_p & & \dDotty_{\Id_{\Omega_Z}} \\
% \Omega_X & \rDotty_p & \Omega_Z  & \rDotty_{\Id_{\Omega_Z}}& \Omega_Z & & \Omega_X & \rDotty_p & \Omega_Z  &  \rDotty_{\Id_{\Omega_Z}}& \Omega_Z
% \end{diagram}
 \end{equation*}
 and defining the relations $(s_n,s_n') : S_n \rMon \Omega_X \times \Omega_X$ and 
 $(r_{n+1},r'_{n+1}) : R_n \rMon \Omega_X \times \Omega_X$ as the following epi-mono factorisations in $\Set$ 
 \[
 \vcenter{\hbox{\includegraphics[scale=0.18]{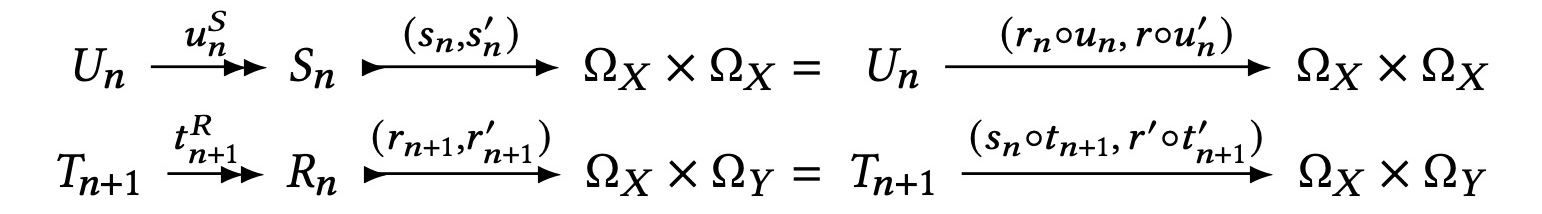}}}
% \begin{diagram}
% U_n & \rEpi^{u_n^S} & S_n & \rMono^{(s_n,s_n')} & \Omega_X \times \Omega_X & ~ ~~~{=} \!\!\!& 
% U_n & \rTo^{(r_n \circ u_n,\, r\circ u_n')} &  \Omega_X \times \Omega_X \\
% T_{n+1} & \rEpi^{t_{n+1}^R} & R_n & \rMono^{(r_{n+1},r'_{n+1})} & \Omega_X \times \Omega_Y & ~ ~~~{=} \!\!\!& 
% T_{n+1} & \rTo^{(s_n \circ t_{n+1},\, r'\circ t'_{n+1})} &  \Omega_X \times \Omega_Y 
% \end{diagram}
 \]
 
We first claim that, for any $n \geq 0$, both diagrams below commute.
% (In fact they are pushouts, but we shall not need this property.)
\begin{equation}
  \label{diagram:twin}
 \vcenter{\hbox{\includegraphics[scale=0.18]{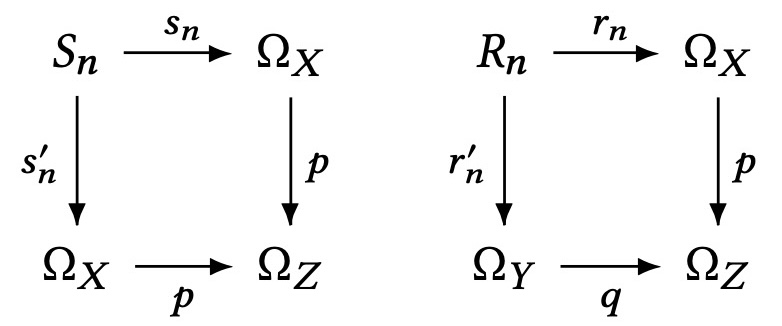}}}
% \begin{diagram}
% S_n & \rTo^{s_n} & \Omega_X  & \quad & R_n & \rTo^{r_n} & \Omega_X \\
% \dTo^{s'_n} & & \dTo_p  & & \dTo^{r'_n} & & \dTo_p \\
% \Omega_X & \rTo_p & \Omega_Z  & & \Omega_Y & \rTo_q & \Omega_Z
% \end{diagram}
 \end{equation}
This first claim is proved by a straightforward induction on $n$. For example, one can use the induction hypothesis to complete
the diagrams involving $U_n$ and $T_{n+1}$ above % \eqref{diagram:R-S-pullbacks} 
with the dotted arrows.
%
%For $R_0$, the claim is true, as it merely asserts that~\eqref{diagram:useful} commutes. Next, assuming that the diagram for $R_n$ commutes, we show that 
%the diagram for $S_n$ does too. Suppose $\omega_X S_n \omega'_X$; i.e., for some $\omega_Y$, we have
%$\omega_X R_n \omega_Y$ and $\omega'_X R \omega_Y$. Then  indeed $p(\omega_X) = q(\omega_Y) = p(\omega'_X)$, with the latter equality because the diagram for $R_n$ commutes.  Lastly, assuming the diagram for $S_n$ commutes, we show that the diagram for $R_{n+1}$ commutes. Suppose $\omega_X R_{n+1} \omega_Y$; i.e., for some $\omega'_X$, we have
%$\omega_X S_n \omega'_X$ and $\omega'_X R \omega_Y$. Then indeed $p(\omega_X) = p(\omega'_X) = q(Y)$, with the first equality because the diagram for $S_n$ commutes.

Our second claim is that, for some $n \geq 0$, the right-hand square of~\eqref{diagram:twin} is a pullback in $\Set$. (The same holds for the left-hand square, but we shall not need this.) This holds because the
 fibres of the pushout maps $p$ and $q$ from \eqref{diagram:useful}
 are the connected components in the bipartite
graph $R \subseteq \Omega_X \times \Omega_Y$ restricted to $\Omega_X$ and $\Omega_Y$
respectively, and the $R_n$ construction approximates the path relation from below, necessarily
reaching a fixed point at some finite $n$. 

Our third claim is that every atomic sheaf $\Sh{P} \in \ShAt(\Sur)$
maps the pushout diagram~\eqref{diagram:useful} to a pullback in $\Set$.
For this, let $x \in \Sh{P}(\Omega_X)$ and $y \in \Sh{P}(\Omega_Y)$ be such that 
$x \cdot r = y \cdot r'$. We prove, by induction on $n$ that
$x \cdot r_n = y \cdot r'_n$ and $x \cdot s_n = x \cdot s'_n$ for all $n$.
For $n = 0$, we have $r_0 = r$ and $r'_0 = r'$ so indeed $x \cdot r_0 = y \cdot r'_0$.
Next, assuming $x \cdot r_n = y \cdot r'_n$, we show $x \cdot s_n = x \cdot s'_n$.
For this, 
% we use $T = \{(\omega_X,\omega_Y,\omega_X') \mid \text{$\omega_X R_n \omega_Y$ and $\omega_X' R \omega_Y$}\}$, and its maps $(\pi_1,\pi_2): T \rTo R_n$ and $(\pi_3,\pi_2): T \rTo R$ and $(\pi_1,\pi_3): T \rTo S_n$.
we have $x \cdot s_n \cdot u^S_n = x \cdot r_n \cdot u_n 
= y \cdot r'_n \cdot u_n = y \cdot r' \cdot u'_n = x \cdot r \cdot  u'_n = 
x \cdot s_n' \cdot u^S_n$; whence by separatedness  $x \cdot s_n = x \cdot s_n'$.
Similarly, assuming $x \cdot s_n = x \cdot s'_n$, we show 
$x \cdot r_{n+1} = y \cdot r'_{n+1}$.
%Define $U = \{(\omega_X,\omega'_X,\omega_Y) \mid \text{$\omega_X S_n \omega_X$ and $\omega_X' R \omega_Y$}\}$, and its maps $(\pi_1,\pi_2): U \rTo S_n$ and $(\pi_2,\pi_3): U \rTo R$ and $(\pi_1,\pi_3): U \rTo R_{n+1}$.
For this, we have $x \cdot r_{n+1} \cdot t^R_{n+1} = x \cdot s_n \cdot t_{n+1} = 
x \cdot s'_n \cdot t_{n+1} = x \cdot r \cdot t'_{n+1}= y \cdot r' \cdot t'_{n+1}=
y \cdot r'_{n+1} \cdot t^R_{n+1}$; whence by separatedness $x \cdot r_{n+1} = y \cdot r'_{n+1}$.
This completes the argument by induction. 
The second claim above now gives us $n$ such that the right-hand square of~\eqref{diagram:twin} is a pullback,
 in $\Set$, hence an 
independent square in $\Sur$.
By Theorem~\ref{theorem:sheaf-nice}, the square is
mapped by $\Sh{P}$ to a pullback in $\Set$. By the statement proved by induction,
$x \cdot r_n = y \cdot r'_n$. So, by the pullback property in $\Set$, 
there exists a unique $z \in \Sh{P}(\Omega_Z)$
such that $z \circ p = x$ and $z \circ q = y$, which is what we needed to show to establish the third claim.

We now establish the property asserted by the theorem.  Consider any pushout diagram in $\Sur$.
\[
\vcenter{\hbox{\includegraphics[scale=0.18]{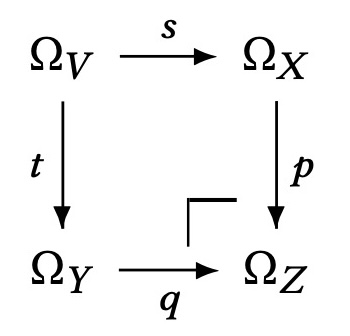}}}
%\begin{diagram}
% \Omega_V & \rTo^{s} & \Omega_X \\
% \dTo^{t} & & \dTo_p \\
% \Omega_Y & \rTo_q &\NWpbk  \Omega_Z
% \end{diagram}
 \]
 Define $R \subseteq \Omega_X \times \Omega_Y$ to be the 
 image $(s,t)(\Omega_V)$. Since $s,t$ are surjective, $R$ is bitotal. 
 By the observations at the start of this section,~\eqref{diagram:useful} is also a pushout.
 By the third claim above, $\Sh{P}$ maps \eqref{diagram:useful} to a pullback in $\Set$.
 This property transfers to the original pushout, by Lemma~\ref{lemma:very-easy}.
\end{proof}

\begin{proof}[Proof of Theorem~\ref{theorem:sur-supports}]
Given a sheaf $\Sh{P} \in \ShAt(\Sur)$ and element $x \in \Sh{P}(\Omega_X)$, the support
is obtained by taking a joint pushout in $\Sur$ of all (inequivalent) representable factorisations 
of $x$, of which there are only finitely many (because there are only  finitely many 
partitions of $\Omega_X$). By Theorem~\ref{theorem:push-pull}, this joint pushout is itself
a representable factorisation of $x$.
\end{proof}

\section{Validity of  axioms~\eqref{indep:C} and  \eqref{indep:D} from Figure~\ref{figure:independence-axioms}}
\label{appendix:CD}

The lemma below establishes the validity of axiom~\eqref{indep:C}. 
\begin{lemma}
Suppose $(x, (y,z), w) \in {\Indep}_{\Sh{A},\,\Sh{B}\times \Sh{C} \Cond\Sh{D}}(X)$ 
then $(x, y, (z, w)) \in {\Indep}_{\Sh{A},\Sh{B} \Cond \Sh{C}\times \Sh{D}}(X)$.
\end{lemma}
\begin{proof}
If $(x, (y,z), w) \in {\Indep}_{\Sh{A},\,\Sh{B}\times \Sh{C} \Cond\Sh{D}}(X)$
then we have a hybrid diagram
\[
\vcenter{\hbox{\includegraphics[scale=0.18]{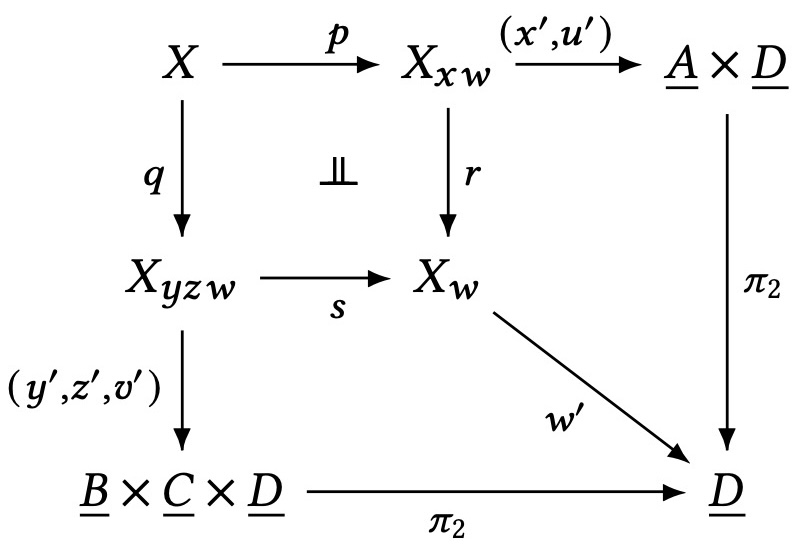}}}
%\begin{diagram}
%X & \rTo^{p} & X_{xw}  & \rTo^{\!\!\!\!\!{(x',u')}\;\;} & \Sh{A}  \times \Sh{D}\\
%\dTo^{q} & \Indep &  \dTo_r &  &  \\
%X_{yzw} & \rTo_s & X_w & & \dTo_{\pi_2} \\
%\dTo^{(y',z',v')} & & & \rdTo_{w'} &  \\
%\Sh{B} \times \Sh{C} \times \Sh{D} & & \rTo_{\pi_2} & & \Sh{D}
%\end{diagram}
\]
% $(x, y, (z, w)) \in {\Indep}_{\Sh{A},\Sh{B} \Cond \Sh(C) \times \Sh{D}}(X)$ 
where $x' \cdot p = s$ and $y' \cdot q = y$ and $z' \cdot q = z$ and
$(X_w,\, r \circ p,\, w')$ is support for  $w$ and, without loss of generality,
$(X_{xw}, p, (x',u'))$ is support for $(x,z)$ and 
$(X_{yzw},q,(y',z',v'))$ is support for $(y,z,w)$. 

The independent square in the diagram  above
can be factorised as a composite of two commuting squares as in the top row below
\begin{equation}
\label{equation:ugly}
\vcenter{\hbox{\includegraphics[scale=0.18]{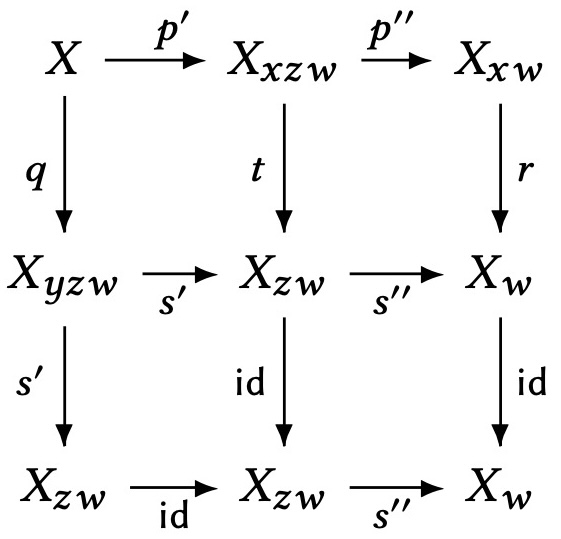}}}
%\begin{diagram}
%X & \rTo^{p'}  & X_{xzw} & \rTo^{p''} & X_{xw}  \\
%\dTo^{q} &  & \dTo^t &  &  \dTo_r   \\
%X_{yzw} & \rTo_{s'} & X_{zw} & \rTo_{s''} & X_w  \\
%\dTo^{s'} & & \dTo^{\Id} & & \dTo_{\Id} \\
%X_{zw} & \rTo_{\Id} & X_{zw} & \rTo_{s''} & X_w
%\end{diagram}
\end{equation}
where all objects are defined as the supports indicated by their names. For example,
$(X_{zw}, \, s' \circ q, \,(z'',w''))$ is support for $(z,w)$ and 
$(X_{xzw},\, p', (x' \cdot p'',\, z'' \cdot t,\, u' \cdot p''))$ is support for $(x,z,w)$.
We show that the top-right square is an independent pullback.

To see it is independent, observe that the full composite square \eqref{equation:ugly} above is a composite of an independent top-row rectangle with the  two independent squares in the bottom row. So  \eqref{equation:ugly} is independent. That is, the
square
\[
\vcenter{\hbox{\includegraphics[scale=0.18]{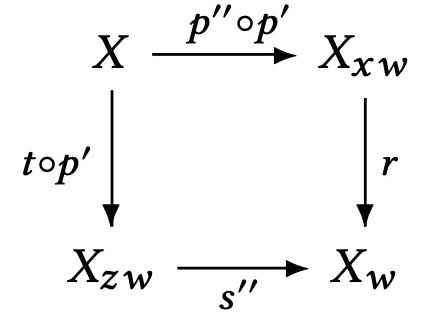}}}
%\begin{diagram}
% X & \rTo^{p'' \circ p'} & X_{xw}  \\
% \dTo^{t\circ p'} &  &  \dTo_r   \\
% X_{zw} & \rTo_{s''} & X_w 
%\end{diagram}
\]
is independent. It thus follows from the descent property that the top-right square in 
\eqref{equation:ugly} is independent.

For the independent pullback property, consider any independent pullback of $r$ along $s''$
\[
\vcenter{\hbox{\includegraphics[scale=0.18]{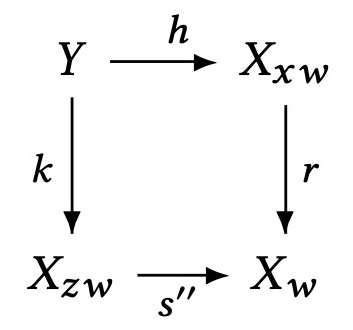}}}
%\begin{diagram}
% Y & \rTo^h & X_{xw}  \\
% \dTo^k &  &  \dTo_r   \\
% X_{zw} & \rTo_{s''} & X_w 
%\end{diagram}
\]
Since the top-right square of \eqref{equation:ugly} is independent, there exists $j : X_{xzw} \to Y$
such that $k \circ j = t$ and $j \circ h = p''$. This gives us a representable factorisation
$(Y, \, j \circ p', \, (x' \cdot h, \,z'' \cdot k, \,u' \cdot h)) $
of $(x,z,w)$. Since $(X_{xzw},\, p', \,(x' \cdot p'',\, z'' \cdot t, \, u' \cdot p''))$ is support for $(x,z,w)$,
we obtain a map $i : Y \to X_{xzw}$  of representable factorisations. However $j$ is also a map of representable factorisations in the opposite direction, so $i$ and $j$ are mutual inverses. Thus the top-right square in 
\eqref{equation:ugly} is indeed an independent pullback.

Since the top-row rectangle of \eqref{equation:ugly} is independent and the top-right square an independent pullback it follows that the top-left square is independent. Using this, we form the hybrid diagram
\[
\vcenter{\hbox{\includegraphics[scale=0.18]{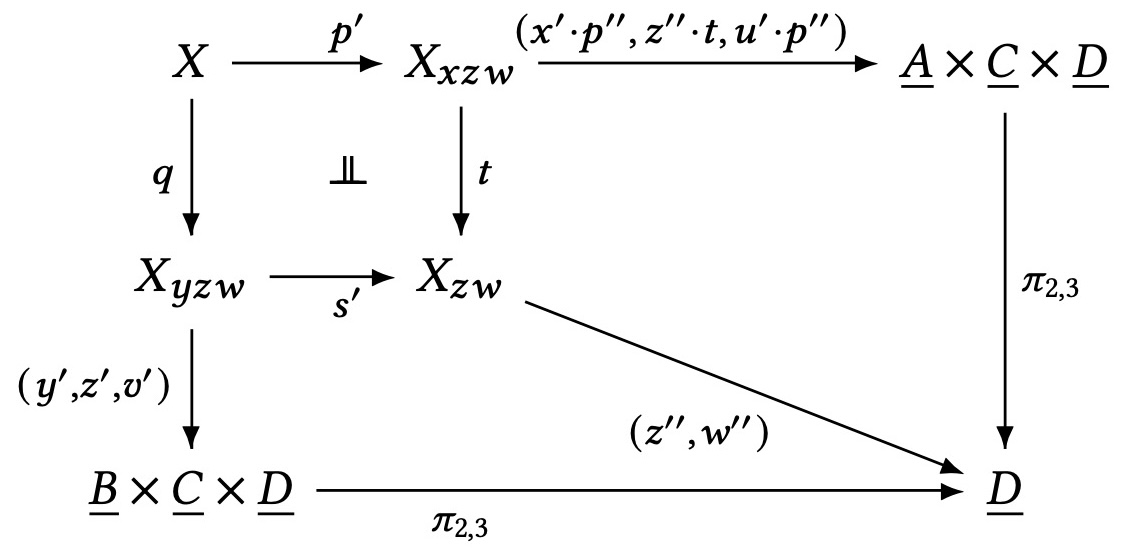}}}
%\begin{diagram}
%X & \rTo^{p'} & X_{xzw}  & \rTo^{\!\!\!\!\!{(x'\cdot p'',\, z''\cdot t,\,u'\cdot p'')}\;\;} & \Sh{A}  \times  \Sh{C} \times \Sh{D}\\
%\dTo^{q} & \Indep &  \dTo_t &  &  \\
%X_{yzw} & \rTo_{s'} & X_{zw} & & \dTo_{\pi_{2,3}} \\
%\dTo^{(y',z',v')} & & & \rdTo_{(z'',w'')} &  \\
%\Sh{B} \times \Sh{C} \times \Sh{D} & & \rTo_{\pi_{2,3}} & & \Sh{D}
%\end{diagram}
\]
showing that indeed $(x, y, (z, w)) \in {\Indep}_{\Sh{A},\Sh{B} \Cond \Sh{C} \times \Sh{D}}(X)$.
\end{proof}

The lemma below establishes the validity of axiom \eqref{indep:D}.
\begin{lemma}
Suppose $(x, y, (z, w)) \in {\Indep}_{\Sh{A},\Sh{B} \Cond \Sh{C} \times \Sh{D}}(X)$
and $(x, z, w) \in {\Indep}_{\Sh{A}, \Sh{C} \Cond \Sh{D}}(X)$
then 
$(x, (y, z), w) \in {\Indep}_{\Sh{A},\,\Sh{B}\times \Sh{C} \Cond\Sh{D}}(X)$.
\end{lemma}
\begin{proof}
The assumption $(x, y, (z, w)) \in {\Indep}_{\Sh{A},\Sh{B} \Cond \Sh{C} \times \Sh{D}}(X)$
gives us:
\begin{equation}
\label{equation:blib}
\vcenter{\hbox{\includegraphics[scale=0.18]{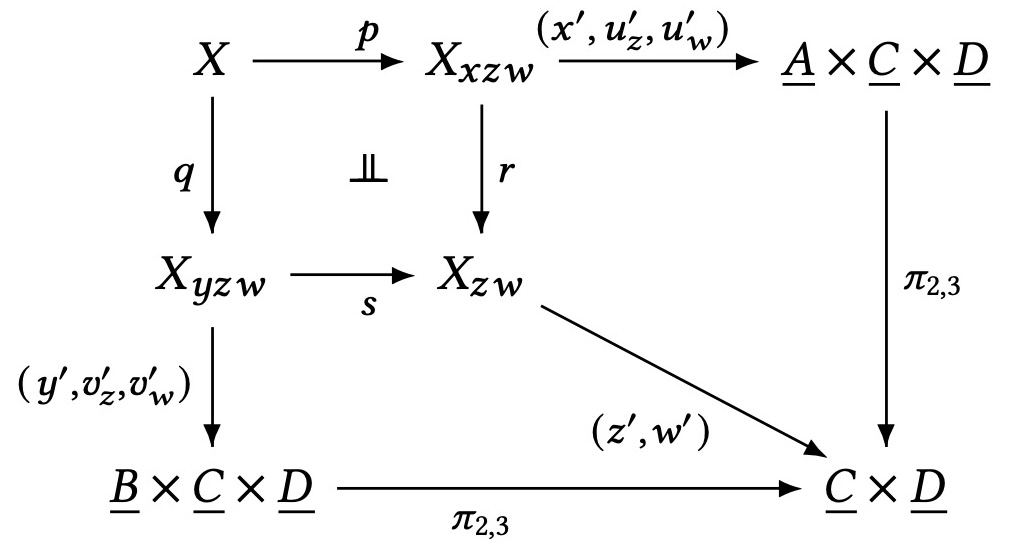}}}
%\begin{diagram}
%X & \rTo^{p} & X_{xzw}  & \rTo^{\!\!\!\!\!{(x',\,u'_z ,\,u'_w)}\;\;} & \Sh{A}  \times  \Sh{C} \times \Sh{D}\\
%\dTo^{q} & \Indep &  \dTo_r &  &  \\
%X_{yzw} & \rTo_{s} & X_{zw} & & \dTo_{\pi_{2,3}} \\
%\dTo^{(y',v'_z,v'_w)} & & & \rdTo_{(z',w')} &  \\
%\Sh{B} \times \Sh{C} \times \Sh{D} & & \rTo_{\pi_{2,3}} & & \Sh{C} \times \Sh{D}
%\end{diagram}
\end{equation}
where $x' \cdot p = x$ and $y' \cdot q = y$ and 
$(X_{zw}, \, r \circ p,\, (z',w'))$  is support for $(z,w)$ and, without loss of generality,
$(X_{xzw}, p, (x',u'_z,u'_w))$ is support for $(x,z,w)$ and
$(X_{yzw}, q, (y',v'_z,v'_w))$ is support for $(y,z,w)$.

Similarly, the assumption $(x, z, w) \in {\Indep}_{\Sh{A}, \Sh{C} \Cond \Sh{D}}(X)$ gives us:
\begin{equation}
\label{equation:bliba}
\vcenter{\hbox{\includegraphics[scale=0.18]{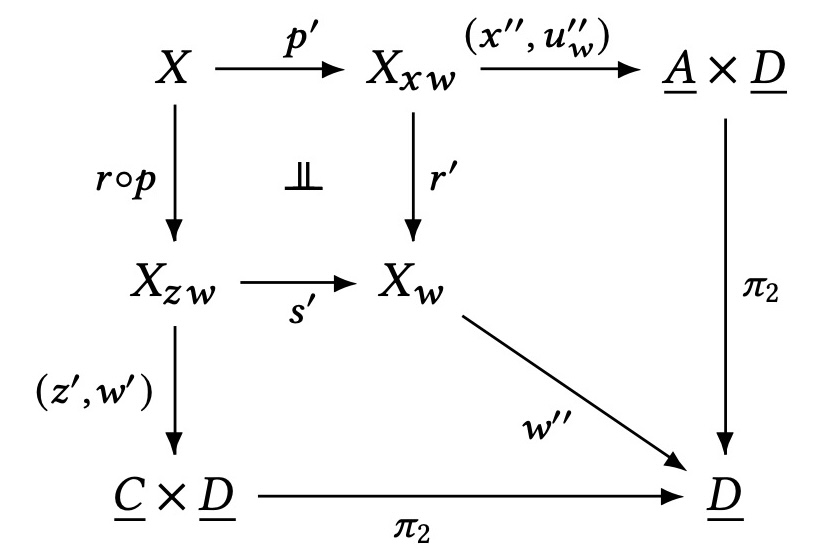}}}
%\begin{diagram}
%X & \rTo^{p'} & X_{xw}  & \rTo^{\!\!\!\!\!{(x'',\, u''_w)}\;\;} & \Sh{A}  \times   \Sh{D}\\
%\dTo^{r \circ p} & \Indep &  \dTo_{r'} &  &  \\
%X_{zw} & \rTo_{s'} & X_{w} & & \dTo_{\pi_{2}} \\
%\dTo^{(z',w')} & & & \rdTo_{w''} &  \\
%\Sh{C} \times \Sh{D} & & \rTo_{\pi_{2}} & & \Sh{D}
%\end{diagram}
\end{equation}
where $x'' \cdot p' = x$ and $z'' \cdot q' = z$ and 
$(X_{w},\, r' \circ p', \, w'')$ is support for $w$ and, without loss of generality,
$(X_{xw}, \, p',\, (x'',\, u''_w))$ is support for $(x,w)$ and we can use
$r\circ p $ because 
$(X_{zw}, \, r \circ p,\, (z',w'))$ is support for $(z,w)$. %, by Lemma~\ref{lemma:weak-supports}.

Exploiting the support property of $X_{xw}$, we obtain $p''$ in
\[
\vcenter{\hbox{\includegraphics[scale=0.18]{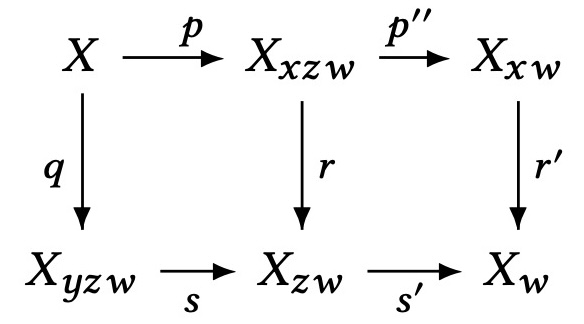}}}
%\begin{diagram}
%X & \rTo^{p} & X_{xzw}  &  \rTo^{p''} & X_{xw}  \\ 
%\dTo^{q} &  &  \dTo_r &  &  \dTo_{r'} \\
%X_{yzw} & \rTo_{s} & X_{zw} & \rTo_{s'} & X_w
%\end{diagram}
\]
such that $p'' \circ p = p'$.
The left-hand square above is the independent square from~\eqref{equation:blib}.
Since $p'' \circ p = p'$, the right-hand square is also independent, by descent along $p$ from the independent square in~\eqref{equation:bliba}. So the composite rectangle is independent.

The composite rectangle provides the independent square in
\[
\vcenter{\hbox{\includegraphics[scale=0.18]{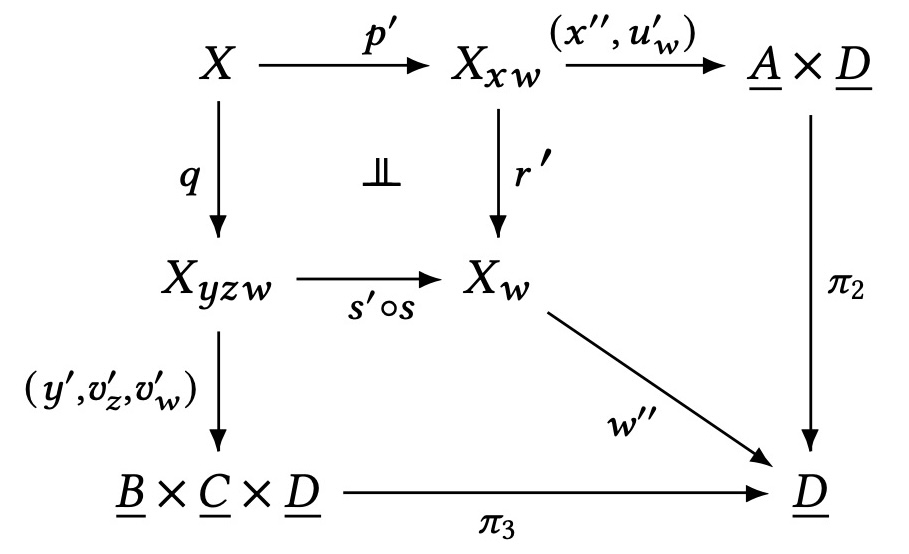}}}
%\begin{diagram}
%X & \rTo^{p'} & X_{xw}  & \rTo^{\!\!\!\!\!{(x'',\,u'_w)}\;\;} & \Sh{A}  \times  \Sh{D}\\
%\dTo^{q} & \Indep &  \dTo_r' &  &  \\
%X_{yzw} & \rTo_{s' \circ s} & X_{w} & & \dTo_{\pi_{2}} \\
%\dTo^{(y',v'_z,v'_w)} & & & \rdTo_{w''} &  \\
%\Sh{B} \times \Sh{C} \times \Sh{D} & & \rTo_{\pi_{3}} & & \Sh{D}
%\end{diagram}
\]
showing that $(x, (y, z), w) \in {\Indep}_{\Sh{A},\,\Sh{B}\times \Sh{C} \Cond\Sh{D}}(X)$ as required.
\end{proof}

\section{Proof of Proposition~\ref{proposition:sBP-indep-pullbacks}}
\label{appendix:sbp}

The goal of the section is to prove Proposition~\ref{proposition:sBP-indep-pullbacks}, which states that
Definition~\ref{definition:indep-square-sbp} endows $\sBP_0$ with independent pullback structure satisfying the descent property. 

Recall that Definition~\ref{definition:indep-square-sbp} defines a commuting square \eqref{equation:square-sbp} in $\sBP_0$  to be \emph{independent} if $p \Indep q \Cond r \circ p$ according to Definition~\ref{definition:conditional-independence-RV}. Since the square is commuting, the property in Definition~\ref{definition:conditional-independence-RV} simplifies to: for every $S \in \mathcal{B}_{\Omega_Y}$ and $T \in \mathcal{B}_{\Omega_Z}$, and for
$P_{\Omega_W}$-almost all $\omega \in \Omega_W$, 
\begin{equation}
\label{sbps-ind-cond}
P_{(r\circ p)^{-1}(\omega)} (p^{-1}(S)  \cap q^{-1}(T))  ~=  ~
P_{r^{-1}(\omega)}(S) \cdot P_{s^{-1}(\omega)}(T) \enspace . 
\end{equation}

The key proposition below characterises the independence of \eqref{equation:square-sbp}
as being equivalent to $p$, considered %(up to almost sure equality) 
as a 
map on fibre sets $q^{-1}(\omega_Z) \to r^{-1}(s(\omega_Z))$, 
preserving the disintegration-induced probability measures, for almost all $\omega_Z$.
% Due to property (D1) of disintegrations, the fibre sets  do not need to be mentioned explicitly. 
\begin{proposition}
\label{proposition:fibre-preserving}
A commuting square in $\sBP_0$~\eqref{equation:square-sbp} is independent if and only if,
for $P_{\Omega_Z}$-almost-all $\omega_Z \in \Omega_Z$, it holds that
$p_*(P_{q^{-1}(\omega_Z)})  = P_{r^{-1}(s(\omega_Z))}\,$.
\end{proposition}
\begin{proof}
We first prove the right-to-left implication. Accordingly, suppose $p_*(P_{q^{-1}(\omega_Z)})  = P_{r^{-1}(s(\omega_Z))}$
holds for $P_{\Omega_Z}$-almost-all $\omega_Z \in \Omega_Z$.
For $P_{\Omega_W}$-almost every $\omega \in \Omega_W$, we prove~\eqref{sbps-ind-cond} by
\begin{align*}
& P_{(r\circ p)^{-1}(\omega)} (p^{-1}(S)  \cap q^{-1}(T))  ~ 
 \\
 & \qquad
 = ~ \int P_{q^{-1}(\omega_Z)}(p^{-1}(S)  \cap q^{-1}(T)) ~ \mathrm{d}P_{s^{-1}(\omega)}(\omega_Z)
 \\
 & \qquad = ~ \int \One_T(\omega_Z) \cdot P_{q^{-1}(\omega_Z)}(p^{-1}(S)) ~ \mathrm{d}P_{s^{-1}(\omega)}(\omega_Z)
 \\
 & \qquad = ~ \int \One_T(\omega_Z) \cdot P_{r^{-1}(s(\omega_Z))}(S)~ \mathrm{d}P_{s^{-1}(\omega)}(\omega_Z)
 & & \text{because $p_*(P_{q^{-1}(\omega_Z)})  = P_{r^{-1}(s(\omega_Z))}$} 
 \\
 & \qquad = ~ \int \One_T(\omega_Z) \cdot P_{r^{-1}(\omega)}(S)~ \mathrm{d}P_{s^{-1}(\omega)}(\omega_Z)
 \\
 & \qquad = ~ P_{r^{-1}(\omega)}(S) \cdot  \int \One_T(\omega_Z) ~ \mathrm{d}P_{s^{-1}(\omega)}(\omega_Z)
 \\
 & \qquad = P_{r^{-1}(\omega)}(S) \cdot P_{s^{-1}(\omega)}(T) \enspace .
\end{align*}

For the left-to-right implication, suppose \eqref{sbps-ind-cond} holds, for $P_{\Omega_W}$-almost every $\omega \in \Omega_W$. % For arbitrary $S \in \mathcal{B}_{\Omega_Y}$ and $T \in \mathcal{B}_{\Omega_Z}$, 
Note that, for any $S \in \mathcal{B}_{\Omega_Y}$ the function 
\[
T ~ \mapsto ~ \int \One_T(\omega_Z) \cdot P_{q^{-1}(\omega_Z)}(p^{-1}(S))~ \mathrm{d}P_{\Omega_Z}(\omega_Z)
\]
is a measure $\mathcal{B}_{\Omega_Z} \to [0,1]$ with density
\[
\omega_Z ~ \mapsto ~ P_{q^{-1}(\omega_Z)}(p^{-1}(S)) \enspace.
\]
Similarly, the function 
\[
T ~ \mapsto ~ \int \One_T(\omega_Z) \cdot P_{r^{-1}(s(\omega_Z))}(S)~ \mathrm{d}P_{\Omega_Z}(\omega_Z)
\]
is a measure with density
\[
\omega_Z ~ \mapsto ~ P_{r^{-1}(s(\omega_Z))}(S) \enspace.
\]
Below we prove
\begin{equation}
\label{eqn:t-measures}
\int \One_T(\omega_Z) \cdot P_{q^{-1}(\omega_Z)}(p^{-1}(S))~ \mathrm{d}P_{\Omega_Z}(\omega_Z)
~ = ~ \int \One_T(\omega_Z) \cdot P_{r^{-1}(s(\omega_Z))}(S)~ \mathrm{d}P_{\Omega_Z}(\omega_Z) \enspace ,
\end{equation}
which establishes that the two measures are equal, and hence their densities are almost surely equal. 
That is, for $P_{\Omega_Z}$-almost-all $\omega_Z \in \Omega_Z$, we have
$P_{q^{-1}(\omega_Z)}(p^{-1}(S)) = P_{r^{-1}(s(\omega_Z))}(S)$, for all $S \in \mathcal{B}_{\Omega_Y}$.
That is, $p_*(P_{q^{-1}(\omega_Z)}) = P_{r^{-1}(s(\omega_Z))}$, as required.

It remains to prove~\eqref{eqn:t-measures}. For this, we calculate
\begin{align*}
& \int \One_T(\omega_Z) \cdot P_{r^{-1}(s(\omega_Z))}(S)~ \mathrm{d}P_{\Omega_Z}(\omega_Z)
\\
& \qquad = ~ \int \int  \One_T(\omega_Z) \cdot P_{r^{-1}(s(\omega_Z))}(S)~ \mathrm{d}P_{s^{-1}(\omega)}(\omega_Z) ~ \mathrm{d}P_{\Omega}(\omega)
\\
& \qquad = ~ \int \int  \One_T(\omega_Z) \cdot P_{r^{-1}(\omega)}(S)~ \mathrm{d}P_{s^{-1}(\omega)}(\omega_Z) ~ \mathrm{d}P_{\Omega}(\omega)
\\
& \qquad = ~ \int  P_{r^{-1}(\omega)}(S)  \cdot \left(\int  \One_T(\omega_Z) ~ \mathrm{d}P_{s^{-1}(\omega)}(\omega_Z) \right)   ~ \mathrm{d}P_{\Omega}(\omega)
\\
& \qquad = ~ \int  P_{r^{-1}(\omega)}(S)  \cdot P_{s^{-1}(\omega)}(T) ~ \mathrm{d}P_{\Omega}(\omega)
\\
& \qquad = ~  \int P_{(s\circ q)^{-1}(\omega)} (p^{-1}(S)  \cap q^{-1}(T)) ~ \mathrm{d}P_{\Omega}(\omega)
& & \text{by \eqref{sbps-ind-cond}}
\\
& \qquad = ~  \int \int P_{q^{-1}(\omega_Z)} (p^{-1}(S)  \cap q^{-1}(T)) ~ \mathrm{d}P_{s^{-1}(\omega)}(\omega_Z) ~ \mathrm{d}P_{\Omega}(\omega)
\\
& \qquad = ~  \int \int \One_T(\omega_Z) \cdot P_{q^{-1}(\omega_Z)} (p^{-1}(S)) ~ \mathrm{d}P_{s^{-1}(\omega)}(\omega_Z) ~ \mathrm{d}P_{\Omega}(\omega)
\\
& \qquad = ~  \int \One_T(\omega_Z) \cdot P_{q^{-1}(\omega_Z)} (p^{-1}(S)) ~ \mathrm{d}P_{\Omega_Z}(\omega_Z) \enspace .
\end{align*}
For any fixed $S \in \mathcal{B}_{\Omega_Y}$ the function mapping any $T$ to the left-hand side 
of~\eqref{eqn:t-measures} is clearly a measure $\mathcal{B}_{\Omega_Z} \to [0,1]$.
\end{proof}

We now verify that independent squares in $\sBP_0$ indeed satisfy the axioms for independent pullback structure.
Axioms (IP1) and (IP2) are straightforward. Axiom (IP3) is an easy consequence of Proposition~\ref{proposition:fibre-preserving}.
For Axiom (IP5), it is not difficult to verify that~\eqref{equation:pullback-measure}
indeed constructs an independent pullback square. The descent property is also straightforward. This leaves us with (IP4), which is established in greater generality by the proposition below. 

\begin{proposition}
In a commuting diagram in $\sBP_0$ as below, if both the composite rectangle (AB) and  right-hand square (B) are independent
and $[q],[t]$ are also jointly monic, then the left-hand square (A) is independent.
\[
\vcenter{\hbox{\includegraphics[scale=0.18]{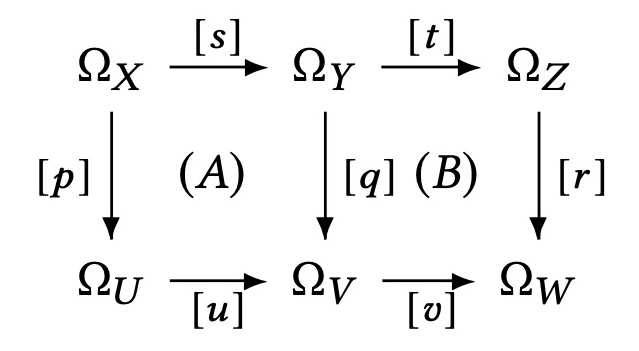}}}
%\begin{diagram}
% \Omega_X & \rTo^{[s]} & \Omega_Y  &  \rTo^{[t]} &  \Omega_Z \\
%\dTo^{[p]} & {(A)~} & \dTo_{[q]} & {\;\;\;(B)~} & \dTo_{[r]} \\
%\Omega_U & \rTo_{[u]}  & \Omega_V  & \rTo_{[v]} & \Omega_W
%\end{diagram}
\]
\end{proposition}
\begin{proof}
We use Proposition~\ref{proposition:fibre-preserving} to prove that (A) is independent. That is, we show that, for 
$P_{\Omega_U}$-almost-all $\omega_U \in \Omega_U$, and for all $C \in \mathcal{B}_{\Omega_Y}$,
\begin{equation}
\label{eqn:last}
P_{p^{-1}(\omega_U)}(s^{-1}(C)) ~ = ~ P_{q^{-1}(u(\omega_U))}(C) \enspace .
\end{equation}
We show this first for $C$ of the form $t^{-1}(A) \cap q^{-1}(B)$, where $A \in \mathcal{B}_{\Omega_Z}$ and
$B \in \mathcal{B}_{\Omega_V}$. In this case, we have
\begin{align*}
& P_{p^{-1}(\omega_U)}(s^{-1}t^{-1}(A) \cap s^{-1}q^{-1}(B)) \\
& \qquad = ~ P_{p^{-1}(\omega_U)}(s^{-1}t^{-1}(A) \cap p^{-1}u^{-1}(B)) \\
& \qquad = ~ \One_{u^{-1}(B)} (\omega_U) \cdot P_{p^{-1}(\omega_U)}(s^{-1}t^{-1}(A)) \\
& \qquad = ~ \One_{B} (u(\omega_U)) \cdot P_{p^{-1}(\omega_U)}(s^{-1}t^{-1}(A)) \\
& \qquad = ~ \One_{B} (u(\omega_U)) \cdot P_{r^{-1}(v(u(\omega_U)))}(A) 
 & & \text{by Proposition~\ref{proposition:fibre-preserving} for (AB)} \\
 & \qquad = ~ \One_{B} (u(\omega_U)) \cdot P_{q^{-1}(u(\omega_U))}(t^{-1}(A)) 
 & & \text{by Proposition~\ref{proposition:fibre-preserving} for (B)} \\
 & \qquad = ~ P_{q^{-1}(u(\omega_U))}(t^{-1}(A) \cap q^{-1}(B)) \enspace .
\end{align*}
The joint monicity of $[q]$ and $[t]$ means that the there is a measure $1$ set $S \in \mathcal{B}_{\Omega_Y}$
such that the paired function $(t,q) : S \to \Omega_Z \times \Omega_V$ is injective. 
Since $S \subseteq 
\Omega_Y$ is Borel, the standard Borel structure on $\Omega_Y$ restricts to $S$, and $(t,q)$ is a measurable embedding 
of the standard Borel space $S$ into the product standard Borel space $\Omega_Z \times \Omega_V$. Thus
every Borel subset of $S$ is the restriction of a Borel subset of $\Omega_Z \times \Omega_V$. Since the $\sigma$-algebra of Borel subsets of $\Omega_Z \times \Omega_V$ is generated by Borel rectangles $A \times B$, it follows that
the Borel subsets of $S$ are generated by sets of the form $S \cap (t^{-1}(A) \cap q^{-1}(B))$. Moreover, such sets are closed under finite intersections. 

The left-hand and right-hand sides of~\eqref{eqn:last} define measures $C \mapsto P_{p^{-1}(\omega_U)}(s^{-1}(C))$ and
$P_{q^{-1}(u(\omega_U))}$ respectively. By the equality we have shown for $C$ of the form $t^{-1}(A) \cap q^{-1}(B)$, these measures agree on a generating set for $\mathcal{B}_{\Omega_Y}$ (restricted to $S$) that is closed under finite intersections. The two measures are therefore equal. This proves~\eqref{eqn:last}.
\end{proof}

\end{document}